\documentclass{article}
\usepackage[left=1in,right=1in,top=1in,bottom=1in]{geometry}
\usepackage{amsfonts,amsmath,mathrsfs,bbm,amsthm, amssymb, enumitem,comment,apptools}
\usepackage[dvipsnames]{xcolor}
\usepackage{authblk}
\usepackage{appendix}
\usepackage{float}
\usepackage{enumitem}
\usepackage{graphicx}
\usepackage[hidelinks]{hyperref}

\hypersetup{
    colorlinks=true,
    urlcolor={CadetBlue},
    citecolor={Periwinkle},
    linkcolor={RoyalBlue!70!gray}
} 
\usepackage{url}
\usepackage{tikz-cd}
\usetikzlibrary{decorations.pathreplacing}

\usepackage{bbm}
\usepackage{pdfpages}

\usepackage{stmaryrd}

\usepackage{orcidlink}

\newcommand{\ESP}{{\hyperlink{ESP}{(ESP)}}}
\newcommand{\ESPa}{\ESP\hspace{0.85mm}}

\newcommand{\PureGen}{{\hyperlink{PureGen}{(EMPS)}}}
\newcommand{\PureGena}{\PureGen\hspace{0.85mm}}

\usepackage{mathrsfs}

\newcommand{\mcL}{\mathcal{L}}
\newcommand{\mcG}{\mathcal{G}}
\DeclareMathOperator{\dom}{\mathrm{dom}}

\DeclareMathOperator{\Tr}{\mathrm{Tr}}

\newcommand{\Hilb}[1]{\mathcal{H}^{#1}}
\newcommand{\one}{\mathbbm{1}}

\newcommand{\acts}[1]{\overset{#1}{\curvearrowright}}

\newcommand{\C}{\mathbb C}
\newcommand{\Z}{\mathbb Z}
\newcommand{\I}{\mathbb I}
\newcommand{\A}[1]{\mathcal{A}_{#1}}

\newcommand{\M}{\mathbb{M}}
\newcommand{\N}{\mathbb N}

\newcommand{\hatrho}{\hat{\rho}}

\newcommand{\localOne}{\one_{\localDim}}

\newcommand{\bondOne}{\one_{\BondDim}}

\newcommand{\gmatrices}{\matrices^g}

\newcommand{\wielength}[1]{\ell\!\left(#1\right)}

\newcommand{\prelength}[1]{\ell_0\!\left(#1\right)}

\newcommand{\bs}[1]{\boldsymbol{#1}}

\newcommand{\tr}[1]{\Tr\!\left( #1\right)}

\renewcommand{\P}{\mathbb{P}}

\newcommand{\<}{\langle}
\renewcommand{\>}{\rangle}
\newcommand{\ket}[1]{\left|#1\right\rangle}
\newcommand{\bra}[1]{\left\langle #1 \right|}

\newcommand{\inner}[2]{\left\<#1, #2\right\>}

\newcommand{\ketbra}[2]{| #1 \rangle\!\langle #2 |}
\newcommand{\braket}[2]{\langle #1 | #2 \rangle}

\let\emptyset\varnothing
\DeclareMathOperator{\loc}{loc}

\newcommand{\positives}[1]{#1_{+}}
\newcommand{\selfadjoints}[1]{#1_{\operatorname{sa}}}
\newcommand{\cstates}[1]{\mathcal{S}\!\left(#1\right)}

\newcommand{\matrices}{\mathbb{M}_{d}}
\newcommand{\states}{\mathbb{S}_{d}}
\newcommand{\stateso}{\mathbb{S}_{d}^{\circ}}

\newcommand{\rmatrices}{\matrices(\Omega)}
\newcommand{\rstates}{\states(\Omega)}
\newcommand{\rstateso}{\stateso(\Omega)}
\newcommand{\seq}[1]{\left(#1\right)}
\newcommand{\set}[1]{\left\{#1\right\}}

\newcommand{\dual}[1]{#1^\star}

\newcommand{\ucps}[1]{\operatorname{ucp}\!\left(#1\right)}

\newcommand{\LocalDim}{n}
\newcommand{\localDim}{\LocalDim}
\newcommand{\BondDim}{D}
\newcommand{\MatricesLocal}{\mathbb{M}_\LocalDim}
\newcommand{\HilbertLocal}{\C^{\LocalDim}}

\newcommand{\HilbertBond}{\C^\BondDim}
\newcommand{\MatricesBond}{\mathbb{M}_\BondDim}

\newcommand{\TrD}[1]{\operatorname{Tr}_D\!\left(#1\right)}

\newcommand{\trD}[1]{\operatorname{Tr}_D\!\left(#1\right)}

\DeclareMathOperator{\proj}{proj}
\newtheorem{thm}{Theorem}[section]
\newtheorem{thmx}{Theorem}

\newtheorem{prop}[thm]{Proposition}

\newtheorem{lem}[thm]{Lemma}
\newtheorem{lemma}[thm]{Lemma}
\newtheorem{cor}[thm]{Corollary}
\newtheorem{corollary}[thm]{Corollary}
\newtheorem*{thm*}{Theorem}
\theoremstyle{definition}

\newtheorem{definition}[thm]{Definition}

\newtheorem{example}[thm]{Example}

\newtheorem{problem}[thm]{Problem}

\newtheorem{rmk}[thm]{Remark}
\newtheorem{remark}[thm]{Remark}

\newtheorem{notation}[thm]{Notation}

\newcommand{\TCV}{\mathfrak{C}}

\numberwithin{equation}{section}
\numberwithin{thm}{section}

\let\emptyset\varnothing

\begin{document}
\title{
\huge\textbf{Parent Hamiltonians of \\
Ergodic Matrix Product States}
%
}
\author[1, 2]{Owen Ekblad\orcidlink{0009-0006-0834-0327}\thanks{\url{oekblad@ucdavis.edu}}}
\author[1]{Eloy Moreno-Nadales\orcidlink{0009-0001-8178-7786}\thanks{\url{morenon4@msu.edu}}}
\author[1]{Eric B. Roon\orcidlink{0009-0000-7566-5734}\thanks{\url{rooneric@msu.edu}}}
\author[1]{Jeffrey H. Schenker\orcidlink{0000-0002-1171-7977}\thanks{\url{schenke6 @ msu.edu}}}

\affil[1]{Department of Mathematics\\ Michigan State University\\ East Lansing, MI, U.S.A.}
\affil[2]{Department of Mathematics and Center for Quantum Mathematics and Physics\\ University of California, Davis\\  
Davis, CA, U.S.A.}
\date{}                     
\setcounter{Maxaffil}{0}
\renewcommand\Affilfont{\itshape\small}
\maketitle

\begin{abstract}
Matrix product states (MPS) are quintessential examples of frustration-free gapped ground states of local interactions called parent Hamiltonians. 
In this work, we investigate parent Hamiltonians for a class of \textit{ergodic} matrix product states (EMPS), which are MPS defined by site-dependent random tensors $\{X_j^{[k]}\}_{j=1}^\BondDim$ which are homogeneously distributed at every site $k$ in the spin chain. 
Here, the EMPS are not translation-invariant but rather \textit{statistically} translation-invariant.
Under a mild injectivity assumption, we show the thermodynamic limit of an EMPS is the unique frustration-free ground state of a parent Hamiltonian on the whole spin chain, which, depending on the statistical properties of the EMPS, may or may not be finite-range. 
In contrast to the translation-invariant regime, these Hamiltonians need not be gapped.
Nevertheless, applying the martingale method while keeping track of local statistics gives conditions for a gap, in addition to pointing towards why there need not be a gap in general. 
We include examples of EMPS both with and without spectral gaps to illustrate our results. 
\end{abstract}
%
%
\tableofcontents
\section{Introduction}
Originally formulated by Fannes, Nachteragele, and Werner \cite{FannesNachtergaeleWerner} to investigate the AKLT antiferromagnet \cite{AKLT}, matrix product states (MPS) have become a cornerstone of condensed matter theory.  
MPS exhibit many properties representative of one-dimensional quantum matter, which, heuristically, is explained by their fundamental role in the density matrix renormalization group that efficiently approximates lowest energies of Hamiltonians of quantum spin chains \cite{White1992DensityGroups, White1993Density-matrixGroups}.
A central aspect of MPS is their role as unique frustration-free ground states of local gapped Hamiltonian interactions, called \textit{parent Hamiltonians} in \cite{Perez-Garcia_et_al}:
indeed, a key use of MPS is as a reverse-engineered approach to studying arbitrary local Hamiltonians, by first approximating ground states of these interactions by MPS then studying directly their parent Hamiltonians.
For the most part, investigations into MPS have been restricted to the translation-invariant (TI) regime. 
Recent studies initiated by the works \cite{MovassaghSchenker_PRX,MovassaghSchenker} extended fundamental results about TI MPS to MPS with homogeneously-distributed on-site disorder. 
That is, the MPS is \textit{statistically} translation invariant with respect to an underlying probability space. 
We examine the case when this disorder is implemented via \textit{ergodic} matrix product states (EMPS). 
In the very recent work \cite{RoonSchenker}, an example of an EMPS which is not the ground state of \textit{any} gapped finite-range interaction was constructed, which stands in stark contrast to the TI regime, where it is generically true that parent Hamiltonians of MPS are gapped. 
Here, we investigate the properties of parent Hamiltonians of EMPS in general so as to explain the discrepancy between TI MPS and EMPS more clearly. 
\subsection{Main results}
A TI MPS $\psi$ is described by an $\LocalDim$-tuple $\mathcal{T} = \big(X_j\big)_{j=1}^\LocalDim$ of matrices $X_j\in\MatricesBond$ such that $\sum_{j=1}^\LocalDim X_j^* X_j = \bondOne$ where $\MatricesBond$ denotes the set of $\BondDim\times \BondDim$ matrices with entries in $\C$, $\bondOne\in\MatricesBond$ is the identity matrix, $\LocalDim$ is the local dimension of the spin at each site of the integer lattice $\Z$, and $\BondDim$ is called the \textit{bond dimension}.  
We call such $\mathcal{T}$ a \textit{tensor} of bond dimension $\BondDim$.
The MPS is then defined locally via the $\Gamma$-maps
\begin{equation*}
\Gamma^{[l, r]}_{\operatorname{TI}}:\MatricesBond\to (\C^n)^{\otimes(r - l + 1)}, \qquad
    \Gamma_{\operatorname{TI}}^{[l, r]}(b)
    :=
    \sum_{i_l=1}^\LocalDim
        \cdots 
    \sum_{i_r = 1}^\LocalDim
    \TrD{bX_{i_r}\cdots X_{i_l}}\ket{i_l\cdots i_r}
\end{equation*}
for any $[l, r]\cap\Z\subset\Z$, where $\set{\ket{j}}_{j=1}^\LocalDim$ is a fixed basis of the local Hilbert space $\C^\LocalDim$ and $\set{\ket{i_l\cdots i_r}}$ is the corresponding basis of $(\C^{\LocalDim})^{\otimes(r - l + 1)}$,
and the MPS on $[l, r]$ is the pure state defined by 
\begin{equation*}
    \hat{\Gamma}_{\operatorname{TI}}^{[l, r]}
        :=
     \Gamma_{\operatorname{TI}}^{[l, r]}(\one_{\BondDim}).
\end{equation*}
Under the assumption that there is $L\in\N$---called the \textit{injectivity length}---for which $\Gamma^{[1, L]}$ is injective (in which case one says the MPS itself is injective), the thermodynamic limit
\begin{equation*}
    \lim_{[l, r]\uparrow\Z}\hat{\Gamma}_{\operatorname{TI}}^{[l, r]} =: \psi_{\operatorname{TI}}
\end{equation*}
exists in the weak$^*$ sense \cite{FannesNachtergaeleWerner}, where $\psi_{\operatorname{TI}}$ is a $C^*$-algebraic state on the $C^*$-algebra $\A{\Z} := \bigotimes_{x\in\Z}\MatricesLocal$.
Moreover, $\psi_{\operatorname{TI}}$ is a pure state. 
Further, if we let $\A{[l, r]} := \bigotimes_{x\in[l, r]\cap\Z}\MatricesLocal$, the parent Hamiltonian of $\psi_{\operatorname{TI}}$ was constructed in \cite{FannesNachtergaeleWerner}: 
it was shown that there is $R\in\N$ and positive semidefinite $h\in\A{[1, R]}$ such that, if we let $\tau:\A{\Z}\to\A{\Z}$ denote the translation, then $\psi$ is the unique frustration-free ground state of the Hamiltonian given formally by the expression
\begin{equation*}
    H_{\operatorname{TI}}
        :=
    \sum_{k\in\Z}\tau^k(h),
\end{equation*}
and $H_{\operatorname{TI}}$ has a spectral gap. 
Our model deviates from the translation-invariant one by allowing the tensors defining the MPS to be site-dependent. 
Specifically, we consider MPS defined by the $\Gamma$-maps
\begin{equation*}
     \Gamma^{[l, r]}(b)
    :=
    \sum_{i_l=1}^\LocalDim
        \cdots 
    \sum_{i_r = 1}^\LocalDim
    \TrD{bX_{i_r}^{[r]}\cdots X_{i_l}^{[l]}}\ket{i_l\cdots i_r},
\end{equation*}
where for all $k\in\Z$, $\mathcal{T}^{[k]} = \big(X_j^{[k]}\big)_{j=1}^\LocalDim$ is a \textit{local} tensor at site $k\in\Z$. 
This situation was investigated for finite-volumes in \cite{Perez-Garcia_et_al}, but we go beyond this work by operating under the additional assumption that $\seq{\mathcal{T}^{[k]}}_{k\in\Z}$ is sampled from an ergodic tensor-valued stochastic process. 
If for all $k\in\Z$ there is a random $L^{[k]}$---called the \textit{local injectivity length}---for which $\Gamma^{[k, k + L^{[k]}]}$ is injective, the thermodynamic limit $\psi$ of the local MPS exists \cite{MovassaghSchenker}, and we call $\psi$ the \textit{injective infinite-volume EMPS} defined by the local tensors $\mathcal{T}^{[k]}$. 
Our first theorem gives says the thermodynamic limit $\psi$ is a pure state and admits parent Hamiltonians for $\psi$. 
\begin{thmx}[Theorem \ref{Thm:Parent_Hamiltonian}, informal]
\label{Thm:Intro:PH}
    The thermodynamic limit $\psi$ of an injective EMPS is a pure state.
    Furthermore, there is an ergodic local projection-valued stochastic process $\seq{h^{[k]}}_{k\in\Z}$ such that $\psi$ satisfies
    \begin{equation}
        \psi(h^{[k]})
        =
        0
        \quad\text{for all $k\in\Z$}\,,
    \end{equation}
    uniquely, where, if we write $[k, R^{[k]}]$ to denote the support of $h^{[k]}$, the random length scales $(R^{[k]} - k)$ need not be bounded in $k$.
    Moreover, there is a random strongly continuous group $\seq{\alpha_t:\A{\Z}\to\A{\Z}}_{t\in\mathbb{R}}$ of $*$-automorphisms of $\A{\Z}$ such that for all $\bs{a}\in\A{\Z}$
    \begin{equation}
    \label{Eqn:Intro:SCUG}
        \lim_{n\to\infty}
        \left\|
            \alpha_t(\bs{a})
                -
            e^{-itH_n}\bs{a}
            e^{itH_n}
        \right\|
            =
            0
    \end{equation}
    holds uniformly for $t$ in compact intervals, where $H_n = \sum_{k\in[-n, n]}h^{[k]}$. 
\end{thmx}
We remark that this theorem pertains to \textit{thermodynamic limits} of MPS, as opposed to finite volume MPS.
Such $\psi$ in the above might be called an ergodic purely generated $C^*$-finitely correlated state on $\A{\Z}$ in the language of \cite{FannesNachtergaeleWerner, RoonSchenker}.  
We call $\seq{h^{[k]}}_{k\in\Z}$ as in the above theorem a parent Hamiltonian of $\psi$. 
The existence of a strongly continuous group $\seq{\alpha_t:\A{\Z}\to\A{\Z}}_{t\in\mathbb{R}}$ satisfying (\ref{Eqn:Intro:SCUG}) is assured for finite ranged interactions via Lieb-Robinson bounds \cite{LiebRobinson, NachtergaeleOgataSims, NachtergaeleSimsYoung, Young_thesis}, but if $\seq{R^{[k]} - k}_{k\in\Z}$ is unbounded, $\seq{h^{[k]}}_{k\in\Z}$ does \textit{not} define a finite ranged interaction: as we show, the interaction it defines is only \textit{locally} finite ranged.
Since each of the local terms is a \textit{projection}, hence has norm 1, this poses nontrivial technical challenges, and one of the main aspects of Theorem \ref{Thm:Intro:PH} is that such dynamics $\seq{\alpha_t}_{t\in\mathbb{R}}$ exist even in this non-finite ranged scenario. 
As we begin discussing the question of spectral gaps for these parent Hamiltonians, the reader should recall from the introduction above that there exist gapless injective EMPS. 
Nevertheless, we \textit{are} able to establish conditions on injective EMPS for which the parent Hamiltonian of Theorem \ref{Thm:Intro:PH} is gapped almost surely.
We follow the martingale method \cite{Nachtergaele} while keeping track of local statistical information along the way. 
To describe more specifically what this local statistical information is, we require some more notation. 
Associated to an EMPS defined by local tensors $\mathcal{T}^{[k]}$ there is a \textit{transfer apparatus} $(E^{[k]})_{k\in\Z}$: 
for all $k\in\Z$, we let $E^{[k]}:\MatricesBond\to\MatricesBond$ denote the random linear map 
\begin{equation*}
    E^{[k]}(b)
        :=
    \sum_{j=1}^\LocalDim
    X_j^{[k]*} b X_j^{[k]},
\end{equation*}
where we note that $E^{[k]}$ is a unital completely positive map for all $k$ since $\mathcal{T}^{[k]}$ is a tensor.
By \cite{Ekblad2024ReducibilityProcesses}, it is automatic that associated to the transfer apparatus there is an ergodic density matrix-valued stochastic process $\seq{\rho_k\in\MatricesBond}_{k\in\Z}$ such that 
\begin{equation*}
    \hatrho_k\circ E^{[k+1]} = \hatrho_{k+1}
\end{equation*}
holds almost surely for all $k\in\Z$, where $\hatrho$ denotes the linear functional $\hatrho(b) = \trD{\rho b}$.
We show in Appendix \ref{App:Injectivity} that the injectivity hypothesis is equivalent to the following dynamical property of the transfer apparatus. There is almost surely $N\in\N$ such that for all projections $p\in\MatricesBond$, $E^{[1]}\circ\cdots\circ E^{[N]}(p)$ is invertible. 
Phrased differently, $\seq{E^{[k]}}_{k\in\Z}$ is \textit{eventually strictly positive (ESP)}. 
This property was studied in depth by Movassagh and Schenker in \cite{MovassaghSchenker_PRX, MovassaghSchenker}, where they showed that if $\seq{E^{[k]}}_{k\in\Z}$ is ESP, then the corresponding random density matrices $\rho_k$ are almost surely invertible for all $k$, and satisfy
\begin{equation}
\label{Eqn:Intro:ESP}
    \Big\|
        E^{[l]}\circ\cdots\circ E^{[r]}(\cdot)
            -
        \frac{1}{\BondDim}\hatrho_r(\cdot) \one_{\BondDim}
    \Big\|
        \leq 
    \mu^{r - l} \quad \text{for } r > \xi^{[l]} + l,
\end{equation} 
where $\mu \in (0,1)$ is deterministic and $\xi^{[l]}$ is a site-dependent random \emph{onset length}, representing a threshold beyond which exponential decay is witnessed.
This shows, among other things, that almost sure exponential clustering holds for injective EMPS \cite{MovassaghSchenker}, but this need not be uniformly visible locally at scales below the onset length. 
Indeed, the gapless example of \cite{RoonSchenker} took advantage of this nonuniformity, obtaining a corresponding lower bound for $r < l + \xi^{[l]}$ and leveraging the general exponential clustering result of \cite{Nachtergaele2006Lieb-RobinsonTheorem} to prove gaplessness.  
It is therefore not surprising that this dynamical data is the main ingredient to our lower bound on the spectral gap.
\begin{thmx}[Theorem \ref{Thm:Local_spectral_gap_prop} and Corollary \ref{Cor:Gaps_for_finite_inj_length}, informal]
\label{Thm:Intro:Gaps}
    For a thermodynamic limit $\psi$ of injective EMPS where the parent Hamiltonian $\seq{h^{[k]}}_{k\in\Z}$ as in Theorem \ref{Thm:Intro:PH} has uniformly bounded length scale $\seq{R^{[k]} - k}_{k\in\Z}$, there is an increasing and absorbing sequence $(\Lambda_N)_{N\in\N}$ and random non-negative constants $\seq{\beta_N}_{N\in\N}$ such that 
    \begin{equation*}
        \operatorname{spec-gap}\!\seq{
            \sum_{[k, R^{[k]}]\cap\Lambda_N\neq\emptyset}
            h^{[k]}
        }
        \geq 
            {\beta}_N
    \end{equation*}
    holds for all $N$, for all $[k, R^{[k]}]\cap\Lambda_N\neq\emptyset$ and $\beta_N$ depends only on the decay constant $\mu$ in \eqref{Eqn:Intro:ESP}, the worst-case onset length $\xi^{[k]}$ for $k\in\Lambda_N$, and the smallest eigenvalues of $\rho_k$ for $k\in\Lambda_N$. 
    In particular, if the onset lengths are uniformly bounded and $\rho_k$ is uniformly bounded away from zero, the parent Hamiltonian is gapped. 
\end{thmx}
To illustrate this theorem, we provide a class of examples (Proposition \ref{Prop:DAKLT}) that satisfies the hypotheses of the above theorem. 
This example is a disordered deformation of the AKLT model, of the same type that \cite{RoonSchenker} showed is gapless, where it can directly be seen that $\beta_N\to 0$ as $N\to\infty$. 
Proposition \ref{Prop:DAKLT}, demonstrates how to modify the example of \cite{RoonSchenker} in such a way that $\beta_N>0$ uniformly in $N$, which, by the above theorem, demonstrates a bulk gap. 
This example also points to a possibly general phenomenon regarding spectral gaps of EMPS, which we describe in more detail below in Section \ref{Sec:Outlook}

\subsection{Related literature}
This paper is based on (and follows the methodology prescribed in) the seminal work of Fannes, Nachtergaele, and Werner \cite{FannesNachtergaeleWerner}, which systematically described what later came to be known as matrix product states. 
The literature on MPS is vast, and we only recount a small selection of relevant works here. 
MPS are useful for approximating arbitrary ground states of quantum spin chains \cite{FannesNachtergaeleWerner_abundance, VerstraeteCirac, Vidal, White1992DensityGroups, White1993Density-matrixGroups}, they provide a rich class of example models for which one can compute symmetry protected topological phase indices and string-order parameters \cite{Ogata_z2, Ogata_cmp, Ogata_cdm, PG_StringOrder, Pollman_et_al_z2, Pollmann_et_al_entanglement}, they exhibit local topological quantum order \cite{NachtergaeleSimsYoung_2}, and are a first example of tensor network states \cite{Cirac2021MatrixTheorems,SanzPerezGarciaWolfCirac}.
Parent Hamiltonians of general, non-translation invariant \textit{finite volume} MPS were described in \cite{Perez-Garcia_et_al}, but there has been less work on thermodynamic limits and conditions for bulk spectral gaps of infinite-volume MPS as we describe here.
It has already been seen in \cite{Fernandez-Gonzalez2015FrustrationStates} that local perturbations of non-injective TI MPS can result in gapless parent Hamiltonians---called ``uncle Hamiltonians'' in \cite{Fernandez-Gonzalez2015FrustrationStates}---but the sort of perturbation we are considering here is different for two reasons: first, we operate under an injectivity assumption for our EMPS, and second, perhaps more notably, our disorder is highly nonlocal in the sense that it permeates the entire spin chain in an ergodic way. 
The works \cite{MovassaghSchenker, MovassaghSchenker_PRX} stimulated a great deal of activity towards understanding infinite-volume EMPS and other classes of non-TI infinite volume states on the spin chain \cite{Nelson2024ErgodicAlgebras, Pathirana2026AsymptoticStates, Pathirana2026CorrelationStates, Souissi2025ErgodicProcesses,  RoonSchenker, RoonSchenker_BGA}.
%
\subsection{Discussion and an open problem}\label{Sec:Outlook}
In this work, we described the parent Hamiltonians of matrix product states defined by site-dependent tensors satisfying a statistical version of translation invariance. 
We then gave a sufficient condition that ensured such a parent Hamiltonian was gapped, as it is known in general that not all EMPS are gapped ground states of suitable parent Hamiltonians. 
This is suggestive of a problem that will likely require new ideas and seems to be an interesting new direction: 
\begin{problem}
    Classify the set of gapped injective EMPS.
\end{problem}
In the case that $\rho_k$ is bounded away from zero, it is natural from our analysis to ask if the unboundedness of the onset length is the \textit{only} obstruction to the EMPS being gapped, even in, e.g., the i.i.d. case. 
It is nevertheless possible that there are other obstructions to injective EMPS being gapped, and this problem seems worthwhile to investigate in the future. 
\section*{Acknowledgments}
EBR would like to thank Dr. Bruno Nachtergaele for the helpful comments about the construction underpinning Proposition~\ref{Prop:DAKLT}. 
EBR is also grateful to Dr. Amanda Young for the helpful discussions about the Martingale method.
OE appreciates discussions with Lubashan Pathirana and Albert H. Werner regarding the content of Appendix \ref{App:Injectivity}.
OE, EM-N, and JHS were supported by the US National Science Foundation under Grant No. DMS-2153946.
%

%
%
\subsection{Organization of the paper}
In Section \ref{sec:Preliminaries}, we describe the technical preliminaries about quantum spin chains and ergodic quantum processes.
In Section \ref{sec:Vectorization}, we prove the basic technical facts about ergodic matrix product states that we require in our work, describing how many of the technical results of \cite[Section 5]{FannesNachtergaeleWerner} extend to the ergodic regime. 
Section \ref{sec:PHams} contains the proof of the existence of parent Hamiltonians (Theorem \ref{Thm:Intro:PH}), and we conclude the main body of the paper in Section \ref{sec:GapEstimates}, where we apply the martingale method to establish lower bounds on the spectral gap of these parent Hamiltonians we constructed (Theorem \ref{Thm:Intro:Gaps}).
In Appendix \ref{App:Injectivity}, we establish the equivalence between injectivity of an ergodic matrix product state and eventual strict positivity of its corresponding transfer apparatus ergodic quantum process, a fact generalizing \cite[Proposition 3]{SanzPerezGarciaWolfCirac} that seems to be of independent interest. 


\section{Preliminaries}\label{sec:Preliminaries}
We begin by setting notation and stating some preliminary technical facts we make reference to later. 
For a Banach space $\mathcal{X}$, we write $\dual{\mathcal{X}}$ to denote the dual space bounded linear functionals $\varphi:\mathcal{X}\to\mathbb{C}$. 
We let $\mathcal{B}(\mathcal{X})$ denote the set of bounded linear maps $\mathcal{X}\to\mathcal{X}$.
Given $A\in\mathcal{B}(\mathcal{X})$, let $\operatorname{spec}(A)$ denote the spectrum of $A$.
If $\operatorname{spec}(A)\subset[0, \infty)$ and $0\in\operatorname{spec}(A)$, we put $\operatorname{spec-gap}(A)$ equal to the quantity
\begin{equation*}
    \operatorname{spec-gap}(A)
        :=
    \inf\set{\delta>0\,\,:\,\,
    (0, \delta)\cap\sigma(A) = \emptyset}.
\end{equation*}
We write $\|x\|_{\mathcal{X}}$ to denote the norm of $x\in \mathcal{X}$, and given a bounded linear operator $A:\mathcal{X}\to\mathcal{X}$, we abuse notation slightly and write $\|A\|$ to denote its operator norm. 
Given a $C^*$-algebra $\mathcal{A}$ with unit $\one$, we let $\cstates{\mathcal{A}}$ denote the subset of $\dual{\mathcal{A}}$ consisting of those $\varphi\in\dual{\mathcal{A}}$ satisfying $\varphi(\one) = 1$ and $\varphi(a^*a)\geq 0$ for all $a\in\mathcal{A}$. Such an element of $\cstates{\mathcal{A}}$ is called a state. 
A state is called \emph{pure} if it is an extreme point of the convex closed set $\cstates{\mathcal{A}}$.
We let $\positives{\mathcal{A}}$ denote the set of $b\in\mathcal{A}$ of the form $b = a^*a$ for some $a\in\mathcal{A}$, and we let $\selfadjoints{\mathcal{A}}$ denote the set of $a\in\mathcal{A}$ with $a^* = a$. 
Given $\varphi\in\cstates{\mathcal{A}}$, the GNS representation of $\mathcal{A}$ with respect to $\varphi$ is the tuple $(\pi_\varphi, \mathcal{K}_\varphi, \Xi_\varphi)$, where $\mathcal{K}_\varphi$ is the Hilbert space defined as the completion of the inner product space
\begin{equation*}
    ({\mathcal{A}/N_\varphi}, \inner{\cdot}{\cdot}_\varphi)
\end{equation*}
where $N_\varphi$ is the closed left ideal $N_\varphi = \set{a\in\mathcal{A}\,\,:\,\, \varphi(a^*a) = 0}$, and $\inner{a}{b}_\varphi := \varphi(a^*b)$ for $a, b\in\mathcal{A}$, $\pi_\varphi:\mathcal{A}\to \mathcal{B}(\mathcal{K}_\varphi)$ is the representation 
\begin{equation*}
    \pi_\varphi(a)[b]
    =
    [ab]
\end{equation*}
where $[b]$ denotes the equivalence class of $b$ in $\mathcal{A}/N_\varphi$, and $\Xi_\varphi$ is the cyclic vector $[\one_{\mathcal{A}}]$.

We call a linear map $\phi:\mathcal{A}\to\mathcal{B}$ between $C^*$-algebras positive if $\phi\seq{\positives{\mathcal{A}}}\subset \positives{\mathcal{B}},$ and we call $\phi$ unital if $\phi(\one_\mathcal{A}) = \one_\mathcal{B}$. 
For $n\in\N$, we let $\mathbb{M}_n(\mathcal{A})$ denote the unital $C^*$-algebra consisting of matrices with entries in $\mathcal{A}$ with the usual operations, which the reader may readily verify may be identified with the tensor product $\mathcal{A}\otimes\mathbb{M}_n$, where $\mathbb{M}_n$ denotes the set of $n\times n$ matrices with entries in $\mathbb{C}$. 
A linear map $\phi:\mathcal{A}\to\mathcal{B}$ is called completely positive if for all $n\in\mathbb{N}$, the map 
\begin{equation*}
    \phi\otimes\operatorname{Id}_n:\mathbb{M}_n(\mathcal{A})\to\mathbb{M}_n(\mathcal{B})\,,
\end{equation*}
is positive, where $\operatorname{Id}_n$ denotes the identity map on $\mathbb{M}_n$. 
For shorthand, we call a linear map $\phi:\mathcal{A}\to\mathcal{B}$ that is both unital and completely positive ucp, and we let $\ucps{\mathcal{A}, \mathcal{B}}$ denote the set of unital and completely positive linear maps $\phi:\mathcal{A}\to\mathcal{B}$.
We write $\ucps{\mathcal{A}}$ to denote $\ucps{\mathcal{A}, \mathcal{A}}$.

Throughout, $\matrices$ will denote the $d\times d$ matrices with complex-valued entries, $\P_d\subset\matrices$ will denote the positive semidefinite elements in $\matrices$.
For $a\in\matrices$, we let $\|a\|$ denote the usual operator norm. 
For $\xi, \eta\in\mathbb{C}^d$, we use the standard bra-ket notation, writing $\braket{\xi}{\eta}$ to denote the usual inner product (and $\ketbra{\xi}{\eta}$ the outer product) on $\mathbb{C}^d$, and we write $\|\xi\|$ to denote $\sqrt{\braket{\xi}{\xi}}$ to denote the corresponding norm. 
To denote the trace of $a\in\matrices$, we write $\operatorname{Tr}_d(a)$, emphasizing the dimension $d$. 
We let $\states\subset \P_d$ denote the density matrices, i.e., $\states$ is the set of $\rho\in\P_d$ such that $\operatorname{Tr}_d(\rho) = 1$, and we let $\stateso$ denote the set of invertible density matrices.
Let $\one_d$ denote the identity matrix in $\matrices$.
We identify $\states$ with $\cstates{\matrices}$ via
\begin{equation}
    \rho\in\states 
    \quad\longleftrightarrow\quad 
    \Big(\hat{\rho}:\matrices\to\C,\,  a\mapsto \operatorname{Tr}_d(\rho a)\Big).
\end{equation}
Accordingly, given $a\in\matrices$ and $\rho\in\states$, we write $\hat{\rho}(a)$ to denote $\operatorname{Tr}_d(\rho a)$. 
Given $\rho\in\states$, we define 
\begin{equation}
    \inner{a}{b}_\rho
        :=
    \hatrho(a^*b),
\end{equation}
which, if $\rho$ is invertible, is a nondegenerate inner product on $\matrices$. 
We let $\|a\|_\rho:= \sqrt{\inner{a}{a}_\rho}$ be the corresponding norm. 
Given two finite-dimensional $\C$-vector spaces $V$ and $W$, we write $\operatorname{Lin}(V, W)$ to denote the set of $\C$-linear maps $V\to W$. 
We make this into a measurable space by giving it the Borel $\sigma$-algebra induced by any norm making $\operatorname{Lin}(V, W)$ into a Banach space. 
%
%
%
\subsection{Ergodic Quantum Processes and Random States}
Throughout this work, we shall take $(\Omega, \mathcal{F}, \Pr)$ to be a fixed probability space, where $\mathcal{F}$ is the $\sigma$-algebra and $\Pr$ the probability measure. 
Any statements about maps being measurable are always with respect to $\mathcal{F}$, and statements about a property holding almost surely are with respect to $\Pr$. 
For a $C^*$-algebra $\mathcal{A}$, we make $\cstates{\mathcal{A}}$ into a measurable space by endowing it with the Borel $\sigma$-algebra $\mathcal{B}(\operatorname{wk}^*)$ induced by the weak$^*$ topology on $\cstates{\mathcal{A}}$. 
We write $\mathcal{S}(\Omega, \mathcal{A})$ to denote the set of functions $\psi:\Omega\to\cstates{\mathcal{A}}$ that are measurable with respect to $\mathcal{B}(\operatorname{wk}^*)$, and we call an element $\psi$ of $\cstates{\Omega, \mathcal{A}}$ a random state (on $\mathcal{A}$).
We write $\rmatrices$ to denote the set of measurable maps $x:\Omega\to\matrices$, and, similarly, we write $\rstates$ (resp. $\rstateso$) to denote the set of measurable maps $\rho:\Omega\to\states$ (resp. $\rho:\Omega\to\stateso$).
One of the primary objects considered in this work is \textit{ergodic quantum processes}, whose definition we now recall.
\begin{definition}[Ergodic quantum process]
    An ergodic quantum process on $\matrices$ is a map 
    \begin{equation*}
        F:\Omega\times\matrices \to\Omega\times\matrices\,,
    \end{equation*}
    defined by $F(\omega, a) = (T(\omega), \phi_{T(\omega)}(a))$, where $T:\Omega\to\Omega$ is an invertible measure-preserving and ergodic map and $\phi:\Omega\to\ucps{\M_d}$ is measurable. 
    We write $(T, \phi)$ to denote an ergodic quantum process. 
\end{definition}
Ergodic quantum processes were first studied at this level of generality by Movassagh and Schenker in \cite{MovassaghSchenker}, but were studied earlier in special cases by various authors under the guise of repeated interactions (see, e.g., \cite{Bruneau2008RandomSystems}, where the authors consider the I.I.D. case).
\begin{rmk}
    In the literature, there are a number of different definitions of ergodic quantum processes. 
    For example, $\phi$ may be assumed to be trace preserving almost surely, or, more generally, $\phi$ may only be assumed to almost surely satisfy certain general positivity and faithfulness assumptions which are automatic for ucp (and trace preserving) maps; this general set up is taken in \cite{MovassaghSchenker}. 
\end{rmk}
Given an ergodic quantum process $(T, \phi)$ on $\matrices$ and $l, r\in\Z$ with $l\leq r$, we define $\phi^{[l, r]}:\Omega\to\ucps{\matrices}$ by 
\begin{equation}\label{eqn:Generic_EQP}
    \phi^{[l, r]}_\omega 
    :=
    \phi_{T^{l}(\omega)}
    \circ\cdots\circ 
    \phi_{T^{r}(\omega)},
\end{equation}
where the reader should recall we have assumed $T$ to be invertible. 
\subsection{Quantum Spin Chains}
In this work, we are concerned with quantum spin chains, for which we recall all the standard notations and definitions here. 
Consider the integer lattice $\Z$. 
For $\Lambda\subset\Z$, if $|\Lambda|<\infty$, we write $\Lambda\Subset\Z$.
Given $x\in\mathbb{Z}$, let $\A{\{x\}}$ denote $\MatricesLocal$, and, more generally, for $\Lambda\Subset\Z$, define $\A{\Lambda}$ by
\begin{equation}
    \A{\Lambda} 
        := 
    \bigotimes_{x\in \Lambda} \A{\{x\}} \cong (\MatricesLocal)^{\otimes |\Lambda|}.
\end{equation}
The algebra $\A{\Lambda}$ is called \emph{local algebra} at $\Lambda$. 
For any $\Lambda'\subset\Lambda''\Subset\Z$, we note that $\A{\Lambda'}$ embeds $*$-isometrically intro $\A{\Lambda''}$ by tensoring with $\one_{\Lambda'' \setminus \Lambda'}:= \bigotimes_{x\in \Lambda'' \setminus \Lambda'} \one_x$. 
When necessary, we denote this embedding by $\iota_{\Lambda'\to\Lambda''}$, but often we identify $\A{\Lambda'}$ with its image $\iota_{\Lambda'\to\Lambda''}(\A{\Lambda'})$. 
For any (possibly infinite) subset $\Xi\subset\Z$, $\{\A{\Lambda}\colon \Lambda \Subset \Xi\}$ forms a net of both $*$-algebras and $C^*$-algebras.
We let $\A{\Xi}^{\loc}$ denote the $*$-algebra containing all the local algebras of observables $\{\A{\Lambda}\colon \Lambda \Subset \Xi\}$\textit{ as $*$-algebras}, i.e., 
\begin{equation}
    \A{\Xi}^{\loc}
        =
    \bigcup_{\Lambda \Subset \Xi} \A{\Lambda}.
\end{equation}
The above union is a unital-$\ast$ algebra once the $\A{\Lambda}$ are identified with their inclusions into larger volumes so that $\A{\Xi}^{\loc}$ is an inductive limit in the category of $*$-algebras. 
This is called the \emph{local algebra}. 
For $\bs{a}\in\A{\Z}^{\loc}$, we write $\operatorname{supp}(\bs{a})\Subset\Z$ to denote the minimal $\Lambda\Subset\Z$ where $\bs{a}\in\A{\Lambda}$, and we write $\|\bs{a}\|$ to denote the operator norm of $\bs{a}$ as an element of $\A{\operatorname{supp}(\bs{a})}$. 
We write $\one$ to denote the unit of $\A{\Z}^{\loc}$. 
The operator-norm closure of this $*$-algebra is therefore a unital $C^*$-algebra, and we denote it by $\A{\Xi}$, which is called the \emph{quasi-local algebra}. 
In symbols,  
\begin{equation}
     \A{\Xi} 
        :=
    \overline{\A{\Xi}^{\loc}}^{\|\cdot \|}\,.
\end{equation} 
Phrased differently, $\A{\Xi}$ is the inductive limit of $\{\A{\Lambda}\colon \Lambda \Subset \Xi\}$ in the category of $C^*$-algebras \cite[Chapter 6]{Murphy}.
We call $\A{\Z}$ the \emph{quasilocal algebra}.
We implicitly identify $\A{\Xi}$ with the $C^*$-subalgebra of $\A{\Z}$ defined to be the image of the map $\A{\Xi}\ni a\mapsto a\otimes \one_{\Z\setminus\Xi}$, where $\one_{\Z\setminus\Xi}$ is the unit of $\A{\Z\setminus\Xi}$.
We let $\tau:\A{\Z}\to\A{\Z}$ denote the $*$-automorphism defined by 
\begin{equation}
    \tau(a_l\otimes \cdots \otimes a_r) = a_{l+1}\otimes \cdots \otimes a_{r+1}\,,
\end{equation}
whenever $l\leq r$, and we write $\tau_k$ to denote $\tau^k: \A{\Z}\to\A{\Z}$ for any $k\in\mathbb{Z}$. 
Thus, $(\tau_k)_{k\in\Z}$ is a group action of translation by $*$-automorphisms $\Z \acts{\tau} \A{\Z}$.
\begin{definition}[Ergodic state]\label{Def:TCV}
    We say that $\psi\in \mathcal{S}(\Omega, \mathcal{A}_{\mathbb{Z}})$ is an \textbf{ergodic state} if for all $k\in\mathbb{Z}$,
    \begin{equation}
        \psi_\omega\circ\tau_k 
        =
        \psi_{T^k(\omega)}\,,
    \end{equation}
    holds for almost every $\omega\in\Omega$. 
    We write $\TCV$ to denote the set of all ergodic states on $\mathcal{A}_\mathbb{Z}$.
\end{definition}
Let $[\LocalDim]$ denote the set $\set{1, \dots, \LocalDim}$, and for $\Lambda\Subset\Z$, let $[\LocalDim]^\Lambda$ be the set of multi-indices, i.e., 
\begin{equation*}
    [\LocalDim]^\Lambda
        =
    \set{
    (j_\lambda)_{\lambda\in\Lambda}
        \,\,:\,\,
    j_\lambda\in[\LocalDim]
    }.
\end{equation*}
Endow $\C^\LocalDim$ with an orthonormal basis $\{\ket{e_j}\}_{j=1}^\LocalDim$, and given $\Lambda\Subset\Z$ and $\sigma = (j_\lambda)_{\lambda\in\Lambda}\in[\LocalDim]^{\Lambda}$, let 
\begin{equation*}
    \ket{e_\sigma}
        =
    \bigotimes_{\lambda\in\Lambda}\ket{e_{j_\lambda}}.
\end{equation*}
Write $\Hilb{\Lambda}$ to denote the vector space with basis $\set{\ket{e_\sigma}}_{\sigma\in[\LocalDim]^\Lambda}$. 
For any $\Lambda\Subset\Z$ and $k\in\Z$, by a slight abuse of notation we write $\tau_k$ to denote the isomorphism in $\operatorname{Lin}(\Hilb{\Lambda}, \Hilb{\Lambda + k})$ defined by 
\begin{equation}\label{Def:tau_for_local_Hilbert_spaces}
    \begin{split}
        \tau_k:\Hilb{\Lambda} &\to \Hilb{\Lambda + k}\,,\\
        \ket{e_\sigma}
        &\mapsto 
        \ket{e_{\tau_k(\sigma)}},
    \end{split}
\end{equation}
where $\tau_k(\sigma) = (j_{\lambda + k})\in[\LocalDim]^{\Lambda + k}$ whenever $\sigma = (j_\lambda)_{\lambda\in\Lambda}$. 
Recall Takeda's theorem \cite{Takeda}, which relates finite-volume states to their bulk counterparts.

\begin{thm}[Takeda]\label{Thm:Takeda}
    Let $\A{\Z}$ be the quasilocal algebra.
    Suppose $\phi_{\Lambda}\in \cstates{\A{\Lambda}}$ is a family of states satisfying the following compatibility condition: whenever $\Lambda_0 \subset \Lambda \Subset \mathbb{Z}$ one has
        \begin{equation}
            \phi_{\Lambda} \circ \iota_{\Lambda_0 \to \Lambda} = \phi_{\Lambda_0}\,.
        \end{equation}
    Then, there is a uniquely defined state $\psi\in \cstates{\A{\Z}}$ so that $\psi|_{\A{\Lambda}} = \phi_{\Lambda}$ for all $\Lambda \Subset \mathbb{Z}$. Conversely, every bulk state is determined by its finite volume restrictions. 
\end{thm} 
\subsection{Interactions, Dynamics, and Ground States}\label{Sec:LRB}
In this section, we recall the standard formalism for studying interacting quantum spin systems \cite{BratteliRobinsonI, BratteliRobinsonII}.
Let $\mathscr{P}_0(\Z)$ denote the finite subsets of $\Z$, and to denote $\Lambda\in\mathscr{P}_0(\Z)$, we write $\Lambda\Subset\Z$. 
\begin{definition}[Interactions and local dynamics]
    An interaction is a function $\Phi: \mathscr{P}_0(\Z) \to \A{\Z}^{\loc}$ such that $\Phi(\Lambda) = \Phi(\Lambda)^*$ for all $\Lambda\Subset\Z$. 
    We say $\Phi$ is a \emph{positive interaction} if $\Phi(\Lambda)\geq 0$ for all $\Lambda\Subset\Z$. 
    For $\Lambda\Subset\Z$, we define $H^{\Phi, \Lambda}$ by
    \begin{equation}
        H^{\Phi, \Lambda}
            :=
        \sum_{Z\subset\Lambda}
        \Phi(Z),
    \end{equation}
    and we call $H^{\Phi, \Lambda}$ the total Hamiltonian in $\Lambda$. 
    We call the one-parameter group of $*$-automorphisms $(\alpha^{\Phi, \Lambda}_t)_{t\in\mathbb{R}}$ defined by
    \begin{equation*}
        \begin{split}
            \alpha_t^{\Phi, \Lambda}:\A{\operatorname{supp}(H^{\Phi, \Lambda})}&\to \A{\operatorname{supp}(H^{\Phi, \Lambda})}\\
            \alpha^{\Phi, \Lambda}_t(\bs{a})
                &:=
             e^{itH^{\Phi, \Lambda}}\,\bs{a}\, e^{-itH^{\Phi, \Lambda}}
        \end{split}
    \end{equation*}
    the \emph{local dynamics generated by} $\Phi$ in $\Lambda$. 
    We let $\delta^{\Phi, \Lambda}:\A{\Z}\to\A{\Z}$ be the derivation defined by 
    \begin{equation*}
        \delta^{\Phi, \Lambda}(\bs{a})
            :=
            i[H^{\Phi, \Lambda}, \bs{a}].
    \end{equation*}
    Notice that $\alpha^{\Phi, \Lambda}_t = \exp(t\delta^{\Phi, \Lambda})$, i.e., $\delta^{\Phi, \Lambda}$ is the infinitesimal generator for $(\alpha^{\Phi, \Lambda}_t)_{t\in\mathbb{R}}$ as a strongly continuous one-parameter group \cite{EngelNagel}. 
\end{definition}
We will be interested in analyzing the thermodynamic limit as $\Lambda \uparrow \Z$. 
The typical method to do this in the literature is to use Lieb-Robinson bounds to show that the limit $\lim_{\Lambda\uparrow\Z}\alpha^{\Phi, \Lambda}_t$ exists in the strong operator topology uniformly for all $t$ in compact subsets of $\mathbb{R}$ \cite{NachtergaeleOgataSims, NachtergaeleSimsYoung}.
In our investigation below, however, we encounter interactions which are only \emph{locally} finite ranged with non-decaying norm and thus not subject to the usual requirement for a Lieb-Robinson bound to hold \cite{NachtergaeleOgataSims}. 
To be precise, we say that an interaction $\Phi$ is \emph{locally finite ranged} if for every $x\in \Z$ the quantity
\[
    n_\Phi(x) = \#\{Z\Subset \Z \colon x\in Z \text{ and } \Phi(Z) \neq 0\}
\] 
is finite. 
If, in addition, $\sup_x n_\Phi(x) < \infty$, then $\Phi$ is called finite range.
Given a locally finite ranged interaction $\Phi$, we define the derivation $\delta^\Phi:\A{\Z}^{\loc}\to \A{\Z}^{\loc}$ by
\begin{equation}\label{eqn:canonical_derivation}
    \delta^\Phi(\bs{a}) := \sum_{\substack{\Lambda\Subset \Z}} i[\Phi(Z), \bs{a}], 
\end{equation} 
where we have noted that since $\operatorname{supp}(\bs{a})\Subset\Z$ and $\Phi$ is locally finite range, the above sum is always a finite sum for any $\bs{a}\in\A{\Z}^{\loc}$. 
The following proposition from \cite{BratteliRobinsonII} gives a condition for which the limit $\lim_{\Lambda\uparrow\Z}\alpha^{\Phi, \Lambda}_t$ exists as described above. 
\begin{prop}[\texorpdfstring{\cite[Theorem 6.2.6]{BratteliRobinsonII}}{l}]
\label{Prop:Dynamics_from_loc_fin_interactions}
    Let $\Phi$ be a locally finite range interaction. 
    If there is an increasing and absorbing sequence $(\Lambda_n)_{n\in\N}$ and $M > 0$ such that 
    \begin{equation}
    \label{Eqn:Surface_energy}
        \left\|
            \sum_{
            \substack{
                \Lambda\cap\Lambda_n\neq\emptyset\\
                \Lambda\cap\Lambda_{n}^c\neq\emptyset
                }
            }
            \Phi(\Lambda)
        \right\|
        \leq M,
    \end{equation} 
    then $\delta^\Phi$ as in \eqref{eqn:canonical_derivation} is a closable unbounded operator on $\A{\Z}$, and the closure of $\overline\delta^\Phi$ generates a strongly continuous group $(\alpha^\Phi_t)_{t\in\mathbb{R}}$ of $*$-automorphisms of $\A{\Z}$. 
    Moreover, 
    \begin{equation*}
       \lim_{n\to\infty}
       \left\| 
            \alpha_t^\Phi(\bs{a})
            -
            \alpha_t^{\Phi, \Lambda_n}(\bs{a})
       \right\|
       =
       0
    \end{equation*}
    holds for all $\bs{a}\in\A{\Z}$ uniformly for $t$ in compact intervals. 
\end{prop}
We say a locally finite ranged interaction has \emph{bounded surface energy} if (\ref{Eqn:Surface_energy}) holds for some $M>0$ and some increasing and absorbing sequence $(\Lambda_n)_{n\in\N}$. Given the existence of dynamics in this way, our main focus are the \textit{ground states}. 
\begin{definition}[Ground states and frustration-freedom]
\label{Def:Ground_state_of_interaction}
    Let $\Phi$ be a locally finite ranged interaction. 
    A state $\varphi\in\cstates{\A{\Z}}$ is called a \textit{ground state} if 
    \begin{equation}\label{eqn:ground_state}
        -i
            \varphi(\bs{a}^*\delta^\Phi(\bs{a}))\geq 0
            \quad 
            \text{for all $\bs{a}\in\A{\Z}^{\loc}$.}
    \end{equation}
    Given a ground state $\varphi$, we say $\Phi$ has spectral gap $\gamma>0$ above $\varphi$ if 
     \begin{equation}\label{eqn:spectral_gap}
        -i\varphi\big( 
            \bs{a}^* \delta^\Phi(\bs{a})
        \big)
        \geq 
        \gamma \varphi(\bs{a}^*\bs{a})
    \end{equation}
    holds for all $\bs{a}\in\A{\Z}^{\loc}$ with $\varphi(\bs{a}) = 0$. 
    A state $\varphi\in\cstates{\A{\Z}}$ is called \textit{frustration-free} if $\varphi(\Phi(\Lambda)) = 0$ for all $\Lambda\Subset\Z$. 
\end{definition}
If $\Phi$ is a positive interaction, it is straightforward to show that frustration-freeness of $\varphi\in\cstates{\A{\Z}}$ implies $\varphi$ is a ground state. 
One can also check that $\varphi\in\cstates{\A{\Z}}$ is a ground state if and only if $\alpha^{\Phi}_t\circ\varphi = \varphi$ for all $t\in\mathbb{R}$:
this equivalence is originally due \cite[Theorem 2]{PowersSakai} (see also \cite[Theorem 6.2.52]{BratteliRobinsonII}).
In this work, will be concerned with the analysis of local spectral gaps for our locally finite ranged interaction which we construct in Theorem~\ref{Thm:Parent_Hamiltonian}. 
One of our aims is to find sufficient conditions for a bulk spectral gap to be open despite the disorder. 
We record the so-called covariant GNS representation of a one-parameter group of $\ast$-automorphisms which have a ground state $\varphi$. 
\begin{prop}
\label{prop:CovariantGNS}
    Let $(\alpha_t)_{t\in \mathbb{R}}$ be a strongly-continuous one-parameter group of $\ast$-automorphisms on $\A{\Z}$ generated by a densely defined closable derivation $\delta$.
    Suppose $\varphi\in\cstates{\A{\Z}}$ is such that $\alpha_t\circ\varphi = \varphi$ for all $t\in\mathbb{R}$, and
    let $(\pi_\varphi, \mathcal{K}_\varphi, \Xi_\varphi)$ be the GNS representation of $\A{\Z}$ with respect to $\varphi$. 
    Then 
    \begin{equation}
        U_t\, \pi_\varphi(\bs{a})\,\Xi_\varphi := \pi_\varphi (\alpha_t(\bs{a}))\, \Xi_\varphi \quad \text{for all $\bs{a}\in\A{\Z}$}
    \end{equation} 
    defines a strongly continuous one-parameter group of unitary operators on $\mathcal{K}_\varphi$. In particular, Stone's theorem implies there is a unique unbounded self-adjoint operator $H^\delta$ with domain $\pi_\varphi(\A{\Z})\Xi_\varphi$ which generates $(U_t)_{t\in \mathbb{R}}$ in the sense that $U_t = e^{-itH^\delta}$ for all $t\in \mathbb{R}$. 
\end{prop}
In accordance with this proposition, given a locally finite range interaction $\Phi$ with bounded surface energy and ground state $\varphi$, we write $H^{\Phi}$ to denote the self-adjoint operator corresponding to $\delta^\Phi$ in the above proposition and call $H^\Phi$ the GNS Hamiltonian corresponding to $\varphi$.  
A key tool in our analysis is the following straightforward generalization of Koma-Nachtergaele \cite{KomaNachtergaele}, whose proof follows exactly as in \cite[Proof of Theorem 3.1.2]{Young_thesis}.
\begin{prop}[Koma-Nachtergaele Inequality]
\label{Prop:Koma_Nachtergaele}
    Suppose $\Phi$ is a positive locally finite ranged interaction with bounded surface energy, let $\varphi$ be a frustration-free ground state, and let $H^\Phi$ be the corresponding GNS Hamiltonian. 
    Then
    \begin{equation*}
        \operatorname{spec-gap}(H^{\Phi})
            \geq 
        \limsup_{n\geq 1}\,
        \operatorname{spec-gap}(H^{\Lambda_n, \Phi})
    \end{equation*}
    for any increasing and absorbing sequence $\Lambda_n\Subset\Z$. 
\end{prop}
\begin{proof}
    We proceed as in \cite[Proof of Theorem 3.1.2]{Young_thesis}.
    Let $\mathcal{G} = \set{\xi\in\dom(H^\Phi)\,\,:\,\, H^{\Phi} \xi = 0}$. 
    To prove the result, we show
    \begin{equation*}
        \inner{\xi}{H^\Phi\xi}
            \geq 
        \gamma' 
        \inner{\xi}{\xi}
    \end{equation*}
    for all $\xi\in\mathcal{G}^\perp$, where $\gamma' =  \limsup_{n\geq 1}
        \operatorname{spec-gap}(H^{\Phi,\Lambda_n})$.
    To do this, it suffices to prove this holds for a dense subset of $\mathcal{G}^\perp\cap \dom(H^\Phi)$. In this case, we may use the fact that $\dom(H^\Phi) \supset \pi(\A{\mathbb{Z}})\Xi$, and the latter is dense in $\mathcal{H}$. 
    Since $H^\Phi\geq 0$ and $H^\Phi\Xi = 0$, we have that 
    \begin{equation*}
        \mathcal{S}
            =
        \set{
            H^\Phi \pi(\bs{a})\Xi\,\,:\,\,\bs{a}\in\A{\Z}^{\operatorname{loc}}
        }
    \end{equation*}
    is dense in $\mathcal{G}^\perp$. 
    So, we prove $  \inner{\xi}{H^\Phi\xi}
            \geq 
        \gamma' 
        \inner{\xi}{\xi}$
    for all $\xi\in\mathcal{S}$. 
    To do this, in turn, it suffices to show that for all $\bs{a}\in\A{\Z}^{\operatorname{loc}}$, there is $N$ such that $n\geq N$ implies 
    \begin{equation*}
        \inner{\Xi}{\pi(\bs{a}^*)(H^\Phi)^3\pi(\bs{a})\Xi}
        \geq 
        \operatorname{spec-gap}(H^{\Phi,\Lambda_n})
         \inner{\Xi}{\pi(\bs{a}^*)(H^\Phi)^2\pi(\bs{a})\Xi}.
    \end{equation*}
    Now, since $\Psi$ is locally finite ranged, for all $\bs{a}\in\A{\Z}^{\operatorname{loc}}$, there exists $N\in\N$ such that $\delta_\Psi^2(\bs{a}), \delta_\Psi^3(\bs{a})\in\A{\Lambda_N}$. 
    Moreover, since $\psi$ is frustration-free, we have that $\pi(H^{\Phi,\Lambda_n})\Xi = 0$ for all $n\geq 1$. 
    So, for any $n\geq N$, we have that 
    \begin{align*}
        \inner{\Xi}{\pi(\bs{a}^*)(H^{\Phi})^3\pi(\bs{a})\Xi}
        &=
        \inner{\Xi}{\pi(\bs{a}^*)\pi(\delta_\Psi^3(\bs{a}))\Xi}\\
        &= 
         \inner{\Xi}
         {
            \pi(\bs{a}^*)
            \pi(H^{\Phi,\Lambda_n})^3
            \pi(\bs{a})
            \Xi}
            \\
        &= 
         \varphi\Big( 
            \bs{a}^* \big(H^{\Phi, \Lambda_n}\big)^3 \bs{a}
         \Big).
    \end{align*}

Using the fact that $H^{\Phi, \Lambda_n}$ is positive semi-definite, it follows trivially that 
\begin{align*}
         \inner{\Xi}{\pi(\bs{a}^*)(H^\Phi)^3\pi(\bs{a})\Xi}
         &=\varphi(
         \bs{a}^* \big(H^{\Phi, \Lambda_n}\big)^3 \bs{a})\\
         &\geq 
         \operatorname{spec-gap}(H^{\Lambda_n})
         \varphi(
         \bs{a}^* \big(H^{\Phi, \Lambda_n}\big)^2 \bs{a})
         =
         \operatorname{spec-gap}(H^{\Lambda_n})
          \inner{\Xi}{\pi(\bs{a}^*)(H^\Phi)^2\pi(\bs{a})\Xi},
    \end{align*} where we have used that $\delta_\Psi^2(\bs{a})\in\A{\Lambda_N}$.\qedhere
\end{proof}
%

\subsubsection{Random Interactions}
In this work, we consider random interactions, defined as follows. 
\begin{definition}[Random interactions]
    A \textit{random interaction} is a function $\Psi:\Omega\times\mathscr{P}_0(\Z)\to\A{\Z}^{\operatorname{loc}}$ such that for all $\omega\in\Omega$, $\Psi(\omega, \cdot) =: \Psi_\omega$ defines an interaction and for all $\Lambda\Subset\Z$, the map $\omega\mapsto \Psi_\omega(\Lambda)$ is Borel measurable with respect to $\mathcal{B}(\operatorname{wk}^*)$.
    We say a random interaction is locally finite range if $\Psi_\omega$ is locally finite range almost surely. 
    If $\Psi$ is a random interaction and $\psi$ is a random state, we say $\psi$ is a ground state (resp. frustration-free ground state) of $\Psi$ if $\psi_\omega$ is a ground state (resp. frustration-free ground state) of $\Psi_\omega$ for almost every $\omega\in\Omega$. 
    A random interaction $\Psi$ is called \textit{ergodic} if 
    \begin{equation}
    \label{eqn:ergodic_interaction}
        \Psi_\omega(\Lambda + k)
        =
        \tau_k\big(
            \Psi_{T^k(\omega)}(\Lambda)
        \big)
    \end{equation}
    for all $\Lambda\Subset\Z$. If $\Psi$ has bounded surface energy almost surely and if $\psi$ is a ground state of $\Psi$, we let $H^\Psi_\omega$ denote the GNS Hamiltonian corresponding to $\Psi_\omega$. 
\end{definition}
Arguing as in \cite[Theorem 4.6]{RoonSchenker_BGA} it is not hard to see the following. 
\begin{lemma}[Spectrum of GNS Hamiltonian is nonrandom]
\label{Lem:GNS_Ham_spec_nonrandom}
    Let $\Psi$ be an ergodic locally finite-ranged random interaction with ergodic ground state $\psi\in\cstates{\Omega, \A{\Z}}$. 
    Suppose that $\Psi_\omega$ has bounded surface energy almost surely.  
    Then there exists a deterministic set $\Sigma\subseteq\mathbb{R}^{\geq 0}$ such that 
    \begin{equation*}
        \sigma(H^{\Psi}_\omega) = \Sigma
    \end{equation*}
    for almost every $\omega\in\Omega$. 
\end{lemma}
In particular, given random interaction and state $(\Psi, \psi)$ as above, we see that the spectral $\gamma(H^\Psi_\omega)$ is a deterministic quantity, so we can say without ambiguity whether $(\Psi, \psi)$ is gapped. 
%
%
\section{Ergodic matrix product states}\label{sec:Vectorization}
In this section, we prove the main technical results that drive the rest of the paper. 
We say that a state $\psi\in\TCV$ satisfies \PureGena if the following condition is satisfied. 
\begin{description}
\hypertarget{PureGen}{}
    \item[\PureGen]
    There is $\BondDim\in\N$, $\hatrho\in\cstates{\Omega, \MatricesBond}$, and measurable $V:\Omega\to\operatorname{Lin}(\HilbertBond, \HilbertLocal\otimes\HilbertBond)$ such that if for all $a\in\MatricesLocal$ we define $E_a:\Omega\to\operatorname{Lin}(\MatricesBond)$ by 
    \begin{equation*}
        E_{a; \omega}(b)
            := 
        V_\omega^*(a\otimes b)V_\omega,
    \end{equation*}
    then $\hatrho_\omega\circ E_{\one_\LocalDim; T(\omega)} = \hatrho_{T(\omega)}$ almost surely, $V_\omega^* V_\omega = \one_{\LocalDim\BondDim}$ almost surely, and for all $[l, r]\subset \Z$ and $a_j\in\MatricesLocal$, 
    \begin{equation}\label{Eqn:PureGen_MPS}
        \psi_\omega(a_l\otimes\cdots\otimes a_r)
            =
        \hatrho_{T^{l - 1}(\omega)}\circ E_{a_l; T^l(\omega)}\circ\cdots\circ E_{a_r; T^r(\omega)}(\one_\BondDim)\,,
    \end{equation}
    holds for almost every $\omega\in\Omega$. 
\end{description}
\begin{notation}[Probabilists' notation]
\label{Notation:Probabilist}
    Given data as in \PureGen, we shall often suppress $\omega$ in our notation unless it is absolutely necessary. 
    To do this, we introduce the following notation. 
    For $x\in\Z$, we let $\hatrho_x\in\cstates{\Omega, \MatricesBond}$ be the random state defined by $\hatrho_{x; \omega} = \hatrho_{T^x(\omega)}$. 
    Similarly, for $a\in\MatricesLocal$ and $x\in\Z$, we let $E^{[x]}_a:\Omega\to\operatorname{Lin}(\MatricesBond)$ be the random completely positive map defined by $E^{[x]}_{a; \omega} = E_{a; T^x(\omega)}$; 
    more generally, for $\bs{a}\in\A{[l, r]}$ and $[l, r]\subset\Z$, we define $E^{[l, r]}_{\bs{a}}$ by
    \begin{equation*}
        E^{[l, r]}_{a_l\otimes\cdots\otimes a_r}
            :=
        E^{[l]}_{a_l}\circ\cdots\circ E^{[r]}_{a_r}.
    \end{equation*}
    More generally, for any random variable $F$ and $x\in\Z$, we shall often write $F_x$ to denote the random variable defined by $F_{x; \omega} = F_{T^x(\omega)}$, and then  manipulate $F_x$ as a random variable as opposed to evaluated pointwise. 
    So, for example, the condition on $\hatrho$ from \PureGena becomes 
    \begin{equation*}
        \hatrho_x \circ E^{[x + 1]}_{\localOne}
        =
        \hatrho_{x + 1}
    \end{equation*}
    holds almost surely. 
    Another important instance of this notation is that the ergodicity of a random state $\varphi\in\cstates{\Omega, \A{\Z}}$ may be phrased as $\varphi_0\circ\tau = \varphi_{1}$ almost surely.
    It may be useful for the reader to think of the subscript as denoting some site-dependent information associated to $\psi$. 
\end{notation}
In the following, $\set{\ket{f_k}}_{k=1}^\BondDim$ will be a fixed orthonormal basis of $\HilbertBond$. 
It is clear that $(a_l, \dots, a_r)\mapsto E_{a_l\otimes\cdots\otimes a_r}^{[l, r]}$ is multilinear in $(a_l, \dots, a_r)$, so if we take (\ref{Eqn:PureGen_MPS}) as a \textit{definition}, then we have that $\psi$ defines at least a measurable  map $\psi:\Omega\to\operatorname{Lin}(\A{\Z}, \C)$. 
Furthermore, we can define $E_{\bs{a}}^{[l, r]}$ for any $\bs{a}\in\A{[l, r]}$ by extending the definition $E_{a_l\otimes\cdots\otimes a_r}^{[l, r]}$ linearly. 
We now show, in fact, that such $\psi$ defines an ergodic state in addition to gathering some other basic facts about \PureGena into a lemma. 
\begin{lemma}\label{lem:generated}\label{Lem:Generated}
    Let $\BondDim$, $\hatrho$, $V$, and $E$ be as in \PureGen. 
    \begin{enumerate}[label = (\alph*)]
        \item $(T, E_{\one_\LocalDim})$ defines an ergodic quantum process on $\MatricesBond$.
        
        \item For all $\bs{a} \in \A{[l, r]}$, $\|E_{\bs{a}}^{[l, r]}\|\leq \|\bs{a}\|$ holds almost surely. 

        \item If we define $\psi$ by (\ref{Eqn:PureGen_MPS}), then $\psi\in\TCV$ as in definition~\ref{Def:TCV}.
    \end{enumerate}
\end{lemma}
\begin{proof}
    It is clear that $E_{\one_{\LocalDim}}$ is completely positive almost surely, and its unitality from the fact that $V$ is almost surely a partial isometry, so (a) holds. 
    To see (b), it is clear when $l = r$ by the fact that $V$ is an isometry, so assume $l < r$. 
    Let $\bs{a}\in\A{[l, r]}$ and write $\bs{a} = \sum_{x=j}^m \bs{a}_j$, where $\bs{a}_j = a_{l, j}\otimes \cdots \otimes a_{r, j}$ for some $a_{i, j}\in\MatricesLocal$. 
    Let also $\ket{\xi}, \ket{\eta}\in\HilbertBond$. 
    Then for all $b\in\MatricesBond$, we have 
    \begin{equation}\label{Eqn:Generated_Eqn_1}
        \bra{\xi} 
        E_{\bs{a}}^{[l, r]}(b)
        \ket{\eta}
            = 
        \sum_{j = 1}^{m}
        \bra{\xi}V_l^* 
        \left( 
            a_{l, j}\otimes E^{[l + 1, r]}_{\bs{a}_j'}(b)
        \right)
        V_{l}\ket{\eta}
    \end{equation}
    almost surely, 
    where $\bs{a}_j' = a_{l+1, j}\otimes \cdots \otimes a_{r, j}$ for all $j$. 
    Writing
    \begin{equation}\label{Eqn:Generated_Eqn_2}
        V_{l}\ket{\xi}
        =
        \sum_{x=1}^\LocalDim
        \sum_{y=1}^\BondDim
        \alpha_{x, y} \ket{e_x}\otimes\ket{f_y}
        \qquad
        \text{and}
        \qquad 
        V_{l}\ket{\eta}
        =
        \sum_{x=1}^\LocalDim
        \sum_{y=1}^\BondDim
        \beta_{x, y} \ket{e_x}\otimes\ket{f_y}
    \end{equation}
    for some (random) $\alpha_{x, y}, \beta_{x, y}\in\C$, (\ref{Eqn:Generated_Eqn_1}) becomes 
    \begin{align*}
        &
        \sum_{j=1}^m 
        \left[
        \sum_{w, x = 1}^{\LocalDim}
        \sum_{y, z = 1}^{\BondDim}
        \overline{\alpha}_{w, y}\beta_{x, z}
        \bra{e_w}\otimes\bra{f_y}
        \left( 
            a_{l, j}\otimes E^{[l + 1, r]}_{\bs{a}_j'}(b)
        \right)
        \ket{e_x}\otimes\ket{f_z}
        \right]\\
        &\phantom{\sum_{j=1}^m 
      \sum_{w, x = 1}^{\LocalDim}
        \sum_{y, z = 1}^{\BondDim}}= 
        \sum_{j=1}^m 
        \left[
        \sum_{w, x = 1}^{\LocalDim}
        \sum_{y, z = 1}^{\BondDim}
        \overline{\alpha}_{w, y}\beta_{x, z}
        \bra{e_w}
            a_{l, j}
        \ket{e_x}
        \bra{f_y}
            E^{[l + 1, r]}_{\bs{a}_j'}(b)
        \ket{f_z}
        \right]
    \end{align*}
    almost surely. 
    Iterating this argument, we conclude that 
    \begin{align*}
     &\bra{\xi} 
        E_{\bs{a}}^{[l, r]}(b)
        \ket{\eta}\\
        &\hspace{5mm}=
        \sum_{w_1, x_1 = 1}^{\LocalDim}\cdots 
        \sum_{w_s, x_s = 1}^{\LocalDim}
        \sum_{y_1, z_1 = 1}^{\BondDim}
        \cdots 
         \sum_{y_s, z_s = 1}^{\BondDim}
        \overline{\alpha}_{w_1, y_1, \dots, w_s, y_s}\beta_{x_1, z_1, \dots, x_s, z_s}
        \bra{e_{(w_1, \dots, w_s)}}
            \bs{a}
        \ket{e_{(x_1, \dots, x_s)}}
        \bra{f_{y_s}}
            b
        \ket{f_{z_s}}\,,
    \end{align*}
    where $s = r - l + 1$ and $\alpha_{w_1, y_1, \dots, w_s, y_s}, \beta_{x_1, z_1, \dots, x_s, z_s}\in\C$ are some constants arising as in (\ref{Eqn:Generated_Eqn_2}). 
    Therefore, since $V$ is almost surely a partial isometry, we conclude from Cauchy-Schwarz that 
    \begin{equation*}
        \left|
            \bra{\xi} 
        E_{\bs{a}}^{[l, r]}(b)
        \ket{\eta}
        \right|
        \leq 
        \|\bs{a}\|
        \|b\|
        \|\xi\|
        \|\eta\|\,,
    \end{equation*}
    almost surely, which shows (b). 
    To see (c), we first show that $\psi\in\cstates{\Omega, \A{\Z}}$, i.e., that $\psi$ defined by (\ref{Eqn:PureGen_MPS}) extends to a continuous linear functional on $\A{\Z}$. 
    This, however, follows from (b) and Takeda's theorem (Theorem \ref{Thm:Takeda}), since (b) shows precisely that $\A{[l, r]}\ni\bs{a}\mapsto \psi(\bs{a})$ is continuous almost surely. 
    The ergodicity of $\psi$ follows from the computation 
    \begin{align*}
        \psi_0\circ\tau (a_l\otimes \cdots \otimes a_r)
            &=
        \psi(\one_\LocalDim\otimes a_l\otimes\cdots\otimes a_r)\\
            &= 
        \hatrho_{l-1}\circ E^{[l]}_{\localOne}\circ E^{[l+1, r+1]}_{a_l\otimes\cdots\otimes a_r}(\bondOne)
            \\
            &= 
        \hatrho_{l}\circ E^{[l+1, r+1]}_{a_l\otimes\cdots\otimes a_r}(\bondOne)\\
            &= 
        \psi_{1}(a_l\otimes\cdots\otimes a_r),
    \end{align*}
    which concludes the proof. 
\end{proof}
\begin{definition}[Ergodic matrix product state]
\label{Def:Transfer_apparatus}
\label{Def:EMPS}
    If $\psi\in\TCV$ is of the form described by \PureGen, we call $\psi$ an \textit{ergodic matrix product state}, and we call $(T, E)$ the \textit{transfer apparatus} of $\psi$. 
    For $a\in\MatricesLocal$, we call $E_a:\Omega\to\operatorname{Lin}(\MatricesBond)$ a transfer operator, and in the special case where $a = \one_\LocalDim$, we call $(T, E_{\one_\LocalDim})$ the \textit{transfer apparatus ergodic quantum process}. 
\end{definition}
\begin{remark}
    Given \textit{any} collection $\set{\tilde{E}_{a}:\Omega\to\operatorname{Lin}(\MatricesBond)}_{a\in\MatricesLocal}$ such that $a\mapsto \tilde{E}_{a}$ is almost surely linear and $\tilde{E}_{\localOne}$ is almost surely ucp, it is generally true that there is $\hatrho\in\cstates{\Omega, \matrices}$ such that $\hatrho_{0}\circ \tilde{E}_{\localOne}^{[1]} = \hatrho_{1}$ almost surely \cite{Ekblad2024ReducibilityProcesses, Ekblad2026PeriodicityProcesses}, in which case one can define a state by (\ref{Eqn:PureGen_MPS}). 
    Therefore, the main assumption of \PureGena is the particular form that the maps $E$ take, i.e., that they are almost surely defined by partial isometries. 
\end{remark}
\begin{remark}
    In the language of \cite{FannesNachtergaeleWerner}, states satisfying \PureGena might be called purely generated ergodic $C^*$-finitely correlated states.
    We have chosen against using this terminology for simplicity. 
\end{remark}
Owed to the description of $E$ in terms of $V$, there is an explicit description of the transfer apparatus in terms of the bases $\set{\ket{e_k}}_{k=1}^\LocalDim$ and $\set{\ket{f_k}}_{k=1}^\BondDim$ which we make use of in the following.
With the notation as in \PureGen, for $j\in[n]$, we define $X^j:\Omega\to\MatricesBond$ by 
    \begin{equation}\label{Eqn:X_def}
        X^j_\omega 
            := 
        \sum_{x, y = 1}^\BondDim 
        \alpha_{y, x; \omega}^{(j)}
        \ketbra{f_y}{f_x}
        \qquad\text{where }
        \alpha_{y, x; \omega}^{(j)}
            :=
            (\bra{e_x}\otimes \bra{f_j})
            V_\omega 
            \ket{f_y}\,,
    \end{equation}
    for $\omega\in\Omega$, where we recall $\set{\ket{e_k}}_{k=1}^\LocalDim$ and $\set{\ket{f_k}}_{k=1}^\BondDim$ were our fixed orthonormal bases of $\HilbertLocal$ and $\HilbertBond$, respectively. 
    In other words, $X^j$ is the $j$th columnar block of $V$, i.e., 
    \begin{equation*}
        V 
            =
        \left[ 
            \begin{array}{c}
                 X^1 \\
                \vdots \\ 
                X^\LocalDim 
            \end{array}
        \right]\,,
    \end{equation*}
    where this block representation is with respect to the bases $\set{\ket{e_k}}_{k=1}^\LocalDim$ and $\set{f_k}_{k=1}^\BondDim$.
    More generally, for $[l, r]\subset\Z$ and $\sigma = (i_l, \dots, i_r)\in[\LocalDim]^{[l, r]}$, we define $X^{\sigma}:\Omega\to\MatricesBond$ by $X^\sigma = X^{i_r}_r\cdots X^{i_l}_l$, i.e., 
    \begin{equation}\label{Eqn:Xsigma_def}
        X^{\sigma}_\omega 
            := 
        X_{T^{r}(\omega)}^{i_r}\cdots X_{T^l(\omega)}^{i_l}\,.
    \end{equation}
    Notice that we have used the fact that $\sigma\in [\LocalDim]^{[l, r]}$ contains information about the interval $[l, r]$. 
    So, for example, by our definition, for any $[l, r]\subset\Z$, $k\in\Z$, $\sigma\in[\LocalDim]^{[l, r]}$, and $\varsigma\in[\LocalDim]^{[l + k, r + k]}$, we have that 
    \begin{equation*}
        X^{\sigma} 
        =
        X^{\varsigma}_{k}.
    \end{equation*}
    Given $l\leq s < s + 1\leq r$, if $\sigma_1\in[\LocalDim]^{[l, s]}$ and $\sigma_2\in [\LocalDim]^{[s + 1, r]}$, we write $\sigma_1\sigma_2$ to denote the concatenation, which is an element of $[\LocalDim]^{[l, r]}$, and we notice that 
    \begin{equation*}
        X^{\sigma_1\sigma_2} 
            =
        X^{\sigma_2} X^{\sigma_1}
    \end{equation*}
    holds almost surely. 
\begin{lemma}\label{Lem:Explicit_description_of_E_in_terms_of_X}\label{Lem:XPlicit}
    Fix $b\in\MatricesBond$. 
    For any $\bs{a}\in\A{[l, r]}$,
    \begin{equation*}
        E_{\bs{a}}^{[l, r]}(b)
            = 
        \sum_{\sigma,\varsigma\in[n]^{[l, r]}}
        \bra{e_\sigma}
        \bs{a}
        \ket{e_\varsigma}
        (X^{\sigma})^*
        b 
        X^{\varsigma} \,,
    \end{equation*}
    holds almost surely. 
\end{lemma}
\begin{proof}
    Because $\bs{a}\mapsto E^{[l, r]}_{\bs{a}}$ is linear in $\bs{a}$, it suffices to prove the lemma when $\bs{a} = a_l\otimes\cdots \otimes a_r$ is a pure tensor.
    In this case, the statement we want to prove becomes
    \begin{equation}\label{pf:Explicit:eq:simplified}
        E^{[l,r]}_{\bs{a}}(b) = \sum_{\sigma, \varsigma \in [n]^{[l,r]}} \prod_{x=l}^r \<e_{i_x}| a_x |e_{j_x}\> (X^\sigma)^* b X^\varsigma\,,
    \end{equation} 
    where we have written $\sigma = (i_l, \dots, i_r)$ and $\varsigma = (j_l, \dots, j_r)$. 
    We proceed by induction. 

    In the case when $l = r =:y$, the analog of \eqref{pf:Explicit:eq:simplified}, results from expressing $a_y = \sum_{i,j} \<e_i|a|e_j\>\,\cdot \ket{e_j}\bra{e_i}$ as a sum of its matrix entries relative to the basis $\{e_i\}$. In this case, using the block-matrix definition of $V$, we obtain 
    \[
        E^{[y]}_{a;\omega}(b) 
            = 
        V_{y}^*(a\otimes b)V_{y} 
            =
        \sum_{i,j = 1}^{\LocalDim}
        \bra{e_i} a_y\ket{e_j}
        V^*_y(\ket{e_i}\bra{e_j}\otimes b)V_y 
            =
        \sum_{i,j = 1}^{\LocalDim}
        \bra{e_i} a_y\ket{e_j}
        (X^i_{y})^* bX^{j}_{y}.
    \] 
    In the case with $l<r$, and $E^{[l,r]}_{\bs{a}}(b)$ with $\bs{a}$ a pure tensor as above, the following holds by construction 
    \[
        E^{[l,r]}_{\bs{a}}(b) = E^{[l]}_{a_l} \circ E^{[l+1, r]}_{\bs{a}'}(b) 
            =
        \sum_{i_l, j_l = 1}^{\LocalDim} 
        \<e_{i_l}|a_l|e_{j_l}\>\, 
        (X^{i_l}_l)^* b' X^{j_l}_l
    \] 
    where $\bs{a}' = a_{l+1} \otimes \cdots \otimes a_r$ and $b' =  E^{[l+1, r]}_{\bs{a}'}(b)$. 
    Now, \eqref{pf:Explicit:eq:simplified}  is immediate by induction, which concludes the proof. 
\end{proof}
\begin{example}[Disordered AKLT \texorpdfstring{\cite{RoonSchenker}}{l}]
\label{Example:DAKLT_introduction}
    The driving example in this work is the ergodic AKLT model described in \cite{RoonSchenker}, which is a deformation of the original AKLT model studied by Affleck, Kennedy, Lieb, and Tasaki in \cite{AKLT}.
    Here, the bond dimension $\BondDim$ is 2 and the local dimension $\LocalDim$ is 3, and the probability space $\Omega$ is $[0, 2\pi]^\Z$ with any measure $\Pr$ invariant and ergodic with respect to the bilateral shift. 
    Then for $\overline{\theta} = \seq{\theta_x}_{x\in\Z}\in\Omega$, we define $V_{\overline{\theta}}$ via its blocks 
    \begin{equation*}
         X^z_{\overline{\theta}} := 
        -\sin(\theta_0)\sigma^z
        ,\quad 
      X^+_{\overline{\theta}} := 
        -\cos(\theta_0)\sigma^+,
        \quad \text{and}\quad 
        X^-_{\overline{\theta}} := 
        \cos(\theta_0)\sigma^-
    \end{equation*}
    where 
    \begin{equation*}
        \sigma^{z}
        =
        \left[ 
            \begin{array}{cc}
               1  &0  \\
                0 & -1
            \end{array}
        \right],
        \quad 
        \sigma^{+}
        =
         \left[ 
            \begin{array}{cc}
                0 & 1  \\
                0 & 0
            \end{array}
        \right],
        \quad\text{and}\quad 
         \sigma^{-}
        =
         \left[ 
            \begin{array}{cc}
                0 & 0  \\
                1 & 0
            \end{array}
        \right].
    \end{equation*}
    Then, by Lemma \ref{Lem:XPlicit}, we have that 
    \begin{equation*}
        E^{[x]}_{a; \overline{\theta}}
        (b)
        =
        \sum_{j, k\in\set{z, +, -}}
        \bra{e_j}a\ket{e_k}
        X^{j*}_{\theta_x}b X^{k}_{\theta_x}
    \end{equation*}
    for any $x\in\Z$, 
    and so the disordered AKLT state is defined by 
    \begin{equation*}
        \psi_{\operatorname{DAKLT}; \overline{\theta}}(a_l\otimes\cdots\otimes a_r)
            :=
        \frac{1}{2}\Tr\seq{
           E^{[l, r]}_{a_l\otimes\cdots\otimes a_r; \overline{\theta}}(\one_2) 
        },
    \end{equation*}
    where we have noted that the unique random state $\hatrho\in\cstates{\Omega, \mathbb{M}_2}$ with $\hatrho_{\theta_0}\circ E_{\one_3; \theta_1} = \hatrho_{\theta_1}$ for almost every $\overline{\theta}$ is the deterministic completely mixed state $\hatrho = \frac{1}{2}\hat{\one}_2$ \cite{Ekblad2024ReducibilityProcesses}.
    It is worthwhile to note that, with respect to the orthonormal basis $\set{\one_2, \sigma^z, \sigma^+, \sigma^-}$ of $\mathbb{M}_2$, the transfer operator $E^{[x]}_{\one_3}$ is diagonalized
    \begin{equation*}
    E^{[x]}_{\one_3; \overline{\theta}}
        =
        \operatorname{diag}
        \Big(
            1, 
            -\cos 2\theta_x, 
            -\sin^2 \theta_x, 
            -\sin^2 \theta_x
        \Big),
    \end{equation*}
    so, in particular, 
    \begin{equation}\label{Eqn:Diagonalization_of_AKLT}
         E^{[l, r]}_{\one_3; \overline{\theta}}
            =
        (-1)^{r - l + 1}
         \operatorname{diag}
          \left(
             (-1)^{r - l + 1}, 
             \prod_{x=l}^r\cos 2\theta_x, 
             \prod_{x=l}^r\sin^2 \theta_x, 
             \prod_{x=l}^r\sin^2 \theta_x
        \right).
    \end{equation}
    From this, it is not hard to compute that the quantum correlations satisfy
    \begin{equation}\label{Eqn:Correlations_DAKLT}
         \big|\psi_{\operatorname{DAKLT}; \overline{\theta}}(S^{\#}\otimes\one_{[l, r]}\otimes S^{\#})\big|
            =
         \begin{cases}
             |\sin 2\theta_{l-1} |\,\displaystyle\seq{\prod_{x=l}^r\sin^2 \theta_x} \, |\sin 2\theta_{r+1} |
                &\#\in\set{+, -},
             \\
            \displaystyle
           \left| \frac{1+\cos 2\theta_{l-1} }{2}\right|\, \seq{\prod_{x=l}^r|\cos 2\theta_x|}\,\left| \frac{1+\cos 2\theta_{r+1} }{2}\right|\,
                &\# = z,
         \end{cases}
    \end{equation}
    where 
    \begin{equation*}
        S^z
        =
        \left[
            \begin{array}{ccc}
                 1& 0&0 \\
                 0&0&0\\
                 0&0&-1
            \end{array}
        \right], 
        \quad 
        S^+
        =
        \left[
            \begin{array}{ccc}
                 0&1&0 \\
                 0&0&1\\
                 0&0&0
            \end{array}
        \right], \quad\text{and}\quad 
        S^-
        =
        \left[
            \begin{array}{ccc}
                 0& 0&0 \\
                 1&0&0\\
                 0&1&0
            \end{array}
        \right].
    \end{equation*} 
    We use this fact later. 
\end{example}
Next, we recall the $\Gamma$-map formalism introduced in \cite{FannesNachtergaeleWerner}, and
we devote this subsection to describing this formalism as it pertains to the present disordered setting and establishing basic facts about it.  
\begin{definition}
    For an EMPS $\psi\in\TCV$, the tuple of random matrices $(X^j)_{j=1}^n$ from (\ref{Eqn:X_def}) is called the (random) \textit{tensor} associated to $\psi$, and we call $(X^j_x)_{j=1}^n$ the local tensor at $x\in\Z$. 
    For all $[l, r]\subset\Z$, define $\Gamma^{[l, r]}:\Omega\to\operatorname{Lin}(\MatricesBond, \Hilb{[l, r]})$ by 
    \begin{equation}\label{eqn:Gamma_maps}
        \Gamma^{[l, r]}(b)
            := 
        \sum_{\sigma\in[\LocalDim]^{[l, r]}}
        \trD{b X^\sigma}
        \ket{e_{\sigma}}
        \qquad\text{for all }b\in\MatricesBond\,.
    \end{equation}
    We call $\seq{\Gamma^{[l, r]}}_{[l, r]\subset\Z}$ the \textit{$\Gamma$-map formalism} associated to the data in \PureGen.
    We let $\mathcal{G}^{[l, r]}$ denote the (random) subspace of $\Hilb{[l, r]}$ defined by 
    \begin{equation*}
        \mathcal{G}^{[l, r]}
            :=
        \Gamma^{[l, r]}(\MatricesBond),
    \end{equation*}
    and we let $G^{[l, r]}\in\A{[l, r]}$ be the random projection $G^{[l, r]}:=\proj{ \mathcal{G}^{[l, r]}}$.
\end{definition}
Throughout the rest of this paper, $\psi$ will always denote a fixed EMPS, and we shall freely use the notation introduced above, i.e., $\rho$ will always refer to the random density matrix defined in \PureGen, $\seq{X^j}_{j=1}^\LocalDim$ will always refer to the tensor associated to $\psi$, etc.
The maps $\Gamma^{[l, r]}$ enable us to translate between dynamical properties of the transfer apparatus ergodic quantum process and the physical properties of the EMPS $\psi\in\TCV$, as we now describe. 
\begin{lemma}\label{lem:Gamma_covariant}\label{Lem:Gamma_covariant}
For $b\in\MatricesBond$, $[l, r]\subset\Z$, and $x, k\in\Z$, 
$ \tau_k\!\seq{
        \Gamma^{[l, r]}_x(b)
    }
        =
    \Gamma^{[l + k, r + k]}
    _{x - k}(b)$
almost surely. 
\end{lemma}
\begin{proof}
    This is immediate from the definitions. 
\end{proof}
\begin{lemma}\label{lem:basis}\label{Lem:Basis}
For all $\bs{a}\in\A{[l, r]}$ and $b_1, b_2\in\MatricesBond$, 
    \begin{equation*}
        \bra{\Gamma^{[l, r]}(b_1)}
            \bs{a}
        \ket{\Gamma^{[l, r]}(b_2)}
            =
        \sum_{x, y=1}^\BondDim
        \bra{f_x}
            E^{[l, r]}_{\bs{a}}
                \Big(
                    b_1^* \ketbra{f_x}{f_y}
                    b_2
                \Big)
        \ket{f_y}\,,
    \end{equation*}
    holds almost surely. 
\end{lemma}
\begin{proof}
We compute 
    \begin{align*}
         \bra{\Gamma^{[l, r]}(b_1)}
            \bs{a}
        \ket{\Gamma^{[l, r]}(b_2)}
            &= 
        \sum_{\sigma, \varsigma\in[\LocalDim]^{[l, r]}}
       \bra{e_\sigma}
        \bs{a}
       \ket{e_\varsigma}
            \overline{\trD{b_1 X^\sigma}}
            \trD{b_2 X^\varsigma}\,,
                \\
            &=
            \sum_{x, y = 1}^\BondDim
             \bra{f_x}
            \left[ 
          \sum_{\sigma, \varsigma\in[\LocalDim]^{[l, r]}}
          \bra{e_\sigma}
        \bs{a}
       \ket{e_\varsigma}
                X^{\sigma *}b_1^*
          \ketbra{f_x}{f_y}
                b_2 X^{\varsigma}
          \right]
                      \ket{f_y}.
    \end{align*}
    The result then follows from Lemma \ref{Lem:XPlicit}.
\end{proof}
\begin{lemma}\label{Lem:Local_action_of_psi}
    For $j, k\in[\BondDim]$, define $B^{j, k}:\Omega\to\MatricesBond$ by $ B^{j, k}
            := 
        \rho_{-1}^{1/2} \ketbra{f_j}{f_k}.$
    Then for any $\bs{a}\in\A{[l, r]}$, 
    \begin{equation*}
        \psi(\bs{a})
            =
            \sum_{j, k=1}^\BondDim 
            \bra{
                \Gamma^{[l, r]}\big(B^{j, k}_{l}\big)
            }
            \bs{a}
            \ket{
                \Gamma^{[l, r]}\big(B^{j, k}_{l}\big)
            }.
    \end{equation*}
    In particular,
    \begin{equation*}
        \psi(\bs{a}) = \psi(G^{[l, r]} \bs{a})
        =
        \psi(\bs{a}G^{[l, r]})
    \end{equation*}
    holds almost surely for all $\bs{a}\in\A{[l, r]}$.
\end{lemma}
\begin{proof}
    It is sufficient by Lemmas~\ref{lem:generated} and~\ref{Lem:XPlicit} to show the equation for only a single site $x\in \Z$. 
    Recall that by definition 
    \[
    E_{a}^{[x]}(b) = \sum_{i,j = 1}^{\LocalDim} \<i|a_x|j\>  (X^{j}_{x})^* b X^i_{x}
    \] 
    where $a_x\in\A{\{x\}}$. 
    By inserting the identity $\bondOne = \sum_{k=1}^{\BondDim} \ket{f_k}\bra{f_k}$ twice below, one finds 
    \begin{align*}
    \psi(a_x) &= \sum_{i,j = 1}^{\LocalDim}
    \bra{e_i}a_x\ket{e_j} \TrD{
        (X^{i}_{x}\rho_{x-1}^{1/2})^*X^{j}_{x}\rho_{x-1}^{1/2}
    }
        \\
    &= 
    \sum_{k,k' = 1}^{\BondDim} 
    \sum_{i,j = 1}^{\LocalDim}
        \bra{e_i}a_x\ket{e_j}
        \trD{
            \ketbra{f_k}{f_k}(X^{i}_{x}\rho_{x-1}^{1/2})^*\ketbra{f_{k'}}{f_{k'}}X^{j}_{x}\rho_{x-1}^{1/2}
        }
        \\
    &= 
    \sum_{k,l=1}^D
    \bra{\Gamma^{[x]}(B^{l, k}_x)} a_x \ket{\Gamma^{[x]}(B^{l, k}_x)}
    \end{align*}
    holds almost surely, 
    as claimed. 
\end{proof}
So, $\ker\psi\vert_{\A{[l, r]}}\subset G^{[l, r]\perp}\A{[l, r]}G^{[l, r]\perp}$ almost surely. 
Due to the freedom of choice in the length of the interval, it is clear that $\psi$ is \textit{a} frustration-free ground state of the random interaction $\Psi$ defined by letting $\Psi([l, r])$ be the projection onto the orthogonal complement of $\mathcal{G}^{[l, r]}$. 
In this case, however, $\psi$ need not be the \textit{unique} ground state. 
The standard way to remedy this for matrix product states is to make an assumption of \textit{injectivity} \cite{Perez-Garcia_et_al}. 
Here, this manifests as follows. 
\begin{definition}[Injectivity length]\label{Def:Injectivity_length}
    Given data as in \PureGen, the \textit{preliminary injectivity length} is the random variable $\ell_0 :\Omega\to \N\cup\{\infty\}$ defined by 
    \begin{equation*}
        \begin{split}
            \tilde{\ell}(\omega)
            &:= 
        \inf
        \set{
            k\in\Z
                \,\,:\,\,
            \Gamma^{[0, k]}_\omega\text{ is injective}
        }
        \end{split}
    \end{equation*}
    and the \textit{injectivity length} is the random variable $\ell:\Omega\to\N\cup\{\infty\}$ defined by 
    \begin{equation*}
       \begin{split}
            \ell(\omega)
            &:= 
            \inf 
        \set{k\in\Z
            \,\,:\,\,
        \Gamma^{[0, m]}_\omega\text{ is injective for all $m\geq k$}}
       \end{split}
    \end{equation*}
    where we take $\inf\emptyset := \infty$.
    As above, we use the notation $\ell_x$ to denote the random variable $\ell_{x}(\omega) = \ell(T^x(\omega))$ for $x\in\Z$. 
\end{definition} 
It is not immediately clear from the definitions that $\tilde{\ell} = \ell$, but we prove below that this is the case. 
This aside, we focus on $\ell$ and prove its basic properties.
\begin{lemma}
\label{Lem:Properties_of_inj_length}
\label{Lem:01_law_inj_length}
    $\Pr[\ell < \infty]\in\set{0, 1}$.
\end{lemma}
\begin{proof}
    By ergodicity of $T$, it suffices to show that $\ell_0 < \infty$ implies $\ell_1 < \infty$ almost surely.
    So, assume we are on the event $\ell<\infty$. 
    Then by Lemma \ref{Lem:Gamma_covariant}, we know that $\Gamma^{[1, 1 + \ell_1 + k]}$ is injective for all $k\in\N$. 
    Now, suppose $\Gamma^{[0, \ell_1 + k]}(b) = 0$ for some $b\in\MatricesLocal$.
    Then we have that 
    \begin{equation*}
        0
            =
        \sum_{j\in[\LocalDim]^{[0, 1)}}
        \sum_{\sigma\in[\LocalDim]^{[1, 1 + \ell_1 + k]}}
            \trD{X^j b X^\sigma}
                \ket{e_j}\otimes \ket{e_\sigma}.
    \end{equation*}
    So, for all $j\in[\LocalDim]$, we have
    \begin{equation*}
         0 
            =
         \sum_{\sigma\in[\LocalDim]^{[1, 1 + \ell_1 + k]}}\trD{X^j b X^\sigma}
            =
        \Gamma^{[1, 1 + \ell_1 + k]}(X^j b).
    \end{equation*}
    Thus, by the injectivity of $\Gamma^{[1, 1 + \ell_1 + k]}$, we have $X^j b = 0$ for all $j\in[\LocalDim]$. 
    Since $\sum_{j=1}^\LocalDim (X^j)^*X^j = \bondOne$, we conclude from this that 
    \begin{equation*}
        b 
        =
        \sum_{j=1}^\LocalDim (X^j)^*X^j b
        =
        0,
    \end{equation*}
    which shows $\Gamma^{[0, \ell_1 + k]}$ is injective and concludes the proof. 
\end{proof}
By the above lemma, we are justified in writing $\ell<\infty$ to mean that $\Pr[\ell < \infty] = 1$.
In the case that $\ell < \infty$, we say that $\psi$ has finite injectivity length, or simply that $\psi$ is injective.
There is a useful implication of this property for the transfer operator ergodic quantum process that we now describe.
Recall a linear map $\varphi:\matrices\to\matrices$ is called strictly positive if $\varphi(\rho)$ is invertible for all $\rho\in\states$. 
\begin{prop}
\label{Prop:Inj_iff_ESP}
\label{prop:ESPiffFiniteInjLength}
    If $\psi$ is injective, then the transfer operator ergodic quantum process satisfies the following condition. 
    \begin{description}
    \hypertarget{ESP}{}
    \item[\ESP] 
    For almost every $\omega\in\Omega$, there exists $M\in\N$ such that $m\geq M$ implies $E^{[0, m]}_{\localOne; \omega}$ is strictly positive.
    \end{description}
\end{prop}
\begin{proof}
    Since $\tilde{\ell}\leq \ell$, it suffices to show that $E^{[0, \tilde{\ell}]}_{\localOne}$ is strictly positive almost surely. 
    Now, notice that the injectivity of $\Gamma^{[0, \tilde{\ell}]}$ is equivalent to $\operatorname{span}\set{X^\sigma\,\,:\,\,\sigma\in[\LocalDim]^{[0, \tilde{\ell}]}} = \MatricesBond$. 
    Following the same proof as \cite[Proposition 1]{SanzPerezGarciaWolfCirac}, we conclude that $E^{[0, \tilde{\ell}]}_{\one}$ is strictly positive almost surely, which concludes the proof. 
\end{proof}
In fact, one can show that \ESPa is \textit{equivalent} to injectivity, a fact we prove in Appendix \ref{App:Injectivity}.
The condition \ESPa is called \textit{eventual strict positivity}, a condition which was studied at length in \cite{MovassaghSchenker}. 
A useful consequence of this condition comes from following theorem of \cite{MovassaghSchenker}.
\begin{thm}[\texorpdfstring{\cite[Theorem 1]{MovassaghSchenker}}{l}]\label{Thm:ESP}
    Let $(T, \phi)$ be an ergodic quantum process on $\matrices$.
    If $(T, \phi)$ satisfies \ESP, there is $\hat{\rho}\in\cstates{\Omega, \matrices}$, a universal constant $\mu\in (0, 1)$, and a measurable function $C:\Omega\to (0, \infty)$ such that $\rho$ is faithful almost surely, 
    \begin{equation}
        \hat{\rho}_0\circ\phi_{1} = \hat{\rho}_{1}\,,
    \end{equation}
    holds almost surely, and for all $[l, r]\subset\Z$ and $x\in[l, r]$, 
    \begin{equation}
        \left\|
            \phi_{}^{[l, r]}
            -
            \Delta_{r}
        \right\|
        \leq 
        C_{x} \mu^{r - l}
    \end{equation} 
    holds almost surely, where $\Delta_r:\Omega\to\operatorname{ucp}(\matrices)$ is defined by $\Delta_{r}(a) := \hatrho_{r}(a)\one_d$.
\end{thm}
In light of Proposition \ref{Prop:Inj_iff_ESP}, this theorem may be immediately understood as saying that a finite injectivity length yields relaxation to equilibrium in the long-time limit of the transfer operator ergodic quantum process associated to $\psi$ satisfying \PureGen, with a {site-dependent prefactor}. 
We now reformulate these properties about $(T, E)$ in terms of $\psi$ and $\Gamma$ more explicitly in the following series of lemmas. 
\begin{lemma}\label{Lem:hatrho_as_in_ESP}
    Assume $\psi$ has finite injectivity length. 
    Then $\hatrho\in\cstates{\Omega, \MatricesBond}$ satisfies the properties as in Theorem \ref{Thm:ESP} for the transfer operator ergodic quantum process.
    In particular, $\rho$ is faithful almost surely. 
\end{lemma}
\begin{proof}
    One immediately verifies that any state $\hatrho\in\cstates{\Omega, \matrices}$ as in Theorem \ref{Thm:ESP} is necessarily unique, which is all that is needed to prove the statement. 
\end{proof}
In accordance with this lemma and Theorem \ref{Thm:ESP}, for all $r\in\Z$, we define $\Delta_r:\Omega\to\operatorname{ucp}\!\seq{\MatricesBond}$ by 
\begin{equation*}
    \Delta_{r}(b)
    =
    \hatrho_{r}(b)\bondOne.
\end{equation*}
\begin{lemma}\label{lem:alpha_estimates}
    For $[l, r]\subset\Z$, we define $\alpha^{[l, r]}, \tilde{\alpha}^{[l, r]}:\Omega\to(0, \infty)$ by
    \begin{equation*}
            \alpha^{[l, r]}(\omega) 
                := 
            \trD{\rho_{r; \omega}^{-1}}
            \big\|
                E_{\localOne; \omega}^{[l, r]}
                -
                \Delta_{r; \omega}
            \big\|
            \quad 
            \text{
            and
            }
            \quad 
        \tilde{\alpha}^{[l, r]}(\omega)
            :=
        \BondDim^2 
         \big\|
                E_{\localOne; \omega}^{[l, r]}
                -
                \Delta_{r; \omega}
            \big\|
    \end{equation*}
    Assume $\psi$ has finite injectivity length. 
    Then the following hold. 
     \begin{enumerate}[label = (\alph*)]
        \item Let $C:\Omega\to (0, \infty)$ and $\mu\in (0, 1)$ be as in Theorem \ref{Thm:ESP} as it pertains to the transfer operator ergodic quantum process $(T, E_{\localOne})$. 
        Then for all $[l, r]\subset\Z$ and $x\in[l, r]$,
        \begin{equation*}
            \alpha^{[l, r]}\leq C_{x} \trD{\rho^{-1}_{r}} \mu^{r - l}
            \quad 
            \text{and}
            \quad 
            \tilde{\alpha}^{[l, r]}\leq C_x D^2 \mu^{r - l}
        \end{equation*}
        almost surely. 
        In particular, for all $l < r$, 
        \begin{equation*}
            \lim_{l\to - \infty} \alpha^{[l, r]} = \liminf_{r\to\infty}\alpha^{[l, r]} = 
            \lim_{l\to-\infty}\tilde{\alpha}^{[l, r]} = \lim_{r\to\infty}\tilde{\alpha}^{[l, r]} = 0
        \end{equation*}
        almost surely. 

        \item For all $b_1, b_2\in\MatricesBond$ and almost every $\omega\in\Omega$, 
        \begin{equation*}
            \Big| 
            \braket{\Gamma^{[l, r]}(b_1)}
            {\Gamma^{[l, r]}(b_2)}
                -
            \inner{b_1}{b_2}_{\rho_r}
        \Big|
        \leq 
            \min\Big( 
                \alpha^{[l, r]}
        \|b_1\|_{\rho_r}
        \|b_2\|_{\rho_r},
        \tilde{\alpha}^{[l, r]}
        \|b_1\|
        \|b_2\|
            \Big)
        \end{equation*}
        almost surely. 

        \item For all $b_1, b_2\in\MatricesBond$, $[s, t]\subset[l, r]$, and $\bs{a}\in\A{[s, t]}$,
        \begin{equation*}
          \begin{split}
                \Big| 
            \bra{\Gamma^{[l, r]}(b_1)}
                \bs{a}
            \ket{\Gamma^{[l, r]}(b_2)}
                -
            \psi(\bs{a})
            \inner{b_1}{b_2}_{\rho_r}
        \Big|&\\
        &\hspace{-25mm}\leq 
        \min\Big( 
            \big( 
            \alpha^{[l, s - 1]}
            +
            \alpha^{[t, r]}
        \big)
        \|\bs{a}\|
        \|b_1\|_{\rho_r}
        \|b_2\|_{\rho_r}, 
        \big( 
            \tilde{\alpha}^{[l, s - 1]}
            +
            \tilde{\alpha}^{[t, r]}
        \big)
        \|\bs{a}\|
        \|b_1\|
        \|b_2\|
        \Big)
          \end{split}
        \end{equation*}
        almost surely. 
    \end{enumerate}
\end{lemma}
Before we begin the proof, a remark is in order. 
\begin{remark}
    The reason we define two control functions $\alpha$ and $\tilde \alpha$ is twofold. 
    First, $\alpha$ is the natural candidate to make the machinery of the parent Hamiltonian proof work: namely, it appears in two crucial places in both the definition of $\kappa$ below, and our spectral gap estimates. Second, a corollary of the analogous result in \cite{FannesNachtergaeleWerner} is that the vector states given by $\Gamma^{[l,r]}(\one)$  weakly approximate the infinite volume state defined by the transfer apparatus. To achieve the first goal, it is necessary that a factor of $\rho_x$ appear, but this yields a pre-factor of $\trD{\rho_x^{-1}}$ which may become unbounded \textit{a priori}. However, switching to $\tilde \alpha$ prevents the key mechanism in the proof that the quantity $\kappa$ is monotone with respect to volumetric inclusions in Lemma~\ref{lem:injectivity_number} below. Since $\tilde \alpha$ controls the weak-$\ast$ convergence well, but $\alpha$ is more algebraically suitable for later purposes, we introduce both in the lemma statement. 
\end{remark}

\begin{proof}[Proof of Lemma~\ref{lem:alpha_estimates}]
    The bounds in (a)
    follow immediately from Theorem \ref{Thm:ESP}, and the three limits are then clear. 
    The limit infimum follows by Poincar\'e recurrence \cite[Theorem 1.4]{Walters1982AnTheory} for $T$, so (a) is proved. 
    To prove (b) and (c), we follow the method of \cite[Lemma 5.2]{FannesNachtergaeleWerner}.  
    By Lemma \ref{Lem:Basis}, 
    \begin{equation*}
        \braket{\Gamma^{[l, r]}(b_1)}
        {\Gamma^{[l, r]}(b_2)}
            =
        \sum_{x, y=1}^\BondDim
        \bra{f_x}
            E^{[l, r]}_{\localOne}
                \Big(
                    b_1^* \ketbra{f_x}{f_y}
                    b_2
                \Big)
        \ket{f_y}.
    \end{equation*}
    On the other hand, 
    \begin{equation*}
        \inner{b_1}{b_2}_{\rho_r}
            =
           \sum_{x, y=1}^\BondDim
           \bra{f_x}
        \Delta_{r}\big(b_1^* \ketbra{f_x}{f_y}
                    b_2\big)
                    \ket{f_y}\,.
    \end{equation*}
    Therefore, by Cauchy-Schwarz, 
    \begin{equation}\label{Eqn:Alpha_estimates_Eqn_1}
          \begin{split}
              \Big| 
            \braket{\Gamma^{[l, r]}(b_1)}
            {\Gamma^{[l, r]}(b_2)}
                -
            \inner{b_1}{b_2}_{\rho_r}
        \Big|
        &\leq
        \|
            E_{\localOne}^{[l, r]}
            -
            \Delta_{r}
        \|
        \seq{\sum_{x=1}^\BondDim
            \sqrt{
            \bra{f_x}b_1b_1^*\ket{f_x}}
        }
        \seq{\sum_{y=1}^\BondDim
            \sqrt{
            \bra{f_y}b_2^*b_2\ket{f_y}}
        },
          \end{split}
    \end{equation} 
    where we have used the fact that $\braket{f_z}{f_z} = 1$ for all $z$.
    From here, it is clear that 
    \begin{equation*}
          \Big| 
            \braket{\Gamma^{[l, r]}(b_1)}
            {\Gamma^{[l, r]}(b_2)}
                -
            \inner{b_1}{b_2}_{\rho_r}
        \Big|
        \leq
        \tilde{\alpha}^{[l,r]}\|b_1\|\|b_2\|
    \end{equation*}
    almost surely. 
    However, since (\ref{Eqn:Alpha_estimates_Eqn_1}) is independent of the choice of basis $\set{f_k}_{k=1}^\BondDim$ of $\HilbertBond$, we may assume that $\set{f_k}_{k=1}^\BondDim$ is a basis with respect to which $\rho_{r}\in\states$ is diagonalized, i.e., that $\rho_{r} = \sum_{k=1}^\BondDim \hatrho_{r}(\ketbra{f_k}{f_k}) \ketbra{f_k}{f_k}$.
    Recalling that $\rho_{r}$ is almost surely invertible, for any $b\in\MatricesBond$ Cauchy-Schwarz again yields 
    \begin{align*}
        \sum_{k=1}^\BondDim
            \sqrt{
            \bra{f_k}b^*b\ket{f_k}}
            &\leq 
        \seq{
            \sum_{k=1}^\BondDim
            \hatrho_{r}(\ketbra{f_k}{f_k})
            \bra{f_k}b^*b\ket{f_k}
        }^{1/2}
         \seq{
            \sum_{k=1}^\BondDim
            \hatrho_{r}(\ketbra{f_k}{f_k})^{-1}
        }^{1/2}\,,\\
        &= 
        \|
            b
        \|_{\rho_r}
        \trD{\rho_{r}^{-1}}^{1/2}.
    \end{align*}
    Therefore, (\ref{Eqn:Alpha_estimates_Eqn_1}) yields 
    \begin{equation*}
          \Big| 
            \braket{\Gamma^{[l, r]}(b_1)}
            {\Gamma^{[l, r]}(b_2)}
                -
             \inner{b_1}{b_2}_{\rho_r}
        \Big|
        \leq 
            \alpha^{[l, r]}
        \|b_1\|_{\rho_r}
        \|b_2\|_{\rho_r}
    \end{equation*}
    almost surely,
    so (b) is proved 
    To prove (c), we again start by Lemma \ref{Lem:Basis} with
    \begin{equation*}
         \bra{\Gamma^{[l, r]}(b_1)}
            \bs{a}
        \ket{\Gamma^{[l, r]}(b_2)}
        =
        \sum_{x, y=1}^\BondDim
         \bra{f_x}
            E^{[l, r]}_{\bs{a}}
                \Big(
                    b_1^* \ketbra{f_x}{f_y}
                    b_2
                \Big)
        \ket{f_y}.
    \end{equation*}
    Writing the image of $\bs{a}\in\A{[s, t]}$ in $\A{[l, r]}$ as $\bs{a} = \one_{[l, s - 1]}\otimes \bs{a}\otimes \one_{[t + 1, r]}$, we may rewrite this as 
    \begin{equation*}
         \bra{\Gamma^{[l, r]}(b_1)}
            \bs{a}
        \ket{\Gamma^{[l, r]}(b_2)}
        =
        \sum_{x, y=1}^\BondDim
         \bra{f_x}
            E^{[l, s - 1]}_{\localOne}
            \circ 
            E^{[s, t]}_{\bs{a}}
            \circ 
            E^{[t + 1, r]}_{\localOne}
                \Big(
                    b_1^* \ketbra{f_x}{f_y}
                    b_2
                \Big)
        \ket{f_y}.
    \end{equation*}
    By writing $E^{[l, s - 1]}_{\localOne}
            \circ 
            E^{[s, t]}_{\bs{a}}
            \circ 
            E^{[t + 1, r]}_{\localOne}
            =
            F_1 + F_2 + F_3$
    where 
    \begin{align*}
          F_1 
        &=\Delta_{s - 1; \omega}\circ  E^{[s, t]}_{\bs{a}}
            \circ \Delta_{r},\\
            F_2&=
            \seq{ E^{[l, s - 1]}_{\localOne} - \Delta_{s - 1}}\circ  E^{[s, t]}_{\bs{a}}
            \circ E^{[t + 1, r]}_{\localOne},\\
            F_3&=
            E^{[l, s - 1]}_{\localOne}\circ E^{[s, t]}_{\bs{a}} \circ \seq{E^{[t + 1, r]}_{\localOne} - \Delta_{r}},
    \end{align*}
    we obtain as in the proof of (b) that 
    \begin{equation*}
        \begin{split}
            \Big| 
            \bra{\Gamma^{[l, r]}(b_1)}
            \bs{a}
        \ket{\Gamma^{[l, r]}(b_2)}
            -
        \psi(\bs{a})\inner{b_1}{b_2}_{\rho_r}
        \Big|
        \leq 
        S_1 + S_2\,,
        \end{split}
    \end{equation*} 
    where
    \begin{equation*}
        \begin{split}
            S_1 &:= \min\Big( 
            \big( 
            \alpha^{[l, s - 1]}
            +
            \alpha^{[t, r]}
        \big)
        \|\bs{a}\|
        \|b_1\|_{\rho_r}
        \|b_2\|_{\rho_r}, 
        \big( 
            \tilde{\alpha}^{[l, s - 1]}
            +
            \tilde{\alpha}^{[t, r]}
        \big)
        \|\bs{a}\|
        \|b_1\|
        \|b_2\|
        \Big),\\
        S_2 &:= \left|\sum_{x, y = 1}^\BondDim
            \bra{f_x}
            F_1(b_1^* \ketbra{f_x}{f_x}b_2)\ket{f_y}
                -
            \psi(\bs{a})\inner{b_1}{b_2}_{\rho_r}
        \right|.
        \end{split}
    \end{equation*}
    Recalling the definition of $\psi$ in terms of \PureGen, we see that 
    \begin{align*}
        \sum_{x, y = 1}^\BondDim
            \bra{f_x}
            F_1(b_1^* \ketbra{f_x}{f_x}b_2)\ket{f_y}
            &=   
        \hatrho_{s-1}\!\seq{E^{[s, t]}_{\bs{a}}(\bondOne)}\inner{b_1}{b_2}_{\rho_r}\\
            &=
        \psi(\bs{a})\inner{b_1}{b_2}_{\rho_r}.
    \end{align*}
    Thus, $S_2 = 0$, which concludes the proof. 
\end{proof}
The above proof shows that 
\begin{equation}\label{Eqn:Lower_bound_inj_length}
    \left\|
        \Gamma^{[l, r]}(b)
    \right\|
    \geq 
    \|b\|_{\rho_r}
    \seq{1 - \alpha^{[l, r]}}^{1/2}
\end{equation}
almost surely for all $b\in\MatricesBond$. 
In particular, part (a)  offers a \textit{quantitative} proof of the fact that $\tilde{\ell}<\infty$ almost surely, since $\hatrho$ is faithful almost surely. 
Due to the possibility that $C$ and $\trD{\rho^{-1}}$ are unbounded above, however, we may \textit{not} use (\ref{Eqn:Lower_bound_inj_length}) to conclude a quantitative version of the fact that $\ell<\infty$, at least not \textit{a priori}.
There is nevertheless a way to achieve such a result. 
Towards this end, define for $[l, r]\subset\Z$ the quantity $\kappa^{[l, r]}:\Omega\to(0, \infty)$ by
\begin{equation}\label{Eqn:Def_of_kappa}
    \kappa^{[l, r]}(\omega)
        :=
    \inf_{b\in\MatricesBond\setminus\{0\}}
    \seq{\cfrac{
    \|
        \Gamma^{[l, r]}_\omega(b)
    \|
    }
    {
    \|b\|_{\rho_{r; \omega}}
    }}^2
    \,,
\end{equation}
where we note that $\hatrho$ is almost surely faithful and therefore $\|b\|_{\rho_r}\neq 0$ for all $b\neq 0$ and all $r\in \Z$ almost surely. 
\begin{lemma}\label{lem:injectivity_number}
\label{Lem:Kappa_lemma}
    Let $[l, r]\subset\Z$. 
    Then the following hold. 
    \begin{enumerate}[label = (\alph*)]
        \item $\kappa^{[l, r]}\geq 1 - \alpha^{[l, r]}$ almost surely. 

        \item $\kappa^{[l, r]}\leq \min\set{\kappa^{[l - 1, r]}, \kappa^{[l, r + 1]}}$ almost surely. 
        So, whenever $[s, t]\subset[l, r]$, $\kappa^{[s, t]}\leq \kappa^{[l, r]}$ almost surely.  

        \item Almost surely, $\Gamma^{[l, r]}$ is injective if and only if $\kappa^{[l, r]} > 0$. 
        In particular, $\tilde{\ell} = \ell$ almost surely. 
    \end{enumerate}
\end{lemma}
\begin{proof}
    (a) follows from (\ref{Eqn:Lower_bound_inj_length}). 
    To see (b), we first show $\kappa^{[l, r]}\leq \kappa^{[l - 1, r]}$. 
    Towards this end, let $b\in\MatricesBond$.
    Then 
    \begin{align*}
        \|
            \Gamma^{[l - 1, r]}(b)
        \|^2
            &= 
        \sum_{\sigma\in[\LocalDim]^{[l-1, r]}}
        \big| 
            \trD{b X^\sigma}
        \big|^2
            \\
            &=
        \sum_{i=1}^{\BondDim}
        \seq{
        \cfrac{
        \|
            \Gamma^{[l, r]}(X^{i}_{l-1}b)
        \|
        }{\|
            X^i_{l-1}b
        \|_{\rho_{r}}
        }}^2
        \|
         X^i_{l-1}b
        \|^2_{\rho_{r}}
            \\
            &\geq 
        \kappa^{[l, r]}
        \inner{b}{\left[\sum_{i=1}^\BondDim
        X_{l-1}^{i*}X_{l-1}^i\right]b}_{\rho_r}
            \\
            &=
        \kappa^{[l, r]}
       \|b\|_{\rho_r}^2,
    \end{align*}
    which follows from the assumption  $\Pr[\sum_{i=1}^\BondDim
        X_{l-1}^{i*}X_{l-1}^i = E_{\localOne}^{[l-1]}(\bondOne)  = \bondOne] =1$. 
    This shows $\kappa^{[l, r]}\leq \kappa^{[l - 1, r]}$ almost surely. 
    Next, we show that $\kappa^{[l, r]}\leq \kappa^{[l, r + 1]}$. 
    To see this, we compute 
    \begin{align*}
        \|
            \Gamma^{[l, r + 1]}(b)
        \|^2    
            = 
        \sum_{\sigma\in[\LocalDim]^{[l, r+1]}}
        \big| 
            \trD{b X^\sigma}
        \big|^2
            = 
        \sum_{i=1}^\BondDim
        \|
            \Gamma^{[l, r]}(b X^{i}_{r+1})
        \|^2.
    \end{align*}
    Therefore, 
    \begin{align*}
        \|
            \Gamma^{[l, r + 1]}(b)
        \|^2 
            \geq 
        \kappa^{[l, r]}
         \hatrho_{r}
         \!\seq{
         \sum_{i=1}^\BondDim
         X^{i*}_{r+1}(b^*b)X^{i}_{r+1}
         }
            &=  
        \kappa^{[l, r]}
         \hatrho_{r}
         \circ E_{\localOne}^{[r+1]}(b^*b)\,,
            \\
            &=
        \hatrho_{r+1}(b^*b),
    \end{align*}
    which shows $\kappa^{[l, r]}\leq \kappa^{[l, r + 1]}$ almost surely, concluding the proof of (b). 
    (c) is automatic from the definitions, so the proof of the lemma is done.
\end{proof}
%

\section{Parent Hamiltonians}\label{sec:PHams}
We now direct our attention to studying parent Hamiltonians of $\psi$ satisfying \PureGena with finite injectivity length, i.e., studying random interactions for which $\psi$ is (at the very least) a ground state. 
We shall see that the randomness inherent in our model causes a stark departure from the parent Hamiltonian constructed in \cite{FannesNachtergaeleWerner}. 
Namely, that there is not necessarily a \textit{finite-ranged} random interaction $\Psi$ for which $\psi$ is almost surely the ground state, at least not \textit{a priori}, which, as we shall see, is due to the possibility that $\ell\not\in L^\infty(\Omega)$. 
Nevertheless, we are able to show that there is a locally finite-ranged random interaction $\Psi$ for which $\psi_\omega$ is almost surely the unique ground state of $\Psi_\omega$, and, moreover, that $\Psi_\omega$ generates dynamics on the quasilocal algebra. 
To do this, we first define local Hamiltonians for which $\psi_\omega$ is clearly a ground state, then use this to produce a canonical interaction for which $\psi_\omega$ is a ground state.
To begin this discussion, we start with a motivating toy example that demonstrates it is not even necessary that $\ell\in L^1(\Omega)$. 
\begin{example}[Injectivity lengths defined by entry times]
\label{exmp:Unbounded_Injectivity}\label{Ex:Unbounded_inj}
    This example shows that we may define an ergodic matrix product state whose injectivity length is equal to the entry time into an arbitrary subset of an arbitrary dynamical system. 
    In particular, this gives a class of examples demonstrating that the injectivity length may be poorly-behaved in full generality. 
    Let $(\Omega, \mathcal{F}, \Pr)$ be any probability space, let $T:\Omega\to\Omega$ be an invertible measure preserving transformation, and let $A\subset\Omega$ be a measurable set with $\Pr[A]\in (0, 1)$. 
    Let $t_A:\Omega\to\N$ be the entry time function
    \begin{equation*}
        t_A(\omega)
            :=
        \min\set{k\in\N\,\,:\,\, T^k(\omega)\in A}.
    \end{equation*}
    Let $\BondDim\in\N$ be arbitrary, and let $\LocalDim \geq \BondDim^2$.
    For $i, j\in[\BondDim]$, let $B_{ij} = \ketbra{f_i}{f_j}$, and define $B_k = 0$ if $k > \BondDim^2$. 
    For $k\in[\LocalDim]$, define $X^k:\Omega\to\MatricesBond$ by 
    \begin{equation*}
        X^k(\omega)
            :=
        \begin{cases}
            B_k &\text{if }\omega\in A\\
            \bondOne &\text{if }\omega\in \Omega\setminus A.
        \end{cases}
    \end{equation*}
    Then it is clear that $\ell = t_A$ in this example. 
    There are many examples of invertible ergodic dynamical systems $T:\Omega\to\Omega$ for which $t_A \not\in L^1(\Omega)$ (see \cite{Haydn2013EntryDistribution}), therefore we have constructed a simple class of EMPS where $\ell$ may be unbounded. 
\end{example}
Therefore, the injectivity length $\ell$ may not only be unbounded, it need not be integrable. We may now state our main result on parent Hamiltonians. 
\begin{thm}[Theorem~\ref{Thm:Intro:PH}]\label{thm:InfinitePHam}
\label{Thm:Parent_Hamiltonian}
    Let $\psi$ be an injective ergodic matrix product state. 
    Then $\psi$ is almost surely pure.
    Moreover, there is an ergodic random locally finite-ranged interaction $\Psi$ with the following properties. 
    \begin{enumerate}[label = (\alph*)]
        \item 
        For all $\Lambda\Subset\Z$, $\Psi(\Lambda) \geq 0$ almost surely. 

        \item 
        $\psi$ is almost surely the unique frustration-free ground state of $\Psi$.

        \item 
        There is a deterministic constant $M>0$ and a random increasing and absorbing sequence $(\Lambda_n)_{n\in\N}$ such that 
        \begin{equation*}
            \left\|
                \sum_{\substack{\Lambda\cap\Lambda_n\neq\emptyset\\
                \Lambda\cap\Lambda_n^c\neq\emptyset}}
                \Psi(\Lambda)
            \right\|
            \leq 
            M
        \end{equation*} 
        almost surely. 
        In particular, $\Psi$ almost surely generates dynamics on the quasilocal algebra. 
    \end{enumerate}
\end{thm}
We devote the remainder of this section to proving this result in full. 
The issue of how poorly $\ell$ may behave becomes a technically challenging issue when following the template laid out in \cite{FannesNachtergaeleWerner} for defining parent Hamiltonians of matrix product states. 
Nevertheless, the first step in the construction of \cite{FannesNachtergaeleWerner} remains unchanged:
for all $[l, r]\subset\Z$, recall $\mathcal{G}^{[l, r]}\subset \mathcal{H}^{[l, r]}$ was defined to be the random subspace  
\begin{equation}\label{Eqn:Def_of_mcG}
    \mathcal{G}^{[l, r]} 
        =
    \operatorname{ran}\Gamma^{[l, r]},
\end{equation}
and $G^{[l, r]}:\Omega\to\A{[l, r]}$ was the random projection defined by $G^{[l, r]} := \operatorname{proj} \mathcal{G}^{[l, r]}$, i.e., $G^{[l, r]}$ is the orthogonal projection onto $\mathcal{G}^{[l, r]}$. 
Define now 
\begin{equation*}
    K^{[l, r]} 
        :=
    \one_{[l, r]} - G^{[l , r]}.
\end{equation*}
From Lemma \ref{Lem:Local_action_of_psi}, we know that $\psi(K^{[l, r]}) = 0$ almost surely for all $[l, r]\subset\Z$, which gives us a template for constructing the random interaction of Theorem \ref{Thm:Parent_Hamiltonian}. 
Indeed, it is already clear that if we define $\Psi([l, r]):= K^{[l, r]}$, then $\psi$ is a frustration-free ground state of $\Psi$. 
However, in order to achieve uniqueness of this ground state property, we must take care in tuning the \textit{scale} of our random interaction. 
Towards this, we make the following definition. 
\begin{definition}[Injectivity sequence]
\label{Def:Inj_sequence}
    For $x\in\Z$, recall $\ell_x:\Omega\to\N$ is defined by $\ell_x(\omega):= \ell(T^x(\omega))$. 
    A sequence of measurable functions $(\mathcal{L}_x:\Omega\to \N)_{x\in\Z}$ is called an \textbf{injectivity sequence} if the following conditions hold. 
    \begin{enumerate}[label = (\alph*)]
        \item For all $x\in\Z$, $\mathcal{L}_x\geq \ell_x$ almost surely.

        \item For all $[l, r]\subset\Z$, $\max\set{x + \mathcal{L}_x\,\,:\,\,x\in[l, r]}\geq r + 1 + \ell_{r + 1}$ almost surely. 
    \end{enumerate}
\end{definition}
It is not \textit{a priori} clear that an injectivity sequence exists. 
Towards this end, we make the following definition. 
\begin{definition}
    If $\eta\in\N$ is such that $\Pr[\ell\leq\eta] > 0$, we call $\eta$ \textbf{admissible}. 
    For such $\eta$, we let $k^\eta:\Omega\to\N$ denote the entry time function 
    \begin{equation*}
        k^\eta(\omega)
            :=
         \min\set{k\in\N\,\,:\,\, T^k(\omega)\in\set{\ell\leq \eta}},
    \end{equation*}
    which is almost surely finite by Poincar\'e recurrence. 
    For any admissible $\eta\in\N$ and $x\in\Z$, define $k^\eta_x(\omega):= k^\eta(T^x(\omega))$. 
\end{definition}
\begin{lemma}
\label{Lem:kx_lemma}
    For any admissible $\eta\in\N$ and $x\in\Z$, $\ell_x \leq k_x^\eta + \eta$ almost surely.  
\end{lemma}
\begin{proof}
    If $\ell_x(\omega) < k_x^\eta(\omega) + \eta$, there is nothing to prove, so suppose that $\ell_x(\omega)\geq k_x^\eta(\omega) + \eta$. 
    Then it holds that $\Gamma^{[x + k_x^\eta(\omega), x + k_x^\eta(\omega) + \eta]}_\omega$ is injective. 
    So, by Lemma \ref{Lem:Kappa_lemma}, we have that $\Gamma^{[x, x + k_x^\eta(\omega) + \eta]}_\omega$ is injective. 
    Therefore, $\ell_x(\omega) = k_x^\eta(\omega) + \eta$, as desired. 
\end{proof}
Motivated by this lemma, we define for $\eta\in\N$ with $\Pr[\ell \leq \eta]>0$ and $x\in\Z$ the function 
 \begin{equation}
 \label{Eqn:Canonical_inj_sequence}
        \mcL^\eta_x:\Omega\to\N,\quad\mcL^\eta_x(\omega) := k_x^\eta + \eta,
\end{equation}
which by the above lemma satisfies the first condition required of an injectivity sequence. 
We now show $(\mcL^\eta_x)_{x\in\Z}$ is, in fact, an injectivity sequence. 
\begin{lemma}
\label{Lem:Canonical_inj_sequence}
    For any admissible $\eta\in\N$, $\seq{\mcL_x^\eta}_{x\in\Z}$ is an injectivity sequence. 
\end{lemma}
\begin{proof}
    Condition (a) in Definition \ref{Def:Inj_sequence} holds by Lemma \ref{Lem:kx_lemma}, so we only need to show condition (b) in Definition \ref{Def:Inj_sequence} holds. 
    To do this, observe that $k_x^\eta \leq k_{x+1}^\eta + 1$, since $\ell_{x + k^\eta_{x+1} + 1}\leq \eta$. 
    Therefore, $x + \mcL^\eta_x\leq x + 1 + \mcL^\eta_{x+1}$ holds almost surely for all $x\in\Z$.
    So, for any $[l, r]\subset\Z$, we have 
    \begin{equation*}
        \max
        \set{x + \mcL^\eta_x\,\,:\,\,x\in[l, r]}
            =
        r + \mcL^\eta_r.
    \end{equation*}
    Therefore, we just need to show that $k_r^\eta - 1 + \eta \geq \ell_{r+1}$ holds almost surely.
    This, however, follows from Lemma \ref{Lem:Kappa_lemma} upon noting that $\Gamma^{[r, r + k^\eta_r + \eta]}$ is injective, therefore $\Gamma^{[r + 1, r + 1 + (k_r^\eta - 1 + \eta)]}$ is injective, concluding the proof. 
\end{proof}
\begin{definition}[Canonical injectivity sequence]
    For admissible $\eta\in\N$, 
    we call $\mathscr{L}_\eta = (\mcL^\eta_x)_{x\in\Z}$ defined in (\ref{Eqn:Canonical_inj_sequence}) the \textbf{canonical injectivity sequence of radial tolerance} $\eta$. 
    We write $r_\eta:\Z\to\Z$ to denote the corresponding random function $r_\eta(x) = x + \mathcal{L}_x^\eta$. 
\end{definition}
It is worth recording the following fact which is a corollary of the proof of Lemma \ref{Lem:Canonical_inj_sequence}.
\begin{lemma}
\label{Lem:r_eta_strictly_increasing}
    For any admissible $\eta$, $r_\eta$ is nondecreasing almost surely. 
\end{lemma}
Now that we have established the existence of injectivity sequences, we prove that they describe a scale at which the following intersection property holds. 
\begin{prop}[Intersection property of injectivity sequences]
\label{Prop:Frustration_free_infinitely_often}
\label{Prop:Intersection_property}
    Let $(\mcL_x:\Omega\to \N)_{x\in\Z}$ be any injectivity sequence and fix $l\in\Z$. 
    For $k\geq l$, define $M_k:\Omega\to\N$ by
    \begin{equation*}
        M_k(\omega) := \max\set{x + \mcL_x(\omega)\,\,:\,\, x\in[l, k]}.
    \end{equation*}
    Then for all $k\in\N$, 
    \begin{equation}\label{Eqn:Frustration_free_infinitely_often}
        \mathcal{G}^{[l, M_k]} 
            =
        \bigcap_{x\in[l, k]}
         \mathcal{H}^{[l, x)}
                \otimes 
            \mathcal{G}^{[x, x + \mcL_x]} 
                \otimes 
            \mathcal{H}^{(x + \mcL_x, M_k]}\,,
    \end{equation}
    holds almost surely. 
\end{prop}
\begin{proof}
    Let $r(x) := x + \mathcal{L}_x$. 
    We prove the statement for $l = 0$ for simplicity. 
    Write $\tilde{\mathcal{G}}^k$ to denote the random subspace on the right-hand side of (\ref{Eqn:Frustration_free_infinitely_often}).
    We first show $\mathcal{G}^{[0, M_{k}]}  \subset \tilde{\mathcal{G}}^k$. 
    To see this, let $v\in  \mathcal{G}^{[0, M_{k}]}$. 
    Then write 
    \begin{equation*}
        v
            =
        \sum_{\sigma\in[\LocalDim]^{[0, M_{k}]}}
            \trD{b X^{\sigma}}
            \ket{e_{\sigma}}
    \end{equation*}
    for some $b\in\MatricesBond$.
    Then by cyclicity of trace, for any $x\in[0, k]$, we have that 
    \begin{equation*}
        v
            =
         \sum_{
         \substack{
            \sigma'\in[\LocalDim]^{[0, x]}\\
            \sigma''\in[\LocalDim]^{(r(x), M_k]}
            }
         }
         \ket{e_{\sigma'}}
         \otimes 
         \left[
        \sum_{\sigma\in[\LocalDim]^{[x, r(x)]}}
            \trD{b^{\sigma', \sigma''} X^\sigma}
            \ket{e_\sigma}
            \right] 
        \otimes 
        \ket{e_{\sigma''}}
    \end{equation*}
    where $b^{\sigma', \sigma''} = X^{\sigma'} b X^{\sigma''}$. 
    From this, it is clear that $v\in  \tilde{\mathcal{G}}^k$, hence $\mathcal{G}^{[0, M_{k}]}  \subset \tilde{\mathcal{G}}^k$ almost surely. 
    Next, we show the other inclusion. 
    To do this, we induct on $k$. 
    It is clear when $k = 0$, so assume $k \geq 1$. 
    Then we have that 
    \begin{equation*}
        \tilde{\mathcal{G}}^{k+1}
            =
        \Big[ 
            \mathcal{G}^{[0, M_k]} 
            \otimes 
            \mathcal{H}^{(M_k, M_{k+1}]}
        \Big]
        \cap 
        \Big[ 
            \mathcal{H}^{[0, k+1)}
            \otimes 
            \mathcal{G}^{[k + 1, r(k+1)]}
            \otimes 
             \mathcal{H}^{(r(k+1), M_{k+1}]}
        \Big].
    \end{equation*}
    There are now two cases: either $M_k\geq r(k+1)$, or $M_k < r(k+1)$.
    In the first case, $M_k\geq r(k+1)$, we have that $M_k = M_{k+1}$, so we just need to show that 
    \begin{equation*}
     \mathcal{G}^{[0, M_{k+1}]}
     \cap 
     \left[
         \mathcal{H}^{[0, k+1)}
            \otimes 
            \mathcal{G}^{[k + 1, r(k+1)]}
            \otimes 
             \mathcal{H}^{(r(k+1), M_{k}]}
        \right]
        =
         \mathcal{G}^{[0, M_{k+1}]}.
    \end{equation*}
    This, however, follows by the argument above, since that showed 
    \begin{equation*}
        \mathcal{G}^{[0, M_{k+1}]}\subset  \mathcal{H}^{[0, k+1)}
            \otimes 
            \mathcal{G}^{[k + 1, r(k+1)]}
            \otimes 
             \mathcal{H}^{(r(k+1), M_{k+ 1}]}.
    \end{equation*}
    So, we may now assume we are in the second case, i.e., that $M_k < r(k+1)$.
    In this case, we have that $M_{k + 1} = r(k+1)$. 
    So, we have that 
   \begin{equation*}
        \tilde{\mathcal{G}}^{k+1} 
            =
        \Big[ 
            \mathcal{G}^{[0, M_k]} 
            \otimes 
            \mathcal{H}^{(M_k, M_{k+1}]}
        \Big]
        \cap 
        \Big[ 
            \mathcal{H}^{[0, k+1)}
            \otimes 
            \mathcal{G}^{[k + 1, M_{k+1}]}
        \Big],
    \end{equation*}
    and we want to show that $\tilde{\mathcal{G}}^{k+1}  \subset  \mathcal{G}^{[0, M_{k + 1}]}$.
    Towards this end, let $v\in \tilde{\mathcal{G}}^{k+1}$.
    Let $I_1 = [0, M_k]$, $I_2 = (M_k, M_{k+1}]$, $J_1 = [0, k+1)$, and $J_2 = [k + 1, M_{k+1}]$. 
    Then there are $b_{\sigma_2}\in\MatricesBond$ for all $\sigma_2\in[\LocalDim]^{I_2}$ so that 
     \begin{equation*}
        v = \sum_{
            \substack{
            \sigma_1\in [\LocalDim]^{I_1}\\
            \sigma_2\in [\LocalDim]^{I_2}
            }
        }
           \trD{b_{\sigma_2}X^{\sigma_1}}
           \ket{e_{\sigma_1}}\otimes \ket{e_{\sigma_2}}\,,
    \end{equation*}
    and, similarly, there are $c_{\varsigma_1}\in\MatricesBond$ for all $\varsigma_1\in[\LocalDim]^{J_1}$ so that 
    \begin{equation*}
        v = \sum_{
            \substack{
            \varsigma_1\in [\LocalDim]^{J_1}\\
            \varsigma_2\in [\LocalDim]^{J_2}
            }
        }
           \trD{c_{\varsigma_1}X^{\varsigma_2}}
           \ket{e_{\varsigma_1}}\otimes \ket{e_{\varsigma_2}}\,.
    \end{equation*}
    Since these expressions are equal, if we let $\sigma_1\sigma_2$ and $\varsigma_1\varsigma_2$ denote the concatenations, we have that 
    \begin{equation*}
        \trD{b_{\sigma_2}X^{\sigma_1}} - \trD{c_{\varsigma_1}X^{\varsigma_2}}
            =
        0
    \end{equation*}
    whenever $\sigma_1\sigma_2 = \varsigma_1\varsigma_2$.  
    Writing $\sigma_1\sigma_2 = \varsigma_1\varsigma_2 = (i_0, \dots, i_{M_{k+1}})$, we have that
    \begin{equation*}
        \begin{split}
            \sigma_j 
                &=
            \begin{cases}
                (i_0, \dots, i_{M_k}) &j = 1\\
                (i_{M_{k} + 1}, \dots, i_{M_{k+1}}) &j = 2
            \end{cases} \\
            \varsigma_j
                &= 
            \begin{cases}
                 (i_0, \dots, i_{k}) &j = 1\\
                (i_{k + 1}, \dots, i_{M_{k+1}}) \phantom{...} &j = 2,
            \end{cases}
        \end{split}
    \end{equation*}
    so by cyclicity of trace, 
    \begin{equation*}
        \trD{
        \left( 
            X^{\varsigma_1} b_{\sigma_2}
            -
            c_{\varsigma_1}
            X^{\sigma_{2}}
        \right)
        X^{(i_{k + 1}, \dots, i_{M_{k}})}
        }   
        =
        0.
    \end{equation*}
    Now, note that $\varsigma_1$ and $\sigma_2$ are independent of $(i_{k + 1}, \dots, i_{M_{k}})$.
    In particular, since the above equation holds for all $(i_{k + 1}, \dots, i_{M_{k}(\omega)})\in[\LocalDim]^{[k + 1, M_k]}$, we have that 
    \begin{equation*}
        \Gamma^{[k+1, M_{k}]}\!\seq{ X^{\varsigma} b_{\sigma}
            -
            c_{\varsigma}
            X^{\sigma}}
                =
                0\,,
    \end{equation*}
    almost surely 
    for any $\varsigma\in[\LocalDim]^{J_1}$ and $\sigma\in[\LocalDim]^{I_2}$.
    Since $M_{k}\geq r(k+1)$,  however, this implies that
    \begin{equation}\label{Eqn:Frustration_free_infinitely_often_Eqn_1}
            X^{\varsigma} b_{\sigma}
            -
            c_{\varsigma}
            X^{\sigma}
            =
            0\,,
    \end{equation}
    almost surely 
    for all $\varsigma\in[\LocalDim]^{J_1}$ and $\sigma\in[\LocalDim]^{I_2}$.
    Now, since since $E^{[l, r]}_{\localOne}(\bondOne) = \bondOne$ almost surely for all $[l, r]\subset\Z$, we have that 
    \begin{equation*}
        \bondOne
            =
        \sum_{\varsigma\in[\LocalDim]^{J_1}}
         X^{\varsigma*}
        X^{\varsigma}.
    \end{equation*}
    Therefore, from (\ref{Eqn:Frustration_free_infinitely_often_Eqn_1}) we conclude that 
    \begin{equation*}
        b_\sigma 
            =
        \sum_{\varsigma\in[\LocalDim]^{J_1}}
         X^{\varsigma*}
        X^{\varsigma}b_\sigma 
            =
        \sum_{\varsigma\in[\LocalDim]^{J_1}}
         X^{\varsigma*}
         c_{\varsigma}
        X^{\sigma}
            =
        \seq{
        \sum_{\varsigma\in[\LocalDim]^{J_1}}
         X^{\varsigma*}
         c_{\varsigma}
         }
         X^{\sigma}
    \end{equation*}
    for all $\sigma\in[\LocalDim]^{I_2}$.
    So, if we let $b' := \sum_{\varsigma\in[\LocalDim]^{J_1}}
         X^{\varsigma*}
         c_{\varsigma}$, then we see 
    \begin{align*}
        \Gamma^{[0, M_{k+1}]}(b')
            &=
        \sum_{\xi\in[\LocalDim]^{[0, M_{k+1}]}}
        \trD{b'
            X^{\xi}
        }
        \ket{e^{[0, M_{k+1}]}_\xi}\,,\\
            &= 
        \sum_{\substack{
        \sigma_1\in[\LocalDim]^{I_1}\\
        \sigma_2\in[\LocalDim]^{I_2}
        }}
        \trD{b'
            X^{\sigma_1\sigma_2}
        }
        \ket{e^{[0, M_{k+1}]}_{\sigma_1\sigma_2}}\,,\\
        &= 
       \sum_{\substack{
        \sigma_1\in[\LocalDim]^{I_1}\\
        \sigma_2\in[\LocalDim]^{I_2}
        }}
        \trD{b_{\sigma_2}
            X^{\sigma_1}
        }
        \ket{e^{I_1}_{\sigma_1}}\otimes \ket{e^{I_2}_{\sigma_2}}\,,\\
        &= 
        v,
    \end{align*}
    hence $v\in\tilde{\mathcal{G}}^{k+1}$, which ends the proof. 
\end{proof}
This motivates the following definition. 
\begin{definition}[Interactions exposing $\psi$]
\label{Def:Interaction_exposing_psi}
    We say a random interaction $\Psi$ exposes $\psi$ if $\Psi(\Lambda)\geq 0$ almost surely for all $\Lambda\Subset\Z$ and if there is an injectivity sequence $(\mcL^\Psi_x)_{x\in\Z}$ associated to $\Psi$ satisfying the following conditions. 
    \begin{enumerate}[label = (\alph*)]
        \item For all $x\in\Z$,
        \begin{equation}
            \operatorname{ker}\Psi_\omega\big([x, x + \mcL^\Psi_x(\omega)]\big) = \mcG^{[x, x + \mcL^\Psi_x(\omega)]}_\omega
        \end{equation}
        holds for almost every $\omega\in\Omega$. 

        \item For almost every $\omega\in\Omega$, if $\Lambda\not\in \set{[x, x + \mcL^\Psi_x(\omega)]\,\,:\,\, x\in\Z}$, then $\Psi_\omega(\Lambda) = 0$. 

        \item If for all $[l, r]\subset\Z$ we define $M^{[l, r]}(\omega):=\max\set{x + \mcL^\Psi_x(\omega)\,\,:\,\, x\in[l, r]}$, then
        \begin{equation}\label{Eqn:Nondegenerate_inj_sequence}
            \liminf_{k\to\infty}\trD{\rho^{-1}_{M^{[l, r]}}} < \infty
        \end{equation}
        holds almost surely. 
    \end{enumerate}
\end{definition}
In light of the intersection property of injectivity sequences established in Proposition \ref{Prop:Frustration_free_infinitely_often}, we may interpret conditions (a) and (b) as saying that our interaction $\Psi$ satisfies a corresponding intersection property. 
Condition (c) requires that the injectivity sequence $(\mcL^\Psi_x)_{x\in\Z}$ be ``ergodic enough". 
We now show that the canonical injectivity sequence $(\mcL^\eta_x)_{x\in\Z}$ satisfies condition (c), thereby giving us our first example of a random interaction exposing $\psi$. 
\begin{lemma}
\label{Lem:Canonical_inj_sequence_nondegen}
    Fix $\eta\in\N$ admissible. 
    Then $(\mcL^\eta_x)_{x\in\Z}$ satisfies (\ref{Eqn:Nondegenerate_inj_sequence}).
\end{lemma}
\begin{proof}
    For $\eta\in\N$, let $A_\eta = \set{\ell\leq \eta}$, and for $M>0$ let $\tilde{A}_M = \set{\trD{\rho^{-1}}\leq M}$. 
    Fix $\eta\in\N$ admissible, and let $M>0$ be large enough so that $\Pr[\tilde{A}_M] > 1 - \Pr[A_\eta]$.
    Then because $T$ is a measure preserving transformation, $\Pr[A_\eta] > 1 - \Pr[T^{-\eta}(\tilde{A}_M)]$, hence $\Pr[A_\eta\cap T^{-\eta}(\tilde{A}_M)] > 0$. 
    Thus, by Poincar\'e recurrence, we know that 
    \begin{equation}
    \label{Eqn:Canonical_inj_sequence_nondegen:Eqn_1}
        \#\set{k\in\N\,\,:\,\, 
        T^k(\omega)\in 
        A_\eta\cap T^{-\eta}(\tilde{A}_M)
        }
        =
        \infty
    \end{equation}
    holds for almost every $\omega\in\Omega$. 
    However, from Lemma \ref{Lem:r_eta_strictly_increasing}, we know that $M^{[l, r]} = r + \mcL^\eta_r = r + \eta + k_r^\eta$ almost surely. 
    But notice that 
    \begin{equation*}
        k_r^\eta(\omega)
        =
        \min\set{k\in\N\,\,:\,\, T^{r + k}(\omega)\in A_\eta}.
    \end{equation*}
    Therefore, from (\ref{Eqn:Canonical_inj_sequence_nondegen:Eqn_1}), we are able to conclude that (\ref{Eqn:Nondegenerate_inj_sequence}) holds almost surely, which concludes the proof. 
\end{proof}
\begin{cor}
\label{Cor:Canonical_interaction_of_tolerance_eta}
    Let $\eta\in\N$ be admissible. 
    The random interaction $\Psi^\eta$ defined by 
    \begin{equation}
    \label{Eqn:Canonical_interaction_eta}
        \Psi^\eta(\Lambda)
            :=
        \begin{cases}
            \one_{[x, r_\eta(x)]}
            -
            G^{[x, r_\eta(x)]} 
            &\text{if }\Lambda = [x, r_\eta(x)]\\
            0 &\text{else}
        \end{cases}
    \end{equation}
    exposes $\psi$. 
\end{cor}
\begin{definition}[Canonical interaction]
\label{Def:Canonical_interaction_of_tolerance_eta}
    For admissible $\eta\in\N$, we call the random interaction $\Psi^\eta$ defined in (\ref{Eqn:Canonical_interaction_eta}) the \textbf{canonical interaction of radial tolerance} $\eta$. 
\end{definition}
Now that we know Definition \ref{Def:Interaction_exposing_psi} is non-vacuous, we move on to prove that all such interactions have $\psi$ as a unique frustration-free ground state, and, as a result, $\psi$ is almost surely a pure state. 
\begin{corollary}[$\psi$ is a unique frustration-free ground state]
\label{Cor:Unique_ground_state}
    Let $\psi$ be an injective ergodic matrix product state, and
    let $\Psi$ be a random interaction exposing $\psi$. 
    For almost every $\omega\in\Omega$, if $\varphi\in\cstates{\A{\Z}}$ satisfies $\varphi(\Psi_\omega(\Lambda)) = 0$ for all $\Lambda\Subset\Z$, then $\varphi = \psi_\omega$. 
\end{corollary}
\begin{proof}
    Fix $\Psi$ a random interaction exposing $\psi$. 
    Let $\omega\in\Omega$ be such that the conclusion of Proposition \ref{Prop:Frustration_free_infinitely_often} holds for $\seq{\mcL^\Psi_x}_{x\geq l}$ for all $l\in\Z$, and note that the set of all such $\omega$ is a set of probability 1. 
    Let $M^{[l, r]}(\omega) := \max\set{x + \mcL^\Psi_x(\omega)\,\,:\,\,x\in[l, r]}$ for all $[l, r]\subset\Z$. 
    Suppose $\varphi\in\cstates{\Omega, \A{\Z}}$ is such that $\varphi(\Psi_\omega(\Lambda)) = 0$ for all $\Lambda\Subset\Z$.
    For $[l, r]\subset\Z$, let $W^{[l, r]}\in\A{[l, M^{[l, r]}(\omega)]}$ be the density matrix of $\varphi\vert_{\A{[l, M^{[l, r]}(\omega)]}}$.
    Because $\varphi(\Psi_\omega(\Lambda)) = 0$ for all $\Lambda\Subset\Z$, we know that $W^{[l, r]}$ is supported on the intersection 
    \begin{equation*}
         \bigcap_{x\in[l, k]}
         \mathcal{H}^{[l, x)}
                \otimes 
            \mathcal{G}^{[x, x + \mcL^\Psi_x(\omega)]}_{\omega} 
                \otimes 
            \mathcal{H}^{(x + \mcL^\Psi_x(\omega), M_k^l(\omega)]}\,,
    \end{equation*}
    which, by Proposition \ref{Prop:Frustration_free_infinitely_often}, is equal to $\mathcal{G}_\omega^{[l, M^{[l, r]}(\omega)]}$. 
    So, $W^{[l, r]} = \sum_{j=1}^{\BondDim^2}\ketbra{\Gamma^{[l, M^{[l, r]}(\omega)]}_\omega(b_{j})}{\Gamma^{[l, M^{[l, r]}(\omega)]}_\omega(b_{j})}$ for some $\set{b_{j}}_{j=1}^{\BondDim^2}\subset\MatricesBond$, whence we conclude 
         \begin{equation*}
        \varphi(\bs{a})
        =
        \sum_{j=1}^{\BondDim^2}
        \bra{\Gamma^{[l, M^{[l, r]}(\omega)]}_\omega(b_{j})}
            \bs{a}
        \ket{\Gamma^{[l, M^{[l, r]}(\omega)]}_\omega(b_{j})}
        \end{equation*}
    for all $\bs{a}\in\A{[l, M^{[l, r]}(\omega)]}.$
    Since $\varphi$ is a state, it holds in particular that
    \begin{equation}\label{Eqn:Prop:Unique_ground_state_Eqn_1}
        1
        =
        \varphi(\one_{[l, M^{[l, r]}(\omega)]})
        =
        \sum_{j=1}^{\BondDim^2}
       \left\|\Gamma^{[l, M^{[l, r]}(\omega)]}_\omega(b_{j})\right\|^2\,.
    \end{equation}
    Thus, for all $\bs{a}\in\A{[l_0, r_0]}$ where $[l_0, r_0]\subsetneq [l, M^{[l, r]}(\omega)]$, we see by Lemma \ref{lem:alpha_estimates} (c) that 
    \begin{align*}
        |
            \varphi(\bs{a})
                -
            \psi_\omega(\bs{a})
        |
            &=
        \left|
        \sum_{j=1}^{\BondDim^2}
            \bra{\Gamma^{[l, M^{[l, r]}(\omega)]}_\omega(b_{j})}
            \bs{a}
        \ket{\Gamma^{[l, M^{[l, r]}(\omega)]}_\omega(b_{j})}
        -
        \psi_\omega(\bs{a})
        \braket{\Gamma^{[l, M^{[l, r]}(\omega)]}_\omega(b_{j})}
        {\Gamma^{[l, M^{[l, r]}(\omega)]}_\omega(b_{j})}
        \right|\,,\\
            &\leq 
        \big(
            \alpha^{[l, l_0-1]}(\omega)
            +
            \alpha^{[r_0 + 1, M^{[l, r]}(\omega)]}(\omega)
        \big)
        \|\bs{a}\|
        \sum_{j=1}^{\BondDim^2}
        \hatrho_{T^{M^{[l, r]}(\omega)}(\omega)}(b_{j}^*b_{j})\,,\\
            &\leq 
        \cfrac{ \alpha^{[l, l_0-1]}(\omega)
            +
            \alpha^{[r_0 + 1, M^{[l, r]}(\omega)]}(\omega)}{
            \kappa^{[l, M^{[l, r]}(\omega)]}(\omega)
            }
            \|\bs{a}\|\,,
    \end{align*}
    holds almost surely by (\ref{Eqn:Prop:Unique_ground_state_Eqn_1}).
    Therefore, by condition (c) in Definition \ref{Def:Interaction_exposing_psi} and Lemma \ref{lem:alpha_estimates} (a), we have 
    \begin{align*}
         |
            \varphi(\bs{a})
                -
            \psi_\omega(\bs{a})
        |
        \leq 
        \liminf_{r\to\infty} 
         \cfrac{ \alpha^{[l, l_0-1]}(\omega)
            +
            \alpha^{[r_0 + 1, M^{[l, r]}(\omega)]}(\omega)}{
            \kappa^{[l, M^{[l, r]}(\omega)]}(\omega)
            }
            \|\bs{a}\|
            =
             \alpha^{[l, l_0-1]}(\omega)
            \|\bs{a}\|\,.
    \end{align*}
    By another application of Lemma \ref{lem:alpha_estimates} (a), the proof is done upon taking the limit infimum as $l\to-\infty$.
\end{proof}
\begin{corollary}[$\psi$ is pure]
\label{Cor:Pure_state}
    Let $\psi$ be an injective ergodic matrix product state. 
    Then $\psi$ is pure almost surely. 
\end{corollary}
\begin{proof}
    Fix a random interaction $\Psi$ exposing $\psi$. 
   Let $\omega\in\Omega$ be such that the conclusion of Corollary \ref{Cor:Unique_ground_state} holds. 
   For any such $\omega$, let $\varphi\in\cstates{\A{\Z}}$ be a convex component of $\psi_\omega$. 
   Then there is $\lambda > 0$ such that $0\leq \varphi \leq \lambda \psi_\omega$. 
   This, however, implies that $\varphi(\Psi_\omega(\Lambda)) = 0$ for all $\Lambda\Subset\Z$. 
   Thus, by Corollary \ref{Cor:Unique_ground_state}, $\varphi = \psi_\omega$, which concludes the proof. 
\end{proof}
    If $\ell$ is essentially unbounded, then any random interaction $\Psi$ exposing $\psi$ is infinite-range almost surely. 
    Indeed, by ergodicity of $T$, $\ell$ being essentially unbounded is equivalent to $\sup_{x\geq l}\ell_x(\omega) = \infty$ for almost every $\omega\in\Omega$ for any $l\in\Z$. 
    In particular, for any such $\Psi$, we see that 
    \begin{equation*}
        \sup
        \{
            \operatorname{diam}(\Lambda)\,\,:\,\,
            \Psi_\omega^\eta(\Lambda)\neq 0
        \}
        \geq  
        \sup 
        \{ 
            \mcL^\eta_x(\omega)\text{ for $x\geq 0$}
        \}
        \geq 
        \sup_{x\geq 0}\ell_x(\omega) = \infty\,.
    \end{equation*}
This aside, we can show that $\Psi^\eta$ is {locally} finite-range almost surely, which we now do. 
\begin{lemma}
\label{Lem:Loc_fin_range}
    For any admissible $\eta\in\N$, $\Psi^\eta$ is locally finite-range.
\end{lemma}
\begin{proof}
    By Poincar\'e recurrence, for any $x\in\Z$ and almost every $\omega\in\Omega$ there is $x_0(\omega) < x - 2\eta$ such that 
    \begin{equation*}
        k^\eta_{x_0(\omega)}(\omega)\leq \eta.
    \end{equation*}
    Therefore, $x_0(\omega) + \mcL^\eta_{x_0(\omega)}(\omega) < x$. 
    It is then clear that
    \begin{equation*}
        \set{y\in\Z\,\,:\,\, x\in[y, y + \mcL^\eta_y(\omega)]}
        \subseteq  
        [x_0(\omega) + 1, x],
    \end{equation*}
    from which the lemma follows. 
\end{proof}
With this in hand, the first part of Theorem \ref{thm:InfinitePHam} is nearly proved.
All that remains is to show that $\Psi^\eta_\omega$ generates dynamics on the quasilocal algebra. 
To prove this, it suffices by Proposition \ref{Prop:Dynamics_from_loc_fin_interactions} to establish the following. 
\begin{lemma}[Bounded surface energy]
\label{Lem:Quasilocal_dynamics_for_eta_interaction}
    For any admissible $\eta\in\N$,
    there are random variables $l_n, r_n:\Omega\to\Z$ with the following properties. 
    \begin{enumerate}[label = (\alph*)]
        \item The intervals $\Lambda_n := [l_n, r_n]$ are increasing and absorbing.

        \item There is a deterministic constant $M\in (2\eta, \infty)$ such that 
        \begin{equation}
        \label{Eqn:Quasilocal_dynamics_for_eta_interaction_Eqn_1}
            \left\|
                \sum_{\substack{\Lambda\cap\Lambda_n\neq\emptyset\\
                \Lambda\cap\Lambda_n^c\neq\emptyset}}
                \Psi^\eta(\Lambda)
            \right\|
            \leq 
            M
        \end{equation} 
        almost surely. 
    \end{enumerate}
\end{lemma}
\begin{proof}
    Fix $M > 0$ such that $\Pr[\mcL^\eta < M]>0$, and notice that by the definition of $\mcL^\eta$ it is necessary that $M > \eta$. 
    Define $t_n:\Omega\to\Z^{<0}$ recursively by $t_1(\omega):= \max\set{t < 0\,\,:\,\, \mcL^\eta_t < M}$ and 
    \begin{equation*}
        t_n(\omega)
            :=
        \max\set{t < t_{n-1}(\omega)\,\,:\,\, \mcL^\eta_t < M}
    \end{equation*}
    for $n > 1$. 
    The maxima of these sets are defined by Poincar\'e recurrence and the choice of $M$. 
    Next, define $s_n:\Omega\to\Z^{>0}$ recursively by $s_1(\omega):=\max\set{s > \max\set{\mcL^\eta_{t_1 + M}(\omega) + t_1(\omega), 0}\,\,:\,\, \mcL^\eta_s(\omega) < M}$ and 
    \begin{equation*}
        s_n(\omega)
            :=
        \max\set{s > \max\set{ \mcL^\eta_{t_n + M}(\omega) + t_n(\omega), s_{n-1}(\omega)}\,\,:\,\, \mcL^\eta_s(\omega) < M}
    \end{equation*}
    for $n > 1$, where again the minima of these sets are defined by Poincar\'e recurrence. 
    Then for $n\geq 1$, let $l_n = t_n + M$ and let $r_n = s_n + M$. 
    Notice that (\ref{Eqn:Quasilocal_dynamics_for_eta_interaction_Eqn_1}) is bounded above by $S_1 + S_2 + S_3$ where 
    \begin{equation*}
        S_j
            :=
        \begin{cases}
            \#\set{
        x\in\Z\,\,:\,\,
            l_n\leq x\leq r_n < x + \mcL^\eta_x
        }&\text{$j = 1$},\\
        \#\set{
        x\in\Z\,\,:\,\,
        x< l_n\leq x + \mcL^\eta_x\leq r_n
        }&\text{$j = 2$},\\
        \#\set{
        x\in\Z\,\,:\,\,
            x< l_n\leq r_n< x + \mcL^\eta_x
        }&\text{$j = 3$}.
        \end{cases}
    \end{equation*}
    Now, because $x + \mcL^\eta_x$ is increasing in $x$ almost surely, by choice of $l_n$ and $r_n$, we have that 
    \begin{equation*}
        l_n + \mcL_{l_n}^\eta \leq r_n\,,
    \end{equation*}
    almost surely, so $S_3 = 0$ almost surely. 
    Moreover, since $l_n - M = t_n$ and $\mcL_{t_n}^\eta < M$, it holds that $l_n - M + \mcL_{l_n - M}^\eta < l_n$, hence for all $x \leq l_n - M$, $x + \mcL^\eta_x < l_n$. 
    In particular, $S_2 \leq M$ almost surely. 
    Lastly, by the same argument, since $r_n - M = s_n$ and $\mcL_{s_n}^\eta < M$, we have that $r_n - M + \mcL_{r_n - M}^\eta < r_n$, so for all $x \leq r_n - M$, we have $x + \mcL^\eta_x < r_n$, whence we conclude $S_1 \leq M$ almost surely. 
    Therefore, (\ref{Eqn:Quasilocal_dynamics_for_eta_interaction_Eqn_1}) is bounded by $2M$ almost surely, which concludes the proof.  
\end{proof}
%
%
%
%
%
%
%
%
\section{Local Spectral Gap Estimates}
\label{sec:GapEstimates}
In this section, we produce lower bounds on the spectral gap of the canonical Hamiltonian $\Psi^\eta$ that depend on the local statistics of the ergodic matrix product state and the dynamical properties of the transfer apparatus ergodic quantum process.  
We follow the martingale method of Nachtergaele introduced in \cite{Nachtergaele}, keeping track of local statistical information along the way. 
The martingale method has become a mainstay for proving the existence of a nonzero bulk gap (see e.g. \cite{BachmannHamzaNachtergaeleYoung, BishopNachtergaeleYoung, NachtergaeleSimsYoung_2, Young2024OnModels, Young_thesis}). 
Here, $\Psi^\eta$ need not be gapped for any $\eta$, as was demonstrated by the example in \cite{RoonSchenker}. 
Therefore, the purpose of this section is twofold: 
first, the investigation of how the local statistics of the EMPS come to play with regards to the spectral gap serves an explanatory role in describing precisely why the gap vanishes for certain examples of EMPS; second, this investigation produces conditions that, when satisfied, show that $\Psi^\eta$ is, in fact, gapped. 
A key technical part of the martingale method is good approximations for local ground state projectors. 
Here, this manifests as follows. 
Let $x_0 < x_1 < x_2 < x_3$. 
Note that for any $\xi\in\mathcal{G}^{[x_0, x_2]}\otimes\Hilb{(x_2, x_3]}$, we can write 
\begin{equation*}
    \xi 
    =
     \sum_{
        \substack{
        \sigma_1\in[\LocalDim]^{[x_0, x_1)}\\
        \sigma_2\in[\LocalDim]^{[x_1, x_2]}\\
        \sigma_3\in[\LocalDim]^{(x_2, x_3]}
        }
        }
            \trD{X^{\sigma_1\sigma_2}b^{\sigma_3}}
        \ket{\sigma_1\sigma_2\sigma_3}\,,
\end{equation*}
for some collection $\set{b^{\sigma_3}}_{\sigma_3\in[\LocalDim]^{[x_1, x_3]}}\subset\MatricesBond$ determined by $\xi$, and, similarly, for any $\zeta\in\Hilb{[x_0, x_1)}\otimes\mathcal{G}^{[x_1, x_3]}$, we can write
\begin{equation*}
    \zeta
    =
     \sum_{
        \substack{
        \sigma_1\in[\LocalDim]^{[x_0, x_1)}\\
        \sigma_2\in[\LocalDim]^{[x_1, x_2]}\\
        \sigma_3\in[\LocalDim]^{(x_2, x_3]}
        }
        }
            \trD{
            c^{\sigma_1}
            X^{\sigma_2\sigma_3}}
        \ket{\sigma_1\sigma_2\sigma_3}\,,
\end{equation*}
for some collection
$\set{c^{\sigma_1}}_{\sigma_1\in[\LocalDim]^{[x_0, x_2]}}\subset\MatricesBond$ determined by $\zeta$. 
\begin{lemma}
\label{lem:Delta_estimates}
\label{Lem:Good_approximations_part_1}
    For $j\in\set{0, 1, 2, 3}$, let $x_j:\Omega\to\Z$ be random variables such that $x_0 < x_1 < x_2 < x_3$, $\Gamma^{[x_1, x_2]}$ is injective, and  $x_3 > x_2 + 1 + \ell_{x_2 + 1}$ almost surely. 
    Fix $\xi\in\mathcal{G}^{[x_0, x_2]}\otimes\Hilb{(x_2, x_3]}$ and $\zeta\in\Hilb{[x_0, x_1)}\otimes \mathcal{G}^{[x_1, x_3]}$, and, as described above, write
    \begin{equation*}
         \xi 
    =
     \sum_{
        \substack{
        \sigma_1\in[\LocalDim]^{[x_0, x_1)}\\
        \sigma_2\in[\LocalDim]^{[x_1, x_2]}\\
        \sigma_3\in[\LocalDim]^{(x_2, x_3]}
        }
        }
            \trD{X^{\sigma_1\sigma_2}b^{\sigma_3}}
        \ket{\sigma_1\sigma_2\sigma_3}\,,
        \quad\text{and}\quad 
        \zeta
    =
     \sum_{
        \substack{
        \sigma_1\in[\LocalDim]^{[x_0, x_1)}\\
        \sigma_2\in[\LocalDim]^{[x_1, x_2]}\\
        \sigma_3\in[\LocalDim]^{(x_2, x_3]}
        }
        }
            \trD{
            c^{\sigma_1}
            X^{\sigma_2\sigma_3}}
        \ket{\sigma_1\sigma_2\sigma_3}\,,
    \end{equation*}
    for some sets of (random) matrices $\set{b^{\sigma_3}}_{\sigma_3\in[\LocalDim]^{[x_1, x_3]}}\subset\MatricesBond$  and $\set{c^{\sigma_1}}_{\sigma_1\in[\LocalDim]^{[x_0, x_2]}}\subset\MatricesBond$.
    Define random matrices $\tilde \Delta_\xi, \Delta_\zeta\in\MatricesBond$ by 
    \begin{equation*}
        \tilde\Delta_\xi 
            :=
        \sum_{\sigma\in[\LocalDim]^{(x_2, x_3]}}
        b^\sigma \rho_{x_2}(X^\sigma)^*\rho_{x_2}^{-1}\,,
        \quad\text{and}\quad
        \Delta_\zeta 
            :=
        \sum_{\sigma\in[\LocalDim]^{[x_0, x_1)}}
        (X^\sigma)^* c^{\sigma}\,.
    \end{equation*}
    Then almost surely 
    \begin{equation}
    \label{Eqn:Good_approximations_part_1_Eqn_1}
        \Big| 
            \braket{\xi}{\zeta}
                -
            \inner{\tilde{\Delta}_\xi}{\Delta_\zeta}_{\rho_{x_2}}
        \Big|
        \leq 
        \|\xi\|\|\zeta\|
        \cfrac{\alpha^{[x_1, x_2]}}{\kappa^{[x_1, x_2]}}.
    \end{equation}    
    Moreover, if $\xi, \zeta\in\big(\mathcal{G}^{[x_0, x_3]}\big)^\perp$, then 
    \begin{equation}
    \label{Eqn:Good_approximations_part_1_Eqn_2}
        \|
            \tilde{\Delta}_{\xi}
        \|_{\rho_{x_2}}
            \leq 
        \|\xi\|
        \cfrac{\alpha^{[x_1, x_2]}}{\sqrt{\kappa^{[x_1, x_2]}}}\,,
        \qquad\text{and}\qquad 
        \|
            \Delta_{\zeta}
        \|_{\rho_{x_3}}\leq 
        \|\zeta\|
        \cfrac{\alpha^{[x_1, x_2]}}{\sqrt{\kappa^{[x_1, x_2]}}}\,,
    \end{equation}
    almost surely. 
\end{lemma}
\begin{proof} 
    From direct computation, we see  
    \begin{align*}
        \braket{\xi}{\zeta}
            &=
        \sum_{\substack{
            \sigma_1\in[\LocalDim]^{[x_0, x_1)}\\
            \sigma_3\in[\LocalDim]^{(x_2, x_3]}
        }}
        \braket{
        \Gamma^{[x_1, x_2]}\big(
            X^{\sigma_1}b^{\sigma_3}
        \big)
        }
        {
        \Gamma^{[x_1, x_2]}\big(
            c^{\sigma_1}X^{\sigma_3}
        \big)
        }\,,
    \end{align*}
    and,
    \begin{align*}
       \inner{\tilde{\Delta}_\xi}{\Delta_\zeta}_{\rho_{x_2}}
            =
         \sum_{\substack{
            \sigma_1\in[\LocalDim]^{[x_0, x_1)}\\
            \sigma_3\in[\LocalDim]^{(x_2, x_3]}
        }}
        \inner{X^{\sigma_1}b^{\sigma_3}}{c^{\sigma_1}X^{\sigma_3}}_{\rho_{x_2}}.
    \end{align*}
    hold almost surely. 
    So, by Lemma \ref{lem:alpha_estimates} (b), we have 
    \begin{align*}
        \Big| 
            \braket{\xi}{\zeta}
                -
            \inner{\tilde{\Delta}_\xi}{\Delta_\zeta}_{\rho_{x_2}}
        \Big|
            &\leq 
        \alpha^{[x_1, x_2]}
         \sum_{\substack{
            \sigma_1\in[\LocalDim]^{[x_0, x_1)}\\
            \sigma_3\in[\LocalDim]^{(x_2, x_3]}
        }}
        \|
            X^{\sigma_1}b^{\sigma_3}
        \|_{\rho_{x_2}}
        \|
            c^{\sigma_1}X^{\sigma_3}
        \|_{\rho_{x_2}}\\
            &\leq 
         \alpha^{[x_1, x_2]}
         \seq{\sum_{\substack{
            \sigma_1\in[\LocalDim]^{[x_0, x_1)}\\
            \sigma_3\in[\LocalDim]^{(x_2, x_3]}
        }}
        \|
            X^{\sigma_1}b^{\sigma_3}
        \|_{\rho_{x_2}}^2
        }^{1/2}
        \seq{
         \sum_{\substack{
            \sigma_1\in[\LocalDim]^{[x_0, x_1)}\\
            \sigma_3\in[\LocalDim]^{(x_2, x_3]}
        }}
        \|
            c^{\sigma_1}X^{\sigma_3}
        \|_{\rho_{x_2}}^2
        }^{1/2}\,.
    \end{align*}
    Now, notice that 
    \begin{align*}
        \sum_{\substack{
            \sigma_1\in[\LocalDim]^{[x_0, x_1)}\\
            \sigma_3\in[\LocalDim]^{(x_2, x_3]}
        }}
        \|
            X^{\sigma_1}b^{\sigma_3}
        \|_{\rho_{x_2}}^2
            &=
        \sum_{\substack{
            \sigma_1\in[\LocalDim]^{[x_0, x_1)}\\
            \sigma_3\in[\LocalDim]^{(x_2, x_3]}
        }}
        \trD{
        \rho_{x_2}
        b^{\sigma_3*}X^{\sigma_1*}
        X^{\sigma_1}b^{\sigma_3}
        }\\
            &= 
         \sum_{\sigma_3\in[\LocalDim]^{(x_2, x_3]}
        }
        \|
            b^{\sigma_3}
        \|_{\rho_{x_2}}^2\,,
    \end{align*}
    where we have used that $\sum_{\sigma}X^{\sigma*}X^\sigma = \bondOne$.
    Similarly, 
    \begin{align*}
        \sum_{\substack{
            \sigma_1\in[\LocalDim]^{[x_0, x_1)}\\
            \sigma_3\in[\LocalDim]^{(x_2, x_3]}
        }}
        \|
            c^{\sigma_1}X^{\sigma_3}
        \|_{\rho_{x_2}}^2
            &=
        \sum_{\substack{
            \sigma_1\in[\LocalDim]^{[x_0, x_1)}\\
            \sigma_3\in[\LocalDim]^{(x_2, x_3]}
        }}
        \trD{
        \rho_{x_2}
        X^{\sigma_3*}c^{\sigma_1*}
        c^{\sigma_1}X^{\sigma_3}
        }\\
            &= 
         \sum_{\sigma_1\in[\LocalDim]^{[x_0, x_1)}}
        \trD{
        E_{\one}^{(x_2, x_3]}(
        \rho_{x_2})
        c^{\sigma_1*}
        c^{\sigma_1}
        }\\
        &= 
         \sum_{\sigma_1\in[\LocalDim]^{[x_0, x_1)}}
        \trD{
        \rho_{x_3}
        c^{\sigma_1*}
        c^{\sigma_1}
        }\\
        &= 
         \sum_{\sigma_1\in[\LocalDim]^{[x_0, x_1)}}
        \|c^{\sigma_1}\|_{\rho_{x_3}}^2.
    \end{align*}
    Thus, from Lemmas \ref{lem:alpha_estimates} and \ref{Lem:Kappa_lemma}, we conclude that 
    \begin{align*}
         \Big| 
            \braket{\xi}{\zeta}
                -
            \inner{\tilde{\Delta}_\xi}{\Delta_\zeta}_{\rho_{x_2}}
        \Big|
                &\leq 
         \alpha^{[x_1, x_2]}
         \seq{ \sum_{\sigma_3\in[\LocalDim]^{(x_2, x_3]}
        }
        \|
            b^{\sigma_3}
        \|_{\rho_{x_2}}^2}^{1/2}
        \seq{ \sum_{\sigma_1\in[\LocalDim]^{[x_0, x_1)}}
        \|c^{\sigma_1}\|_{\rho_{x_3}}^2}^{1/2}
                    \\
                &\leq 
        \frac{\alpha^{[x_1, x_2]}}{\kappa^{[x_1, x_2]}}
         \seq{ 
            \sum_{\sigma_3\in[\LocalDim]^{(x_2, x_3]}
        }
       \left\|
        \Gamma^{[x_0, x_2]}(b^{\sigma_3})
       \right\|^2
        }^{1/2}
        \seq{ 
            \sum_{\sigma_1\in[\LocalDim]^{[x_0, x_1)}}
        \left\|
        \Gamma^{[x_1, x_3]}(c^{\sigma_1})
       \right\|^2
        }^{1/2}
                    \\
                &= 
        \|\xi\|
        \|\zeta\|
        \frac{\alpha^{[x_1, x_2]}}{\kappa^{[x_1, x_2]}},
    \end{align*}
    which is \eqref{Eqn:Good_approximations_part_1_Eqn_1}.
    To prove \eqref{Eqn:Good_approximations_part_1_Eqn_2}, first notice that for any $b\in\MatricesBond$, if we let $\chi = \Gamma^{[x_0, x_3]}(b)$, by the intersection property we have that $\chi\in \mathcal{G}^{[x_0, x_2]}\otimes\Hilb{(x_2, x_3]}$, so 
    \begin{equation*}
        \chi
            =
          \sum_{
        \substack{
        \sigma_1\in[\LocalDim]^{[x_0, x_1)}\\
        \sigma_2\in[\LocalDim]^{[x_1, x_2]}\\
        \sigma_3\in[\LocalDim]^{(x_2, x_3]}
        }
        }
            \trD{X^{\sigma_1\sigma_2}d^{\sigma_3}}
        \ket{\sigma_1\sigma_2\sigma_3}
    \end{equation*}
    for some $\set{d^{\sigma_{3}}}_{ \sigma_3\in[\LocalDim]^{(x_2, x_3]}}\subset\MatricesBond$. 
    Because $\Gamma^{[x_0, x_2]}$ is injective, it is immediate that $d^{\sigma_3} = b X^{\sigma_3}$ for all $ \sigma_3\in[\LocalDim]^{(x_2, x_3]}$. 
    So, 
    \begin{equation*}
        \tilde{\Delta}_{\chi}
            =
            b
        \seq{\sum_{\sigma\in[\LocalDim]^{(x_2, x_3]}}
        X^{\sigma} \rho_{x_2} (X^\sigma)^*}\rho_{x_2}^{-1}
        =
        b
            E_{\one}^{(x_2, x_3]}(\rho_{x_2})
        \rho_{x_2}^{-1}
        =
        b\rho_{x_3}\rho_{x_2}^{-1}.
    \end{equation*}    
    So, $ \inner{b}{\Delta_{\zeta}}_{\rho_{x_3}} =  \inner{ \tilde{\Delta}_{\chi}}{\Delta_{\zeta}}_{\rho_{x_2}}$.
    Thus, by the above arguments proving \eqref{Eqn:Good_approximations_part_1_Eqn_1}, whenever $\zeta\perp\mathcal{G}^{[x_0, x_3]}$, for any $b\in\MatricesBond$ we find that 
    \begin{equation*}
        \left|
         \inner{b}{\Delta_{\zeta}}_{\rho_{x_3}}
        \right|
        =
        \left| 
            \braket{\chi}{\zeta}
                -
            \inner{ \tilde{\Delta}_{\chi}}{\Delta_{\zeta}}_{\rho_{x_2}}
        \right|
        \leq 
        \|b\|_{\rho_{x_3}}
        \|\zeta\|
         \frac{\alpha^{[x_1, x_2]}}{\kappa^{[x_1, x_2]}},
    \end{equation*}
    which shows the corresponding bound in \eqref{Eqn:Good_approximations_part_1_Eqn_2}. 
    A similar argument shows the other bound of \eqref{Eqn:Good_approximations_part_1_Eqn_2}, which concludes the proof. 
\end{proof}
\begin{corollary}
\label{Lem:good_approximations_part_2}
\label{Cor:good_approximations_part_2}
    For $j\in\set{0, 1, 2, 3}$, let $x_j:\Omega\to\Z$ be as in Lemma \ref{Lem:Good_approximations_part_1}.
    Then 
    \begin{equation*}
        \big\|
            (
            G^{[x_0, x_2]}
                \otimes 
            \one_{(x_2, x_3]}
            )
            (
            \one_{[x_0, x_1)}
                \otimes 
            G^{[x_1, x_3]}    
            )
            -
            G^{[x_0, x_3]}
        \big\|
        \leq 
        \alpha^{[x_1, x_2]}
        \left(
        \cfrac{
        \big\| 
            \rho_{x_3}^{-1/2}\rho_{x_2}\rho_{x_3}^{-1/2}
        \big\|^{1/2}
        +
            \alpha^{[x_1, x_2]}
        }{
        \kappa^{[x_1, x_2]}
        }
        \right)
    \end{equation*}
    holds almost surely. 
\end{corollary}
\begin{proof}
    Note that our hypotheses ensure that 
    \begin{equation*}
        \mathcal{G}^{[x_0, x_3]}
            =
        \seq{\mathcal{G}^{[x_0, x_2]} \otimes\Hilb{(x_2, x_3]}}
        \cap
        \seq{\Hilb{x_0, x_1)}\otimes \mathcal{G}^{[x_1, x_3]}}
    \end{equation*}
    by the argument of Proposition \ref{Prop:Intersection_property}.
    In particular, 
    \begin{align*}
         \big\| (
            G^{[x_0, x_2]}
                \otimes 
            \one_{(x_2, x_3]}
            )
            (
            \one_{[x_0, x_1)}
                \otimes 
            G^{[x_1, x_3]}    
            )
            -
            G^{[x_0, x_3]}\big\|  
            &\\
            &\hspace{-35mm}=
          \big\|(
            G^{[x_0, x_2]}
                \otimes 
            \one_{(x_2, x_3]}
            -
             G^{[x_0, x_3]} 
            )
            (
            \one_{[x_0, x_1)}
                \otimes 
            G^{[x_1, x_3]}  
            - G^{[x_0, x_3]} 
            )\big\|
    \end{align*}
    holds almost surely. 
    Therefore, 
    \begin{equation*}
         \big\| (
            G^{[x_0, x_2]}
                \otimes 
            \one_{(x_2, x_3]}
            )
            (
            \one_{[x_0, x_1)}
                \otimes 
            G^{[x_1, x_3]}    
            )
            -
            G^{[x_0, x_3]}\big\|  
            = 
        \sup_{\substack{
            \xi\in \mathcal{G}^{[x_0, x_2]}\otimes \mathcal{H}^{(x_2,x_3]}\\
            \zeta\in \mathcal{H}^{[x_0, x_1)}\otimes \mathcal{G}^{[x_1, x_3]}\\
            \xi, \zeta\perp \mathcal{G}^{[x_0, x_3]}
        }}
        \cfrac{\big|
            \braket{\xi}{\zeta}
        \big|}{
        \|\xi\|
        \|\zeta\|
        }
    \end{equation*}
    By Lemma \ref{Lem:Good_approximations_part_1}, for all $\xi, \zeta$ as in Lemma~\ref{Lem:Good_approximations_part_1} which additionally satisfy $\xi, \eta\perp \mathcal{G}^{[x_0, x_3]}$, we have that 
    \begin{align*}
        \big|
            \braket{\xi}{\zeta}
        \big|
            &\leq 
           \left|
            \inner{\tilde{\Delta}_\xi}{\Delta_\zeta}_{\rho_{x_2}}
            \right|
            +
             \|\xi\|
              \|\zeta\|
        \cfrac{\alpha^{[x_1, x_2]}}{\kappa^{[x_1, x_2]}}
            \\
        &\leq 
            \|\tilde{\Delta}_\xi\|_{\rho_{x_2}}
            \|\Delta_\zeta\|_{\rho_{x_2}}
             +
             \|\xi\|
              \|\zeta\|
        \cfrac{\alpha^{[x_1, x_2]}}{\kappa^{[x_1, x_2]}}
    \end{align*}
    by Cauchy-Schwartz applied to $\inner{\cdot}{\cdot}_{\rho_2}$.
    So, since 
    \begin{align*}
        \|
            \Delta_{\zeta}
        \|_{\rho_{x_2}}^2
                &=
        \left|
            \trD{\rho_{x_3}^{-1/2}\rho_{x_2}\rho_{x_3}^{-1/2}
            (\Delta_\zeta\rho_{x_3}^{1/2})^*\Delta_\zeta\rho_{x_3}^{1/2}
            }
        \right|
                \\
                &\leq 
        \big\| 
            \rho_{x_3}^{-1/2}\rho_{x_2}\rho_{x_3}^{-1/2}
        \big\|
        \big\|
            \Delta_{\zeta}
        \big\|_{\rho_{x_3}}^2,
    \end{align*}
    we conclude from Lemma \ref{Lem:Good_approximations_part_1} that 
    \begin{equation*}
        \big|
            \braket{\xi}{\zeta}
        \big|
        \leq 
         \|\xi\|
              \|\zeta\|   
        \alpha^{[x_1, x_2]}
        \left(
        \cfrac{
         \big\| 
            \rho_{x_3}^{-1/2}\rho_{x_2}\rho_{x_3}^{-1/2}
        \big\|^{1/2}
        +
            \alpha^{[x_1, x_2]}
        }{
        \kappa^{[x_1, x_2]}
        }
        \right),
    \end{equation*}
    ending the proof. 
\end{proof}
%
%
%
%
\subsection{Coarse graining}
With Corollary \ref{Cor:good_approximations_part_2} in hand, we now have the main technical tool required to establish our lower spectral bound. 
The idea from here is to follow the coarse graining procedure described by Fannes, Nachtergaele, and Werner in \cite{FannesNachtergaeleWerner}, which, owed to the disordered setting, is necessarily nonuniform and site-dependent. 
Nevertheless, the core idea described in \cite{FannesNachtergaeleWerner} is unchanged: 
instead of proving lower spectral bounds on $\Psi^\eta$ itself, we produce such bounds for a modification of this interaction with nearest-neighbor interaction on a ``stochastically regrouped" spin chain. 
\begin{definition}[Stochastic regrouping]
\label{Def:Stochastic_regrouping}
    Fix $x\in\Z$ and admissible $\eta\in\N$. 
    A stochastic regrouping of $\Z$ started at $x$ with tolerance $\eta$ is a sequence of random variables $\mathscr{R}(x, \eta) := \seq{x_j:\Omega\to [x, \infty)}_{j\geq 0}$ defined by 
    \begin{equation*}
        x_j := 
        r_\eta^{p_j}(x)
    \end{equation*}
    where $p_j:\Omega\to\N$ is a sequence of random variables with $p_j < p_{j+1}$ almost surely for all $j\geq 0$. 
    Given such a regrouping $\mathscr{R}(x, \eta)$, for any $j\geq 0$ define
    \begin{equation}
    \label{Eqn:Stochastic_regrouping}
        \mathfrak{h}^{j, j + 1}_{\mathscr{R}(x, \eta)}
            :=
        \displaystyle\sum_{y = x_j}^{r_\eta^{-1}(x_{j+2})}
            \Psi^\eta\big(
                [y, r_\eta(y)]
            \big)
            \in 
            \A{[x_j, x_{j+2}]},
    \end{equation}
    where $r_\eta^{-1}(x_{j+2}) = r_\eta^{p_{j+2} - 1}(x)$.
\end{definition}
\begin{figure}[H]
      \centering
      \begin{tikzpicture}[scale=1.5, every node/.style={scale=.8}]

        \node[draw, circle, black, fill=black, minimum size=.2cm, inner sep=0pt, label=$x$](n1) at (0.7, 0) {1};

        \node[draw, circle, black, fill=black, minimum size=.2cm, inner sep=0pt, label=$r_\eta^{p_1}(x)$](n2) at (2.2, 0) {2};

        \node[draw, circle, black, fill=black, minimum size=.2cm, inner sep=0pt, label=$r_\eta^{p_2}(x)$](n3) at (2.8, 0) {3};

        \node[draw, circle, black, fill=black, minimum size=.2cm, inner sep=0pt, label=$r_\eta^{p_3}(x)$](n4) at (4.2, 0) {4};

        \node[draw, circle, black, fill=black, minimum size=.2cm, inner sep=0pt, label=$r_\eta^{p_{m-1}}(x)$](n5) at (7.4, 0) {5};

        \node[draw, circle, black, fill=black, minimum size=.2cm, inner sep=0pt, label=$r_\eta^{p_m}(x)$](n6) at (9.1, 0) {6};

        \node[label=$\cdots$](n7) at (5.8, 0) {\phantom{7}};
        
        \draw[thick, black](0.7, 0) -- (9.5, 0);
        \draw[thick, dotted, black] (9.5, 0) -- (10, 0);


        \draw[line width = 0.39cm, PineGreen, draw opacity = 0.32, line cap=round] (0.7, 0) -- (2.8, 0);

        \draw[line width = 0.39cm, DarkOrchid, draw opacity = 0.32, line cap=round] (2.2, 0) -- (4.2, 0);

        \draw[line width = 0.39cm, YellowGreen, draw opacity = 0.32, line cap=round] (4.2, 0) -- (5.2, 0);

        \draw[line width = 0.39cm, Goldenrod, draw opacity = 0.32, line cap=round] (6.3, 0) -- (7.4, 0);
        
        \draw[line width = 0.39cm, Bittersweet, draw opacity = 0.32, line cap=round] (7.4, 0) -- (9.1, 0);


        \draw[
        black, 
        decoration={brace}, decorate] (0.7, 0.5) -- (2.7, 0.5);  
        \node[PineGreen] (H) at (1.7, 0.7){$ \mathfrak{h}^{0, 1}_{\mathscr{R}(x, \eta)}$};

        \draw[
        black, 
        decoration={brace}, decorate] (4.2, -0.4) -- (2.2, -0.4);  
        \node[DarkOrchid] (H) at (3.2, -0.6){$ \mathfrak{h}^{1, 2}_{\mathscr{R}(x, \eta)}$};

        \draw[
        black, 
        decoration={brace}, decorate] (7.4, 0.5) -- (9.1, 0.5);  
        \node[Bittersweet] (H) at (8.25, 0.7){$ \mathfrak{h}^{m-1, m}_{\mathscr{R}(x, \eta)}$};
        
    \end{tikzpicture}
    \caption{Stochastic regrouping}
\end{figure}
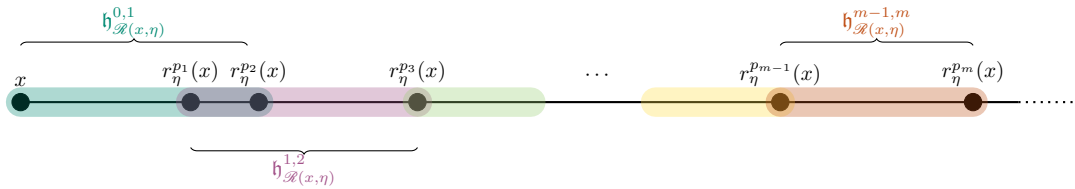
Given a stochastic regrouping $\mathscr{R}(x, \eta)$, we understand $ \mathfrak{h}^{j, j + 1}_{\mathscr{R}(x, \eta)}$ as defining a nearest-neighbor interaction on the regrouped chain with on-site algebra $\A{[x_j, x_{j+1})}$ at site $j$. 
Notice that 
\begin{equation*}
    H^{\Psi^\eta, [x, x_{m}]}
    \geq 
    \frac{1}{2}
    \sum_{j=0}^{m-2}
     \mathfrak{h}^{j, j + 1}_{\mathscr{R}(x, \eta)}
\end{equation*}
holds almost surely for any $m\geq 0$. 
Thus, to produce lower bounds on the spectral gap of $H^{\Psi^\eta, [x, x_{m}]}$, it suffices to produce lower bounds for the sum $\sum_{j=0}^{m-2}
     \mathfrak{h}_{\mathscr{R}(x, \eta)}^{j, j + 1}$.
Towards this end, notice by the intersection property (Proposition \ref{Prop:Intersection_property}) that 
\begin{equation}
      \mathfrak{h}_{\mathscr{R}(x, \eta)}^{j, j + 1}
        \geq 
    \tilde{\gamma}_{x, j} \seq{\one_{[x_j, x_{j+2}]} - G^{[x_j, x_{j+2}]}}
\end{equation}
where $\tilde{\gamma}_{x, j}$ is the smallest nonzero eigenvalue of $\mathfrak{h}_{\mathscr{R}(x, \eta)}^{j, j + 1}$. 
Therefore, 
\begin{equation}
\label{Eqn:Lower_bound_first_one}
    H^{\Psi^\eta, [x, r_\eta^{m}(x)]}
    \geq 
    \frac{1}{2}
    \seq{\min_{j\in\set{0, \dots, m-2}}
    \tilde{\gamma}_{x, j}}
    \sum_{j=0}^{m-2}
        \seq{\one_{[x_j, x_{j+2}]} - G^{[x_j, x_{j+2}]}}.
\end{equation}
We devote the rest of this section to proving lower bounds on $ \sum_{j=0}^{m-2}
        \seq{\one_{[x_j, x_{j+2}]} - G^{[x_j, x_{j+2}]}}.$
Towards this end, we introduce the notation 
\begin{equation}
\label{Eqn:Mathfrak_Ks}
    \mathfrak{K}^{[m]}_{\mathscr{R}(x, \eta)} 
        :=
    \sum_{j=0}^{m-2}
     \mathfrak{K}_{\mathscr{R}(x, \eta)}^{j, j + 1}
        \quad
        \text{where}
        \quad 
     \mathfrak{K}_{\mathscr{R}(x, \eta)}^{j, j + 1}
        := 
    \one_{[x_j, x_{j+2}]} - G^{[x_j, x_{j+2}]}.
\end{equation}
Our goal is thus to lower bound the spectral gap of $\mathfrak{K}^{[m]}_{\mathscr{R}(x, \eta)}$.
Towards this end, we recall the following lemma from \cite{FannesNachtergaeleWerner}, where for two projections $E$ and $F$, $E\wedge F$ denotes the projection onto $\operatorname{ran}(E)\cap\operatorname{ran}(F)$, and $E^\perp := \one - E$. 
\begin{lemma}[\texorpdfstring{\cite[Lemma 6.3]{FannesNachtergaeleWerner}}{l}]
\label{Lem:FNW_proj_lemma}
    Let $E$ and $F$ be orthogonal projections on a finite-dimensional Hilbert space. 
    Then
    \begin{enumerate}[label = (\alph*)]
        \item $\|EF - E\wedge F\| = \|E^\perp F^\perp - E^\perp\wedge F^\perp\|$

        \item $EF + FE \geq -\|EF - E\wedge F\|(E + F)$.
    \end{enumerate}
\end{lemma}
With this in hand, we can prove our main theorem. 
\begin{thm}
\label{Thm:Local_spectral_gap_prop}
    Fix $x\in\Z$ and $\eta\in\N$ admissible. 
    With notation as above, for any regrouping $\mathscr{R}(x, \eta)$ we have
    \begin{equation}
        \operatorname{spec-gap}\seq{H^{\Psi^\eta, [x, x_m]}}
            \geq 
        \frac{1}{2}
         \seq{\min_{j\in\set{0, \dots, m-2}}
    \tilde{\gamma}_{x, j}}
    \seq{1
    -
    \max_{j\in\set{0, \dots, m-2}}
    \beta_{x, j}}
    \end{equation}
    almost surely, where $\beta_{x, j}:\Omega\to\mathbb{R}$ is defined by
    \begin{equation}
    \label{Eqn:Def_of_beta}
        \beta_{x, j}
            :=
        2\alpha^{[x_{j+1}, x_{j+2}]}
        \left(
        \cfrac{
        \big\|
            \rho_{x_{j+3}}^{-1/2}
            \rho_{x_{j+2}}
            \rho_{x_{j+3}}^{-1/2}
        \big\|^{1/2}
        +
            \alpha^{[x_{j+1}, x_{j+2}]}
        }{
        1 - \alpha^{[x_{j+1}, x_{j+2}]}
        }
        \right)
    \end{equation}
    for all $j=0, \dots, m$. 
\end{thm}
\begin{proof}
    As discussed above, it suffices to show that 
    \begin{equation*}
        \seq{\mathfrak{K}^{[m]}_{\mathscr{R}(x, \eta)}}^2
            \geq 
         \seq{1
    -
    \max_{j\in\set{0, \dots, m-2}}
    \beta_{x, j}}
    \mathfrak{K}^{[m]}_{\mathscr{R}(x, \eta)}
    \end{equation*}
    almost surely. 
    To see this, we first compute 
    \begin{align}
         \seq{\mathfrak{K}^{[m]}_{\mathscr{R}(x, \eta)}}^2
         \geq 
         \mathfrak{K}^{[m]}_{\mathscr{R}(x, \eta)}
         +
         \sum_{j=0}^{N-2}
         \{
            \mathfrak{K}^{j, j + 1}_{\mathscr{R}(x, \eta)}, \mathfrak{K}^{{j + 1, j + 2}}_{\mathscr{R}(x, \eta)}
         \} 
    \label{Eqn:Local_spectral_gap_proof_Eqn_1}
    \end{align}
    where $\{\cdot, \cdot\}$ denotes the anticommutator, and we have noted that, by construction, $[\mathfrak{K}^{j, j + 1}_{\mathscr{R}(x, \eta)}, \mathfrak{K}^{{j + 1, j + 2}}_{\mathscr{R}(x, \eta)}] = 0$ whenever $|i - j|\geq 2$, which, since $\mathfrak{K}^{j, j + 1}_{\mathscr{R}(x, \eta)}$ is a projection, implies $\{
       \mathfrak{K}^{j, j + 1}_{\mathscr{R}(x, \eta)}, \mathfrak{K}^{{j + 1, j + 2}}_{\mathscr{R}(x, \eta)}
         \}\geq 0$.
    By Lemma \ref{Lem:FNW_proj_lemma}, we know that 
    \begin{equation*}
          \{
        \mathfrak{K}^{j, j + 1}_{\mathscr{R}(x, \eta)}, \mathfrak{K}^{{j + 1, j + 2}}_{\mathscr{R}(x, \eta)}
         \}
         \geq 
         -\big\|
            (
            G^{[x_j, x_{j + 2}]}
                \otimes 
            \one_{(x_{j + 2}, x_{j + 3}]}
            )
            (
            \one_{[x_j, x_{j + 1})}
                \otimes 
            G^{[x_{j + 1}, x_{j + 3}]}    
            )
            -
            G^{[x_j, x_{j + 3}]}
        \big\|
        \big( 
             \mathfrak{K}^{j, j + 1}_{\mathscr{R}(x, \eta)}
                +
            \mathfrak{K}^{{j + 1, j + 2}}_{\mathscr{R}(x, \eta)}
        \big)
    \end{equation*}
    for all $j = 0, \dots, m-2$, so from Corollary \ref{Cor:good_approximations_part_2} we conclude 
    \begin{equation*}
         \{
       \mathfrak{K}^{j, j + 1}_{\mathscr{R}(x, \eta)}, \mathfrak{K}^{{j + 1, j + 2}}_{\mathscr{R}(x, \eta)}
         \}
         \geq 
         -
         \frac{\beta_{x, j}}{2}
         \big( 
             \mathfrak{K}^{j, j + 1}_{\mathscr{R}(x, \eta)} + \mathfrak{K}^{{j + 1, j + 2}}_{\mathscr{R}(x, \eta)}
        \big).
    \end{equation*}
    Thus, from (\ref{Eqn:Local_spectral_gap_proof_Eqn_1}), we see
    \begin{equation*}
         \left(\mathfrak{K}^{[m]}_{\mathscr{R}(x, \eta)}\right)^2
         \geq 
          \seq{
        1
        -
        \max_{j=0, \dots, m-2}
            \beta_{x, j}
        }
        \mathfrak{K}^{[m]}_{\mathscr{R}(x, \eta)}
    \end{equation*}
    almost surely,
    which concludes the proof. 
\end{proof}
\begin{corollary}
    Fix notation as in Theorem \ref{Thm:Local_spectral_gap_prop}.
    Let $H^\eta$ denote the GNS Hamiltonian corresponding to $\Psi^\eta$. 
    For any decreasing sequence $(x_j)_{j\in\N}\subset\Z$ with $\lim_j x_j = -\infty$ and any random sequence $(m_j)_{j\in\N}$ with $r_\eta^{m_j}(x_j)< r_\eta^{m_{j+1}}(x_{j+1})$ almost surely,
    \begin{equation*}
        \operatorname{spec-gap}(H^\eta)
        \geq 
        \frac{1}{2}
        \limsup_{j\geq 1}
         \seq{\min_{j\in\set{0, \dots, m-2}}
    \tilde{\gamma}_{x, j}}
    \seq{1
    -
    \max_{j\in\set{0, \dots, m-2}}
    \beta_{x, j}}
    \end{equation*}
    holds almost surely. 
\end{corollary}
\begin{proof}
    This follows immediately from the Koma-Nachtergaele result Proposition \ref{Prop:Koma_Nachtergaele} upon noting the assumptions on $x_j$ and $m_j$ insure that $\Lambda_j = [x_j, r_{\eta}^{m_j}(x_j)]$ is an increasing and absorbing sequence almost surely. 
\end{proof}
%
%
%
\subsection{Uniformly bounded injectivity length and examples of gapped EMPS}
In this section, we briefly discuss how the above results may be improved in the situation that $\psi$ is uniformly injective. 
\begin{definition}
    If $\ell\in L^\infty(\Omega)$, we call $\psi$ \textit{uniformly injective}. 
\end{definition}
It may be useful to rephrase this condition as follows. 
\begin{lemma}
    $\Pr[\sup_{x\in\Z}\ell_x < \infty]\in\set{0, 1}$.
    In particular, 
    $\psi$ is uniformly injective if and only if $\sup_{x\in\Z}\ell_x < \infty$ almost surely. 
\end{lemma}
\begin{proof}
    It is immediate from ergodicity that $\Pr[\sup_{x\in\Z}\ell_x < \infty]\in\set{0, 1}$, after which the lemma is clear. 
\end{proof}
The first thing to note is that when $\ell\in L^\infty(\Omega)$, there is deterministic $L\in\N$ for which $\mathcal{L}_x = L$ is an injectivity sequence: simply take $L > \|\ell\|_{L^\infty(\Omega)}$. 
Then, by taking $\eta = L$, we have that 
\begin{equation*}
    \Psi^{\eta}\big(
        \Lambda
    \big)
        =
        \begin{cases}
            K^{[x, x + \eta]} &\text{if $\Lambda = [x, x + \eta]$ for some $x\in\Z$}\\
            0 &\text{else.}
        \end{cases}
\end{equation*}
As a result, we see that, for any $p > \|\ell\|_{L^\infty(\Omega)}$, $\Psi^{2p}$ is a random interaction exposing $\psi$ that satisfies all the properties of Theorem \ref{Thm:Parent_Hamiltonian}. 
Moreover, for all $j\geq 0$, we have that 
\begin{equation*}
\tau_{pj}\big(
    \Psi^{2p}([x, x + 2p])
    \big)
    =
    \mathfrak{K}^{j, j + 1}_{\mathscr{R}(x, p)}
\end{equation*}
where $\mathscr{R}(x, p)$ is defined by the deterministic sequence $p_j = jp$. 
Therefore, arguing as in the proof of Theorem \ref{Thm:Local_spectral_gap_prop}, we conclude the following. 
\begin{cor}
\label{Cor:Gaps_for_finite_inj_length}
    Assume $\ell\in L^\infty(\Omega)$ and fix deterministic $p > \|\ell\|_{L^\infty(\Omega)}$. 
    Then for all $m>0$, the lower bound
    \begin{equation*}
        \operatorname{spec-gap}\left(
            H^{\Psi^{2p}, [x, mp]}
        \right)
        \geq 
        \frac{1}{2}
        \seq{
        1 
        -
        \max_{j\in\set{0, \dots, m-2}}\beta_{x, j}
        }
    \end{equation*}
    holds almost surely,
    where
    \begin{equation*}
        \beta_{x, j}
            :=
        2\alpha^{[x + (j+1)p, x + (j + 2)p]}
        \seq{
        \cfrac{
            \big\|
                \rho_{x + (j+3)p}^{-1/2}
                \rho_{x + (j+2)p}
                \rho_{x + (j+3)p}^{-1/2}
            \big\|^{1/2}
            +
            \alpha^{[x + (j+1)p, x + (j + 2)p]}
        }
        {
        1 - \alpha^{[x + (j+1)p, x + (j + 2)p]}
        }
        }.
    \end{equation*}
\end{cor}
In particular, for example, if $\rho$ is deterministic and if $\alpha^{[l, r]}\to 0$ uniformly in the disorder as $|l - r|\to \infty$, then there is some $\eta>0$ such that $\Psi^{\eta}$ has a spectral gap above $\psi$. 
To make this result concrete---and to showcase that it can indeed be used to prove the existence of a gap---we revisit the disordered AKLT model introduced in Example \ref{Example:DAKLT_introduction}. 
\begin{definition}
    We say the disordered AKLT model of Example \ref{Example:DAKLT_introduction} is \textit{nondegenerate} if it is uniformly injective.
\end{definition}
So, the disordered AKLT model is nondegenerate if the cosines and sines defining it do not vanish too often. 
\begin{prop}\label{Prop:DAKLT}
    With the notation of Example \ref{Example:DAKLT_introduction}, for a nondegenerate disordered AKLT model, there is $\eta\in\N$ such that $\Psi^\eta$ is gapped if and only if 
    \begin{equation}\label{Eqn:DAKLT_quantity}
        \lim_{N\to\infty}\sup_{\substack{[l, r]\subset\Z\\
        |r - l| = N}}
        \max\left( 
            \prod_{x=l}^r
            |\cos 2\theta_x|,
            \prod_{x=l}^r \sin^2\theta_x
        \right)
        =
        0
    \end{equation}
    holds for $\Pr$-almost every $\overline{\theta}\in[0, 2\pi]^\Z$.
\end{prop}
\begin{proof}
    Assume first that the limit converges to zero almost surely. 
    The assumption of nondegeneracy ensures we may apply Corollary \ref{Cor:Gaps_for_finite_inj_length}. 
    Notice that 
    \begin{equation*}
        \big\|
            E^{[l, r]}_{\one_3; \overline{\theta}}
            -
            \frac{1}{2}
            \operatorname{Tr}_D(\cdot)\one_2
        \big\|
        =
         \max\left( 
            \prod_{x=l}^r
            |\cos 2\theta_x|,
            \prod_{x=l}^r \sin^2\theta_x
        \right)
    \end{equation*}   
    by the diagonal representation described in Example \ref{Example:DAKLT_introduction}.
    In particular, the quantity $\beta_{x, j}$ in the statement of Corollary \ref{Cor:Gaps_for_finite_inj_length} is 
    \begin{equation*}
        \beta_{x, j}(\overline{\theta})
        =
        2\max\left( 
            \prod_{y=x + (j+1)p}^{x + (j+2)p}
            |\cos 2\theta_y|,
            \prod_{x=l}^r \sin^2\theta_y
        \right)
        \seq{
        \cfrac{
            1
            +
            \displaystyle\max\left( 
            \prod_{y=x + (j+1)p}^{x + (j+2)p}
            |\cos 2\theta_y|,
            \prod_{x=l}^r \sin^2\theta_y
        \right)
        }
        {
        1 -  \displaystyle\max\left( 
            \prod_{y=x + (j+1)p}^{x + (j+2)p}
            |\cos 2\theta_y|,
            \prod_{x=l}^r \sin^2\theta_y
        \right)
        }
        }.
    \end{equation*}
    Because (\ref{Eqn:DAKLT_quantity}) holds, for large enough $p$, we have by Corollary \ref{Cor:Gaps_for_finite_inj_length} that there is deterministic $\epsilon>0$ such that
    \begin{equation*}
        \operatorname{spec-gap}\seq{H_{\overline{\theta}}^{\Psi^{2p}, [x, x + mp]}}
            > \epsilon
    \end{equation*}
    holds almost surely uniformly in $x$ and $m$, which by Koma-Nachtergaele \cite{KomaNachtergaele} ensures a bulk gap for $\Psi^{2p}$. 
    Conversely, if there were $\Psi^\eta$ gapped, then the gap would be deterministic, so arguing as in \cite{RoonSchenker}, by the general exponential clustering phenomenon \cite{Nachtergaele2006Lieb-RobinsonTheorem}, we necessarily have that
    \begin{equation*}
        \operatorname{ess-sup}
        \big|
            \psi_{\operatorname{DAKLT}}(S^\#\otimes\one_{[l, r]}\otimes S^\#)
        \big|
        \to 0
    \end{equation*}
    uniformly $|r - l|$, where $\#\in\set{+, -, z}$. 
    Thus, by (\ref{Eqn:Correlations_DAKLT}), we have that (\ref{Eqn:DAKLT_quantity}) holds, concluding the proof. 
\end{proof}
So, for example, if we take $\seq{\theta_x}_{x\in\Z}$ i.i.d. with $\theta_x\in (s, t)$ almost surely for some $0 < s < t < \pi/4$, then the corresponding disordered AKLT model is gapped.

\appendix
\section{Finite injectivity length is equivalent to eventual strict positivity}
\label{App:Injectivity}
In this appendix, we discuss the notion of injectivity length for general ergodic quantum processes. 
As in the main body, let $\seq{\Omega, \mathcal{F}, \mu}$ be a probability space and let $T:\Omega\to\Omega$ be an ergodic and invertible measure-preserving transformation. 
To make this appendix self-contained and since we work in a somewhat more general setting than in the main body, we introduce and reintroduce some notation and terminology specific to this appendix.
Let $\gmatrices$ denote the set of $g$-tuples of matrices, equipped with the natural Borel $\sigma$-algebra structure.
Given $\boldsymbol{V} = \seq{V_i}_{i=1}^g\in\gmatrices$, define $\boldsymbol{V}^* := \seq{V_i^*}_{i=1}^g\in\gmatrices$, and also define $\psi_{\boldsymbol{V}}:\matrices\to\matrices$ to be the completely positive map $\psi_{\boldsymbol{V}}(a) := \sum_{i=1}^g V_i a V_i^*$. 
In this appendix, we study the properties of $\gmatrices$-valued random variables $\boldsymbol{S}:\Omega\to\gmatrices$, where we write $\boldsymbol{S}_\omega = \seq{A_{i; \omega}}_{i=1}^g$ for random matrices $A_i:\Omega\to\matrices$. 
Given such $\bs{S}$ and $k\in\Z$, write $\boldsymbol{S}_k$ to denote the $\gmatrices$-valued random variable $\boldsymbol{S}_{k; \omega} := \boldsymbol{S}_{T^k(\omega)}$. 
For $n\in\N$ and $\sigma = \seq{i_1, \dots, i_n}\in[g]^n$, let $\Pi^{(\boldsymbol{S}, n)}_{\sigma}:\Omega\to\matrices$ be the random matrix defined by
\begin{equation}\label{Eqn:Product_of_matrices}
    \Pi^{(\boldsymbol{S}, n)}_{\sigma; \omega} 
    =
    A_{i_n; T^{n-1}(\omega)}\cdots A_{i_1; \omega}.
\end{equation}
Because $\boldsymbol{S}$ will usually be fixed, we simply write $\Pi^{(n)}_\sigma$ to denote $\Pi^{(\boldsymbol{S}, n)}_{\sigma}$. 
We then define the function $\prelength{\boldsymbol{S}}:\Omega \to\N\cup\{\infty\}$ by 
\begin{equation}
    \prelength{\boldsymbol{S}_\omega}
        :=
    \inf
    \Bigg\{n\in\N\,\,:\,\,
        \matrices 
        =
        \operatorname{span}
        \Big\{
            \Pi^{(\boldsymbol{S}, n)}_{\sigma; \omega}\,\,:\,\, \sigma\in[g]^n
        \Big\}
    \Bigg\},
\end{equation}
where we define $\prelength{\boldsymbol{S}_\omega} = \infty$ if the above set over which the infimum is taken is the empty set.
Similarly, define $\wielength{\boldsymbol{S}}:\Omega\to\N\cup\{\infty\}$ by
\begin{equation}
    \wielength{\boldsymbol{S}_\omega}
        :=
    \inf
    \Bigg\{n\in\N
        \,\,:\,\,
    \matrices 
        =
        \operatorname{span}
        \Big\{
            \Pi^{(\boldsymbol{S}, m)}_{\sigma; \omega}\,\,:\,\, \sigma\in[g]^m
        \Big\}
        \text{ for all $m\geq n$}
    \Bigg\},
\end{equation}
where again we take $\wielength{\boldsymbol{S}_\omega} = \infty$ when necessary.
The goal of this appendix is to relate these quantities---which pertain to $\boldsymbol{S}$---to properties of the random completely positive map $\phi_{\boldsymbol{S}}:\matrices\to\matrices$ defined by
\begin{equation}
    \begin{split}
        \phi_{\boldsymbol{S}; \omega}(a)
            :=
        \sum_{i=1}^g
        A_{i; \omega}aA_{i; \omega}^*
        \quad\text{for $a\in\matrices$.}
    \end{split}
\end{equation}
For any $n\in\Z$, let $\phi_{\boldsymbol{S}, n}$ be the random completely positive superoperator
$\phi_{\boldsymbol{S}, n; \omega} = \phi_{\boldsymbol{S}; T^{n}(\omega)}$. 
We write $\Phi^{(n)}_{\boldsymbol{S}}$ to denote $\phi_{\boldsymbol{S}, n-1}\circ\cdots\circ\phi_{\boldsymbol{S}, 0}$ for $n\in\N$ (so $\Phi^{(1)}_{\boldsymbol{S}} = \phi_{\boldsymbol{S}}$).
Given $m, n\in\Z$ with $m < n$, we define $\Phi^{(m, n)}_{\boldsymbol{S}}$ by 
\begin{equation}
    \Phi^{(m, n)}_{\boldsymbol{S}}
    =
    \phi_{\boldsymbol{S}, n}\circ\cdots\circ\phi_{\boldsymbol{S}, m}.
\end{equation}
Notice that 
\begin{equation}
    \Phi^{(n)}_{\boldsymbol{S}; \omega}\seq{a}
    =
    \sum_{\sigma\in[g]^n}
     \Pi^{(n)}_{\sigma; \omega}a \Pi^{(n)*}_{\sigma; \omega}\
\qquad\text{and}\qquad
    \Phi^{(m, n)}_{\boldsymbol{S}; \omega}\seq{a}
    =
    \sum_{\sigma\in[g]^{n-m +1}}
    \Pi^{(n-m + 1)}_{\sigma; T^{m}(\omega)}a \Pi^{(n-m + 1)*}_{\sigma; T^{m}(\omega)}
\end{equation}
for all $a\in\matrices$.
We call a linear map $\psi:\matrices\to\matrices$ strictly positive if $\psi(\rho)$ is full rank for any density matrix $\rho\in\matrices$, and we call a map $\psi:\matrices\to\matrices$ is called faithful if the only positive semidefinite matrix $p\in\matrices$ satisfying $\psi(p) = 0$ is $p = 0$. 
We call $\boldsymbol{S}$ faithful if $\psi_{\boldsymbol{S}}$ is faithful almost surely. 
For a linear map $\psi:\matrices\to\matrices$, write $\psi^*$ to denote the Hilbert space adjoint of $\psi$ with respect to the Hilbert-Schmidt inner product on $\matrices$. 
We now define quantities $q_0(\boldsymbol{S}), q(\boldsymbol{S}):\Omega\to\mathbb{N}\cup\{\infty\}$ by
\begin{equation*}
    q_0(\boldsymbol{S}_\omega)
        := 
    \inf
    \Bigg\{n\in\N
        \,\,:\,\,
    \Phi^{(n)}_{\boldsymbol{S}; \omega}\text{ is strictly positive}
    \Bigg\}
\end{equation*}
and
\begin{equation*}
    q(\boldsymbol{S}_\omega)
        := 
    \inf
    \Bigg\{n\in\N
        \,\,:\,\,
    \Phi^{(m)}_{\boldsymbol{S}; \omega}\text{ is strictly positive for all $m\geq n$}
    \Bigg\},
\end{equation*}
with these quantities defined to be infinite when the infima are undefined. 
We begin by making the following simple observation.
\begin{lem}\label{App:Lem:Wie_dont_care_about_adjoint}
    $\prelength{\boldsymbol{S}} = \prelength{\boldsymbol{S}^*}$ and $\wielength{\boldsymbol{S}} = \wielength{\boldsymbol{S}^*}$ almost surely. 
\end{lem}
\begin{proof}
    Note that $\Pi^{(\boldsymbol{S}^*, k)}_{\sigma; \omega}
        =
        \Pi^{(\boldsymbol{S}, k)*}_{\sigma^*; \omega}$, where $\sigma^* = (i_k, \dots, i_1)$ whenever $\sigma = (i_1, \dots, i_k)$. 
    So, from the fact that $(\cdot)^*:[g]^k\to[g]^k$ is a bijection and $\matrices = \matrices^*$, the lemma is clear. 
\end{proof}
Now, for any completely positive map $\psi:\matrices\to\matrices$, let $J(\psi)$ denote the Choi matrix in $\mathbb{M}_{d^2}$, which we recall is defined as follows. 
Let $\set{\ket{i}}_{i=1}^d$ be the standard basis of $\C^d$, and define $E\in\mathbb{M}_{d^2}$ to be the matrix $E = \sum_{i, j=1}^d E_{i, j}\otimes E_{i, j}$ where $E_{i, j} := \ketbra{i}{j}\in\matrices$.  
Then the Choi matrix $J(\psi)$ of $\psi$ is
\begin{equation*}
    J(\psi)
    =
    (\psi\otimes \operatorname{Id}_{\matrices})(E).
\end{equation*}
If $\psi$ is of the form $\psi_{\boldsymbol{V}}$ for some $\boldsymbol{V} = \seq{V_i}_{i=1}^g \in\gmatrices$, then, under the standard identification of $\matrices$ with $\C^{d}\otimes\C^d$ via
\begin{equation}
\begin{split}
    \operatorname{Vec}:\matrices &\to\C^d\otimes\C^d\\
    \operatorname{Vec}\seq{\sum_{i, j=1}^da_{ij}E_{i, j}}
        &= 
    \sum_{i, j=1}^da_{ij}\ket{i}\otimes\ket{j},
\end{split}
\end{equation}
it is a standard fact that $J(\psi)$ may be written $J(\psi) = \sum_{k=1}^K\operatorname{Vec}(V_k)\operatorname{Vec}(V_k)^*$ \cite{Watrous2018TheInformation}. 
Under this identification, it is clear that $\operatorname{rank}(J(\psi)) = \dim\operatorname{span}\set{V_k}_{k=1}^K$. 
In particular, we have that 
\begin{equation}
    \prelength{\boldsymbol{S}_\omega}
        =
    \inf\set{
    n\in\N
        \,\,:\,\,
    \operatorname{rank}\seq{J\seq{\Phi_{\boldsymbol{S}; \omega}^{(n)}}}
        =
    d^2
    },
\end{equation}
and an analogous identity holds for $\wielength{\boldsymbol{S}}$.
So, since the Choi matrix of a completely positive map is positive semidefinite, we see that 
\begin{equation}
    \prelength{\boldsymbol{S}_\omega}
        =
    \inf\set{
    n\in\N
        \,\,:\,\,
    J\seq{\Phi_{\boldsymbol{S}; \omega}^{(n)}}>0
    },
\end{equation}
and similarly for $\wielength{\boldsymbol{S}}$.
\begin{lem}\label{App:Lem:Faithful_adjoint_implies_PSD_to_PSD}
    Let $\psi:\matrices\to\matrices$ be a positive linear map. 
    If $\psi^*$ is faithful, then $\psi(p)>0$ for all $p > 0$. 
\end{lem}
\begin{proof}
    Let $p>0$ and let $q\in\matrices$ be a projection. 
    Then there is $\varepsilon>0$ such that $\tr{q\psi(p)}\geq \varepsilon\tr{q\psi(\I)}$. 
    So, since $\psi^*$ is faithful, it holds that $\tr{q\psi(p)}>0$ for all projections $q$, which implies $\psi(p)>0$. 
\end{proof}
\begin{lem}\label{App:Lem:Pre_is_wie}
    If $\boldsymbol{S}^*$ is faithful, $\prelength{\boldsymbol{S}} = \wielength{\boldsymbol{S}}$ almost surely. 
\end{lem}
\begin{proof}
    It is clear that $\prelength{\boldsymbol{S}}\leq\wielength{\boldsymbol{S}}$ almost surely, so it suffices to show that on the event $\prelength{\boldsymbol{S}}<\infty$, it holds that $\prelength{\boldsymbol{S}} = \wielength{\boldsymbol{S}}$. 
    So, assume $\prelength{\boldsymbol{S}_\omega}<\infty$. 
    Then let $m\in\N$ and let $N = \prelength{\boldsymbol{S}_\omega} + m$. 
    Then 
    \begin{equation*}
        J\left(\Phi^{(N)}_{\boldsymbol{S}; \omega}\right)
        =
        \Big[
            \Phi^{(m)}_{\boldsymbol{S}; T^{N - m}(\omega)}
            \otimes 
            \operatorname{Id}_{\matrices}
        \Big]
        \left(
        J\left(
            \Phi^{(\prelength{\boldsymbol{S}_\omega})}_{\boldsymbol{S}; \omega}\right)
        \right).
    \end{equation*}
    Now, $ J\left(
            \Phi^{(\prelength{\boldsymbol{S}_\omega})}_{\boldsymbol{S}; \omega}\right)>0$. 
    So, since $\boldsymbol{S}^*$ is faithful, it follows from Lemma \ref{App:Lem:Faithful_adjoint_implies_PSD_to_PSD} that $J\left(\Phi^{(N)}_{\boldsymbol{S}; \omega}\right) > 0$. 
    Since $m$ was arbitrary, this shows that $\prelength{\boldsymbol{S}_\omega} = \wielength{\boldsymbol{S}_\omega}$, as desired. 
\end{proof}
\begin{lem}
    If $\boldsymbol{S}^*$ is faithful, $q_0(\boldsymbol{S}) = q(\boldsymbol{S})$ almost surely. 
\end{lem}
\begin{proof}
    This follows immediately from \cite[Corollary 3.2]{MovassaghSchenker}.
\end{proof}
At this juncture, we notice that the condition $q(\boldsymbol{S})<\infty$ almost surely is precisely what we called eventual strict positivity for $\phi_{\boldsymbol{S}}$ in the main body above. 
In \cite{MovassaghSchenker}, the following characterization of this was given. 
\begin{lem}[\texorpdfstring{\cite[Lemma 2.1]{MovassaghSchenker}}{l}]
    $q(\boldsymbol{S})<\infty$ almost surely if and only if the following two conditions hold. 
    \begin{enumerate}[label = (\alph*)]
        \item $\boldsymbol{S}$ and $\boldsymbol{S}^*$ are faithful. 

        \item There exists $N\in\N$ such that $q_0(\boldsymbol{S}) = N$ with positive probability. 
    \end{enumerate}
\end{lem}
We now prove an analog version of this lemma for the condition $\wielength{\boldsymbol{S}}<\infty$ almost surely. 
\begin{lem}\label{App:Lem:Char_of_finite_inj}
    $\wielength{\boldsymbol{S}}<\infty$ almost surely if and only if the following two conditions hold. 
    \begin{enumerate}[label = (\alph*)]
        \item $\boldsymbol{S}$ and $\boldsymbol{S}^*$ are faithful. 

        \item There exists $N\in\N$ such that $\prelength{\boldsymbol{S}} = N$ with positive probability. 
    \end{enumerate}
\end{lem}
To prove this lemma, we proceed in steps. 
\begin{lemma}\label{App:Lem:Finite_wie_implies_faithfulness}
    If $\wielength{\boldsymbol{S}}<\infty$ almost surely, then $\boldsymbol{S}$ and $\boldsymbol{S}^*$ are faithful.
\end{lemma}
\begin{proof}
    for $\omega\in\Omega$ such that $\boldsymbol{S}$ is not faithful, there is a rank-1 projection $p_\omega\in\matrices$ such that $\Phi^{(n)}_{\boldsymbol{S}; \omega}(p) = 0$ for all $n\in\N$. 
    Writing $p_\omega = \ketbra{\xi}{\xi}$ for some $\ket{\xi}\in\C^d$, the positivity of $\Phi^{(n)}_{\boldsymbol{S}; \omega}$ implies $\Pi_{\sigma; \omega}^{(n)}\ket{\xi}$ for all $\sigma\in[g]^n$ and all $n$. 
    In particular, for all $n$, $p_\omega\not\in\operatorname{span}\set{\Pi_{\sigma; \omega}^{(n)}\,\,:\,\,\sigma\in[g]^n}$, hence $\prelength{\boldsymbol{S}_\omega} = \infty$. 
    This shows that $\wielength{\boldsymbol{S}} = \infty$ on the event that $\boldsymbol{S}$ is not faithful. 
    By a similar argument, one can show that on the event that $\boldsymbol{S}^*$ is not faithful,  $\wielength{\boldsymbol{S}^*} = \wielength{\boldsymbol{S}} = \infty$. 
    By contrapositive, this shows that if $\wielength{\boldsymbol{S}}<\infty$ almost surely, then $\boldsymbol{S}$ and $\boldsymbol{S}^*$ are faithful. 
\end{proof}
\begin{lemma}\label{App:Lem:If_faithful_then_01_law}
    If $\boldsymbol{S}$ and $\boldsymbol{S}^*$ are faithful, then $\mu[\wielength{\boldsymbol{S}} < \infty]\in\set{0, 1}$.
\end{lemma}
\begin{proof}
    Since $\boldsymbol{S}^*$ is faithful, Lemma \ref{App:Lem:Pre_is_wie} implies
    \begin{equation}
        \set{\wielength{\boldsymbol{S}}<\infty}
            =
        \bigcup_{n\in\N}
        \set{\prelength{\boldsymbol{S}} = n}.
    \end{equation}
    So, it suffices to show that if there is $N\in\N$ such that $\mu[\prelength{S} = N] > 0$, then $\mu[\wielength{\boldsymbol{S}} < \infty] = 1$. 
    Under the assumption such $N$ exists, we have 
    \begin{equation}
        \mu\left[
            \bigcup_{k\geq 0}
            \set{\omega\in\Omega\,\,:\,\,
            \prelength{\boldsymbol{S}_{T^{k}(\omega)}}
            =
            N
            }
        \right]
        =
        1
    \end{equation}
    by ergodicity of $T$. 
    Therefore, for almost every $\omega\in\Omega$, there is $k$ such that $J\seq{\Phi^{(N)}_{\boldsymbol{S}; T^k(\omega)}}>0$. 
    But by Lemma \ref{App:Lem:Wie_dont_care_about_adjoint}, 
    $J\seq{\Phi^{(N)}_{\boldsymbol{S}; T^k(\omega)}}>0$ if and only if $J\seq{\Phi^{(N)}_{\boldsymbol{S}^*; T^k(\omega)}} = J\seq{\Phi^{(N)*}_{\boldsymbol{S}; T^k(\omega)}}>0$.
    Thus, since $\boldsymbol{S}$ is faithful, arguing as in the proof of Lemma \ref{App:Lem:Pre_is_wie} we conclude
    \begin{equation}
        J\seq{\Phi^{(N + k)*}_{\boldsymbol{S};\omega}}
        =
        \left[\Phi^{(k)*}_{\boldsymbol{S}; \omega}\otimes\operatorname{Id}_{\matrices}\right]
        \!
        \seq{
        J\seq{\Phi^{(N)*}_{\boldsymbol{S}; T^k(\omega)}}
        }
        >0,
    \end{equation}
    which shows that $\prelength{\boldsymbol{S}_\omega}<\infty$. 
    But $\omega$ was an arbitrary element of a full probability set, so we conclude that $\mu[\prelength{\boldsymbol{S}}<\infty] = \mu[\wielength{\boldsymbol{S}}<\infty]= 1$.
\end{proof}
\begin{proof}[Proof of Lemma \ref{App:Lem:Char_of_finite_inj}]
    Assuming $\wielength{\boldsymbol{S}}<\infty$ almost surely, Lemma \ref{App:Lem:Finite_wie_implies_faithfulness} already shows that $\boldsymbol{S}$ and $\boldsymbol{S}^*$. 
    From $ \set{\wielength{\boldsymbol{S}}<\infty}
        =
        \bigcup_{n\in\N}
        \set{\wielength{\boldsymbol{S}} = n}$
    and the fact that $\set{\wielength{\boldsymbol{S}} = n} = \set{\prelength{\boldsymbol{S}} = n}$ via Lemma \ref{App:Lem:Pre_is_wie}, the fact that $\wielength{\boldsymbol{S}}<\infty$ almost surely already shows that there is $N$ for which $\mu[\prelength{\boldsymbol{S}} = N]>0$.
    Conversely, assuming (a) and (b), Lemma \ref{App:Lem:If_faithful_then_01_law} together with $\set{\wielength{\boldsymbol{S}}<\infty}
        =
        \bigcup_{n\in\N}
        \set{\wielength{\boldsymbol{S}} = n}$ and $\wielength{\boldsymbol{S}} = \prelength{\boldsymbol{S}}$ shows that $\mu[\wielength{\boldsymbol{S}}<\infty] = 1$, concluding the proof.  
\end{proof}
We conclude this appendix with the following proposition. 
\begin{prop}\label{App:Prop:ESP_iff_finite_inj}
    Assume $\phi_{\boldsymbol{S}}$ is almost surely trace preserving. 
    Then $q(\boldsymbol{S})<\infty$ almost surely if and only if $\wielength{\boldsymbol{S}}<\infty$ almost surely. 
\end{prop}
\begin{proof}
    We proceed as in \cite{SanzPerezGarciaWolfCirac}.
    First, we claim that $q(\boldsymbol{S})\leq \wielength{\boldsymbol{S}}$ almost surely.
    Indeed, on the event $\wielength{\boldsymbol{S}} = \infty$ there is nothing to prove, so assume $\wielength{\boldsymbol{S}_\omega} < \infty$. 
    For any $n\geq \wielength{\boldsymbol{S}_\omega}$, we have that $J\seq{\Phi^{(n)}_{\boldsymbol{S}; \omega}}>0$. 
    In particular, because for any rank one projection $p = \ketbra{\xi}{\xi}\in\matrices$ we have that 
    \begin{equation}
        \Phi^{(n)}_S(p)
            =
       \operatorname{Tr}_{\matrices}\seq{J\seq{\Phi^{(n)}_{\boldsymbol{S}; \omega}}\seq{\I\otimes p'}}
    \end{equation}
    where $\operatorname{Tr}_{\matrices}$ denotes the partial trace and where $p'$ is the transpose of $p$, which is itself a rank one projection, we conclude that $\Phi^{(n)}_{\boldsymbol{S}; \omega}(p)>0$.
    Since $p$ was arbitrary, this shows that $\Phi^{(n)}_{\boldsymbol{S}; \omega}$ is strictly positive for all $n\geq\wielength{\boldsymbol{S}}$. 
    Thus, $q(\boldsymbol{S}_\omega)\leq \wielength{\boldsymbol{S}_\omega}$, concluding the proof of the claim. 
    The claim implies the backwards direction of the proposition, and so it just remains to show the forwards direction. 
    To prove the forwards direction, we proceed by contradiction.
    So, assume that $q\seq{\boldsymbol{S}}<\infty$ almost surely and that $\mu[\wielength{\boldsymbol{S}}=\infty] > 0$. 
    By Lemma \ref{App:Lem:Char_of_finite_inj}, we know that $\boldsymbol{S}$ and $\boldsymbol{S}^*$ are faithful, therefore by Lemma \ref{App:Lem:If_faithful_then_01_law}, we know that $\mu[\wielength{\boldsymbol{S}}=\infty] = 1$.
    First, by \cite[Theorem 1]{MovassaghSchenker}, there exists $Z:\Omega\to\matrices$ such that $\tr{Z} = 1$, $Z>0$, and $\phi_{\boldsymbol{S}; \omega}(Z_\omega) = Z_{T(\omega)}$ almost surely. 
    Furthermore, by \cite[Theorem 2]{MovassaghSchenker}, there is a universal constant $\gamma\in (0, 1)$ and a measurable function $C:\Omega\to (0, \infty)$ such that 
    \begin{equation}\label{App:Prop:ESP_iff_finite_inj:Eqn_1}
        \left\|
            \Phi^{(n)}
            _{\boldsymbol{S}; \omega}
            -
            \Delta_{n; \omega}
        \right\|_1
        \leq 
        C_\omega \gamma^n
    \end{equation} 
    holds almost surely for all $n$, where $\Delta_{n;\omega}(a) := \tr{a}Z_{T^n(\omega)}$ for all $a\in\matrices$ and $\|\cdot\|_1$ denotes the trace norm $\operatorname{Tr}|\cdot|$. 
    We shall use these facts freely in the following. 
    Now, since $\wielength{\boldsymbol{S}} = \infty$ almost surely, for almost every $\omega\in\Omega$ and all $n\in\N$, there is an orthogonal projection $p_{n; \omega}\in\matrices$ such that $ \Pi^{(n)}_{\sigma; \omega}p_{n; \omega} = 0$ for all $\sigma\in[g]^n$. 
    Since $\tr{p_{n; \omega}}\geq 1$, we see that 
\begin{equation}
    \frac{1}{\|Z_{T^n(\omega)}^{-1}\|_\infty}
    \leq 
    \tr{p_{n; \omega} Z_{T^n(\omega)} p_{n; \omega}}
\end{equation}
holds for all $n$ and almost every $\omega\in\Omega$. 
Thus, for all $n\in\N$, we have that 
\begin{align}
  \frac{1}{\|Z_{T^n(\omega)}^{-1}\|_\infty}
    &\leq 
    \tr{p_{n; \omega} Z_{T^n(\omega)} p_{n; \omega}}\\
    &= 
    \left| 
        \tr{p_{n; \omega} Z_{T^n(\omega)} p_{n; \omega}}
        -
        \sum_{\sigma\in[g]^n}
        \left| 
            \tr{\Pi^{(n)}_{\sigma; \omega}p_{n; \omega}}
        \right|^2 
    \right|\\
    &= 
    \Big| 
    \operatorname{Tr}_{{d^2}}\!\seq{
        E \seq{\Delta_{n; \omega}\otimes \operatorname{Id}_{\matrices}}\seq{
        \tilde{p}_{n; \omega} E \tilde{p}_{n; \omega}
        }
        }
        -
        \operatorname{Tr}_{{d^2}}\!\seq{
        E\!\left[\Phi^{(n)}_{S; \omega}\otimes\operatorname{Id}_{\matrices}\right]\!\seq{
        \tilde{p}_{n; \omega} E \tilde{p}_{n; \omega}
        }
        }
    \Big|\\
    &\leq 
    D_{\omega}\gamma^n \label{Eqn:Wielength_finite_characterizations:Eqn2}
\end{align}
where $ \tilde{p}_{n; \omega} = p_{n; \omega}\otimes\I$ and $D_{\omega}\in (0, \infty)$ is some constant independent of $p_{n; \omega}$, which holds by (\ref{App:Prop:ESP_iff_finite_inj:Eqn_1}).
This, however, is a contradiction: because $Z>0$ almost surely, there is $\varepsilon>0$ such that $\|Z^{-1}\|_\infty^{-1}>\varepsilon$ with positive probability, so by Poincar\'e recurrence, we know that 
\begin{equation}
    \limsup_{n\to\infty}\frac{1}{\|Z_{T^n(\omega)}^{-1}\|_\infty}
    \geq \varepsilon
\end{equation}
holds with positive probability. 
This, however, contradicts (\ref{Eqn:Wielength_finite_characterizations:Eqn2}) holding almost surely, because $\gamma\in(0, 1)$. 
Therefore, it must be that $\wielength{\boldsymbol{S}} < \infty$ almost surely, which concludes the proof. 
\end{proof}


\section*{Conflict of Interest Statement}
On behalf of all authors, the corresponding author states there are no conflicts of interest. 

\section*{Data Availability Statement}
This manuscript has no associated data. 


\bibliographystyle{plain}
\bibliography{consolidated_refs}

\end{document}